\documentclass[english,nodate]{elsarticle}

\usepackage[utf8]{inputenc}
\usepackage[active]{srcltx}
\usepackage{array}
\usepackage{float}
\usepackage{multirow}
\usepackage{amsmath}
\usepackage{amssymb}
\PassOptionsToPackage{normalem}{ulem}
\usepackage{ulem}
\usepackage{natbib}
\usepackage{rotating}
\makeatletter

\newcommand*{\FULL}{}

\newcommand{\noun}[1]{\textsc{#1}}
\providecommand{\tabularnewline}{\\}
\floatstyle{ruled}
\newfloat{algorithm}{tbp}{loa}
\providecommand{\algorithmname}{Algorithm}
\floatname{algorithm}{\protect\algorithmname}

\usepackage[usenames,dvipsnames]{pstricks}
 \usepackage{epsfig}
 \usepackage{pst-grad} 
 \usepackage{pst-plot} 

\usepackage{amsfonts}
\usepackage{algorithm}
\usepackage{algorithmic}

\usepackage{babel}

\newtheorem{theorem}{Theorem}

\newtheorem{prop}[theorem]{Proposition}

\newtheorem{example}{Example}

\newcommand{\proof}{{\bf Proof.}\ }
\newcommand{\kommentar}[1]{}
\newcommand{\hs}{\hspace{-0.5cm}}

\def\real{\hbox{\rm\vrule\kern-1pt R}}
\def\nat{\hbox{\rm\vrule\kern-1pt N}}

\def\eps{\varepsilon}

\def\ed2v{ed^{2}}
\def\ed{ed}
\def\gvc{{H}}
\def\vvc{{V_H}}
\def\evc{{E_H}}
\def\gscg{{G}}
\def\vscg{{V}}
\def\escg{{E}}

\global\long\def\spg{{\sc Shortest Path Game }}

\global\long\def\scg{{\sc Shortest Connection Game}}
\global\long\def\geo{{\sc Geography}}
\global\long\def\vgeo{{\sc Vertex Geography }}
\global\long\def\pageo{{\sc Partizan Arc Geography }}
\global\long\def\pvgeo{{\sc Partizan Vertex Geography }}
\global\long\def\psp{{\sf PSPACE}-complete}

\global\long\def\np{{\sf NP}}

\def\opath{{\em spe-conn}}

\global\long\def\pathdec{{\sc Shortest Connection Game}} 

\makeatother

\begin{document}


\title{The Shortest Connection Game\tnoteref{titlenote}}
\tnotetext[titlenote]{
Andreas Darmann was supported by the Austrian Science Fund (FWF):
{[}P 23724-G11{]}.
Ulrich Pferschy and Joachim Schauer were supported by the Austrian
Science Fund (FWF): {[}P 23829-N13{]}.
}

\author{Andreas Darmann}
\address{Institute of Public Economics, University of Graz,\\
Universitaetsstr.~15, 8010 Graz, Austria,\\
\texttt{andreas.darmann@uni-graz.at}}

\author{Ulrich Pferschy}
\address{Department of Statistics and Operations Research, University of Graz,\\
Universitaetsstr.~15, 8010 Graz, Austria,\\
\texttt{pferschy@uni-graz.at}}

\author{Joachim Schauer}
\address{Department of Mathematics, University of Klagenfurt,\\
 Universitaetsstr.~65-67, 9020 Klagenfurt, Austria,\\
\texttt{joachim.schauer@aau.at}}
\date{}


\begin{abstract}
We introduce \scg, a two-player game played on a directed graph with edge costs. 
Given two designated vertices in which they start, the players take turns in choosing edges emanating from the vertex they are currently located at.
In this way, each of the players forms a path that origins from its respective starting vertex. 
The game ends as soon as the two paths meet, i.e., a connection between the players is established. Each player has to carry the cost of its chosen edges and thus aims at minimizing its own total cost. 

In this work we analyze the computational complexity of \scg. 
On the negative side, \scg\ turns out to be computationally hard even on restricted graph classes such as bipartite, acyclic and cactus graphs. 
On the positive side, we can give a polynomial time algorithm for cactus graphs when 
the game is restricted to simple paths. 
\end{abstract}

\begin{keyword}
shortest path problem,
game theory,
computational complexity,
cactus graph
\end{keyword}

\maketitle


\section{Introduction}
\label{sec:intro}

We consider the following game on a directed graph $G=(V,E)$ 
with vertex set $V$ and a set of directed edges $E$
with nonnegative edge costs $c(u,v)$ for each edge $(u,v)\in E$
and two designated vertices $s,t\in V$.

In \scg\  two players $A$ and $B$ start from their respective homebases 
($A$ in $s$ and $B$ in $t$). 
The aim of the game is to establish a connection
between $s$ and $t$ in the following sense: 
The players take turns in moving along an edge and thus 
each of them constructs a directed path.
The game ends as soon as one player reaches a vertex $m$, 
i.e.\ a meeting point, which was already visited by the other player. 
This means that at the end of the game one player, say $A$, has selected
a path from $s$ to $m$, while $B$ has selected a path from $t$
to $m$ and possibly further on to additional vertices (or vice versa).
Each player has to carry the cost of its chosen edges and wants to minimize 
the total costs it has to pay.

Note that it is not always beneficial for both players to move closer
to each other. Instead, one player may take advantage of cheap edges
and move away from the other player, who then has to bear the burden
of building the connection.

To avoid unnecessary technicalities we assume that for every graph
considered a solution of the game does exist. 
This could be checked in a preprocessing step e.g.\ by executing a breadth-first-search
 starting in parallel from $s$ and $t$. 

\smallskip
As a motivational example consider the decision problem of two persons $A$ and $B$ 
(let's say a married couple)
who want to decide on a joint holiday trip.
Each such trip consists of a number of components such as
transportation by car, train, or plane,
accommodation, sight-seeing trips, special events etc.
Thus, it requires a certain effort to organize such a trip
(inquiries, reservations, 
etc.)
The two persons have different ideas of such a trip but want to reach a consensus,
i.e.\ a trip they both agree with.
Their decision procedure runs as follows:

Person $A$ makes a first offer, i.e.\ a detailed plan of a trip
containing all its components. 
Naturally, putting together such a trip requires a certain effort (e.g.\ measured in time).
Then Person $B$ makes an offer corresponding to its own preferences.
Most likely the two offered trip plans will not match.
Thus, it is $A$'s turn again to modify its previous offer,
e.g.\ by changing an excursion, moving to a different hotel
or changing the mode of transportation. 
This change of plans again requires a certain effort. 
In this way, the decision process goes on with a sequence of offers
alternating between $A$ and $B$.
A final decision is reached as soon as one person proposes a trip
(= offer) which is identical to an offer proposed by the other person
in one of the previous rounds. 
It means that a trip was identified that both persons agree with.

The total cost incurred by the decision process for each person
consists of the sum of efforts spent by that person until the
final decision was reached.
Note that the effort to generate an offer depends on the sequence of offers
proposed before since the effort of changing from one offer to another
may differ a lot and thus different sequences of offers give rise to different total efforts.

Formally, we can represent each offer by a vertex $v\in V$ in a graph $G=(V,E)$.
Each edge $(u,v)\in E$ represent the change from the offer associated to $u$ to offer $v$
and incurs costs $c(u,v)$.
The two persons $A,B$ start with initial offers $s\in V$ and  $t\in V$ respectively.
The sequence of trips offered by each person corresponds to a simple path
in the graph
(it does not make sense in this setting to make the same offer twice).
The final decision is reached as soon as a vertex is offered that lies on both paths.
Note that this is not necessarily the last offer proposed by both persons.
Assuming total knowledge (an {\em old} married couple) each person will try to minimize
its own total cost taking into account the rational optimal decisions by the other person.


\medskip
We will impose two restrictions on the problem setting to ensure that 
the two players actually meet in some vertex and to guarantee finiteness of the game. 
To enable a feasible outcome of the game we restrict
the players in every decision to choose only edges which still permit
a meeting point of the two paths. 
\begin{quote}
\textbf{(R1)} The players can not select an edge which does not permit
the paths of the players to meet. 
\end{quote}
Note that this aspect may lead to interesting strategic behavior, since
we also require that each player must be able to choose an edge in every round. 
Thus, each player has to take care that the other player
does not get stuck.

\medskip
Secondly, we want to guarantee that the game remains finite.
This raises the question of how to handle cycles.
It could be argued that moving in a cycle makes sense for a player
who wants to avoid an edge of large costs.
However, for the sake of finiteness each cycle should be traversed 
at most once as guaranteed by the following restriction: 
\begin{quote}
\textbf{(R2)} Each edge may be used at most once. 
\end{quote}
Note that (R2) does not rule out the possibility to visit a vertex more than once.
However, it is also interesting to restrict the game to simple paths 
and thus exclude cycles completely by the following condition.
\begin{quote}
\textbf{(R3)} Each vertex may be used at most once by each player. 
\end{quote}
While (R1) and (R2) will be strictly enforced throughout the paper,
we will consider the variants with and without (R3).

\bigskip
We will study \scg\ on general directed graphs 
(also the case of an undirected graph is interesting, but this is beyond the scope of the current paper)
and on relevant special graph classes 
namely bipartite graphs, acyclic graphs, cactus graphs and trees.
Each of these special cases (with the trivial exception of acyclic graphs)
arises if the associated undirected graph implied by removing the
orientation from all edges has the stated property\footnote{This means that
two directed edges $(u,v)$ and $(v,u)$ imply two parallel edges $(u,v)$ in
the associated undirected graph.}.

\subsection{Game Theoretic Setting}
\label{sec:game}

The game described above can be represented as a finite game 
in extensive form. 
All feasible decisions for the players can be represented in a game tree,
where each node corresponds to the decision of a certain player in a
vertex of the graph $G$.
A similar representation was given in \cite{dps15a} for \spg (see also Section~\ref{sec:lit}).

The standard procedure to determine equilibria in a game tree is {\em backward induction}
(see \citet[ch.~5]{osb04}).
This means that for each node in the game tree 
whose child nodes are all leaves, the associated player can reach
a decision by choosing the best of all child nodes w.r.t.\ the
cost resulting from its corresponding path in $G$.
Then these leaf nodes can be deleted and the
costs accrued by each of the two players is moved to its parent node.
In this way, we can move upwards in the game tree towards the root and settle
all decisions along the way.

This backward induction procedure implies a strategy for each player
given by the decisions in the game tree:
{\em Always choose the edge according to the result of backward induction.}
This strategy for both players is a {\em Nash equilibrium} and also a so-called
{\em subgame perfect equilibrium} (a slightly stronger property),
since the decisions made in the underlying backward induction procedure
are also optimal for every subtree.

The outcome, if both players follow this strategy,
is a set of two directed paths, namely from $s$ to some vertex $v_A$ and 
from $t$ to some vertex $v_B$,  
where both paths contain the unique common vertex $m$ which coincides with at least one of the vertices $v_A, v_B$.
These paths make up the unique 
{\bf s}ubgame {\bf p}erfect {\bf e}quilibrium\footnote{
As a tie-breaking rule we will use the ``optimistic case'', where in case of indifference
a player chooses the option with lowest possible cost for the other player.
} and will be denoted by {\bf \opath}.
A \opath\ for \scg\ is the particular solution in the game tree
with minimal cost for both selfish players
under the assumption that they have complete and perfect information
of the game and know that the opponent will
also strive for its own selfish optimal value.
An illustration of \scg\ is given in the following example.

\begin{example}\label{ex:strat}
Consider the following graph where numbers refer to the actual costs of the corresponding edges and the vertices are labelled by letters.
In the associated game tree (depicted below)
\opath\ is determined by backward induction.
Thick arrows indicate the optimal decision in each node.
Each node of the game tree contains the current positions of both players
and the cost until the end of the game from the current position
as it is computed by backward induction; 
e.g.\ $(a,4|c,2)$ says that if $A$ is in $a$ and $B$ in $c$,
in the subgame perfect equilibrium the respective total cost for establishing a meeting point is $4$ for $A$ and $2$ for $B$.
The number depicted for each edge of the game tree gives the cost of the edge
in the graph associated with the corresponding decision.
In our example \opath\ is given by the sequence $s,d,b,e$ for $A$ (cost $4$), and $t,c,e$ (cost $5$) for $B$.
Note that in general the feasible moves from the current position might be restricted by the
paths chosen for reaching this position from the starting point $(s,t)$, 
e.g.\ because of (R2); in the given acyclic graph, however, this is not the case.
\end{example}

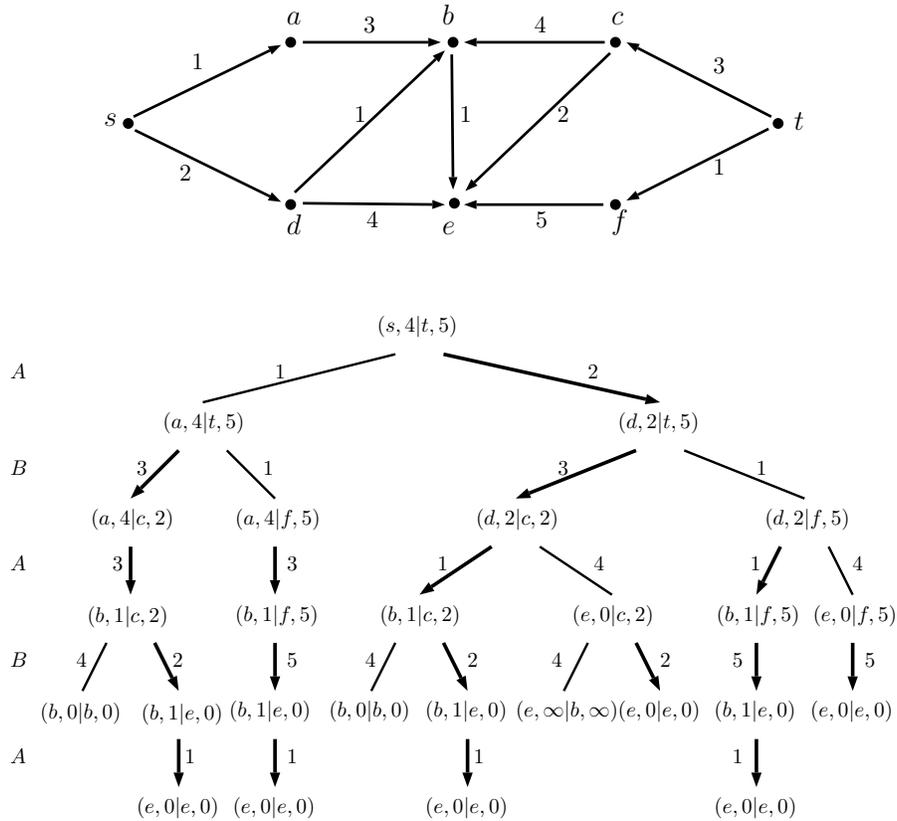
\begin{center}
\psscalebox{0.9 0.9} 
{
\begin{pspicture}(0,-1.66)(10.32,1.66)
\psdots[linecolor=black, dotsize=0.16](0.36,-0.06)
\psdots[linecolor=black, dotsize=0.16](2.76,1.14)
\psdots[linecolor=black, dotsize=0.16](2.76,-1.26)
\psdots[linecolor=black, dotsize=0.16](5.16,1.14)
\psline[linecolor=black, linewidth=0.04, arrowsize=0.05291666666666667cm 2.01,arrowlength=1.4,arrowinset=0.0]{->}(0.44,0.04)(2.62,1.1)
\psline[linecolor=black, linewidth=0.04, arrowsize=0.05291666666666667cm 2.01,arrowlength=1.4,arrowinset=0.0]{->}(0.46,-0.16)(2.62,-1.22)
\psline[linecolor=black, linewidth=0.04, arrowsize=0.05291666666666667cm 2.01,arrowlength=1.4,arrowinset=0.0]{->}(2.94,1.14)(4.98,1.14)
\psdots[linecolor=black, dotsize=0.16](5.18,-1.24)
\psdots[linecolor=black, dotsize=0.16](7.56,1.14)
\psdots[linecolor=black, dotsize=0.16](7.56,-1.26)
\psdots[linecolor=black, dotsize=0.16](9.96,-0.06)
\psline[linecolor=black, linewidth=0.04, arrowsize=0.05291666666666667cm 2.01,arrowlength=1.4,arrowinset=0.0]{->}(2.94,-1.24)(5.04,-1.26)
\psline[linecolor=black, linewidth=0.04, arrowsize=0.05291666666666667cm 2.01,arrowlength=1.4,arrowinset=0.0]{->}(7.4,1.14)(5.32,1.14)
\psline[linecolor=black, linewidth=0.04, arrowsize=0.05291666666666667cm 2.01,arrowlength=1.4,arrowinset=0.0]{->}(7.38,-1.26)(5.32,-1.26)
\psline[linecolor=black, linewidth=0.04, arrowsize=0.05291666666666667cm 2.01,arrowlength=1.4,arrowinset=0.0]{->}(9.86,0.04)(7.72,1.14)
\psline[linecolor=black, linewidth=0.04, arrowsize=0.05291666666666667cm 2.01,arrowlength=1.4,arrowinset=0.0]{->}(9.82,-0.14)(7.72,-1.2)
\psline[linecolor=black, linewidth=0.04, arrowsize=0.05291666666666667cm 2.01,arrowlength=1.4,arrowinset=0.0]{->}(7.48,0.98)(5.34,-1.06)
\psline[linecolor=black, linewidth=0.04, arrowsize=0.05291666666666667cm 2.01,arrowlength=1.4,arrowinset=0.0]{->}(5.14,0.96)(5.16,-1.04)
\psline[linecolor=black, linewidth=0.04, arrowsize=0.05291666666666667cm 2.01,arrowlength=1.4,arrowinset=0.0]{->}(2.82,-1.08)(5.04,1.02)
\rput[bl](0.0,-0.1){\large $s$}
\rput[bl](10.2,-0.16){\large $t$}
\rput[bl](2.7,1.4){\large $a$}
\rput[bl](5.0,1.4){\large $b$}
\rput[bl](7.5,1.4){\large $c$}
\rput[bl](2.7,-1.7){\large $d$}
\rput[bl](5.0,-1.7){\large $e$}
\rput[bl](7.5,-1.7){\large $f$}
\rput[bl](1.3,0.74){$1$}
\rput[bl](3.84,1.28){$3$}
\rput[bl](6.36,1.28){$4$}
\rput[bl](9.0,0.68){$3$}
\rput[bl](3.7,-0.04){$1$}
\rput[bl](5.26,-0.04){$1$}
\rput[bl](6.7,-0.04){$2$}
\rput[bl](1.12,-0.9){$2$}
\rput[bl](3.88,-1.6){$4$}
\rput[bl](6.38,-1.6){$5$}
\rput[bl](9.0,-0.82){$1$}
\end{pspicture}
}
\end{center}

\vspace*{3mm}
\begin{center}
\psscalebox{0.8 0.8} 
{
\begin{pspicture}(0,-4.175)(14.72,4.175)
\rput[bl](0.5,-2.595){$(b,0|b,0)$}
\rput[bl](3.64,-2.575){$(b,1|e,0)$}
\rput[bl](1.28,-0.995){$(b,1|c,2)$}
\rput[bl](2.1,-4.175){$(e,0|e,0)$}
\rput[bl](2.18,-2.615){$(b,1|e,0)$}
\rput[bl](3.74,-0.975){$(b,1|f,5)$}
\rput[bl](3.74,0.625){$(a,4|f,5)$}
\rput[bl](2.54,2.225){$(a,4|t,5)$}
\rput[bl](1.34,0.625){$(a,4|c,2)$}
\rput[bl](7.74,0.625){$(d,2|c,2)$}
\rput[bl](6.14,-0.975){$(b,1|c,2)$}
\rput[bl](9.34,-0.975){$(e,0|c,2)$}
\rput[bl](11.74,-0.975){$(b,1|f,5)$}
\rput[bl](13.34,-0.975){$(e,0|f,5)$}
\rput[bl](12.54,0.625){$(d,2|f,5)$}
\rput[bl](5.3,-2.575){$(b,0|b,0)$}
\rput[bl](6.1,3.825){$(s,4|t,5)$}
\rput[bl](10.1,2.225){$(d,2|t,5)$}
\rput[bl](13.3,-2.575){$(e,0|e,0)$}
\rput[bl](11.7,-2.575){$(b,1|e,0)$}
\rput[bl](10.1,-2.575){$(e,0|e,0)$}
\rput[bl](8.4,-2.575){$(e,\infty|b,\infty)$}
\rput[bl](6.9,-2.575){$(b,1|e,0)$}
\rput[bl](3.7,-4.175){$(e,0|e,0)$}
\rput[bl](6.9,-4.175){$(e,0|e,0)$}
\rput[bl](11.7,-4.175){$(e,0|e,0)$}
\rput[bl](0.0,3.145){$A$}
\rput[bl](4.4,3.145){1}
\rput[bl](9.6,3.145){2}
\rput[bl](0.0,1.545){$B$}
\rput[bl](2.1,1.545){3}
\rput[bl](4.2,1.545){1}
\rput[bl](9.1,1.545){3}
\rput[bl](12.4,1.545){1}
\rput[bl](0.0,-0.055){$A$}
\rput[bl](1.7,-0.055){3}
\rput[bl](4.6,-0.055){3}
\rput[bl](7.1,-0.055){1}
\rput[bl](9.7,-0.055){4}
\rput[bl](12.3,-0.055){1}
\rput[bl](14.0,-0.055){4}
\rput[bl](0.0,-1.655){$B$}
\rput[bl](1.1,-1.655){4}
\rput[bl](2.7,-1.655){2}
\rput[bl](4.6,-1.655){5}
\rput[bl](5.9,-1.655){4}
\rput[bl](7.6,-1.655){2}
\rput[bl](9.0,-1.655){4}
\rput[bl](10.8,-1.655){2}
\rput[bl](12.0,-1.655){5}
\rput[bl](14.2,-1.655){5}
\rput[bl](0.0,-3.255){$A$}
\rput[bl](2.9,-3.255){1}
\rput[bl](4.6,-3.255){1}
\rput[bl](7.7,-3.255){1}
\rput[bl](12.0,-3.255){1}
\psline[linecolor=black, linewidth=0.04](6.4,3.545)(3.2,2.745)
\psline[linecolor=black, linewidth=0.04](3.6,1.945)(4.4,1.145)
\psline[linecolor=black, linewidth=0.04](11.2,1.945)(13.2,1.145)
\psline[linecolor=black, linewidth=0.06, arrowsize=0.05291666666666667cm 2.0,arrowlength=1.4,arrowinset=0.0]{->}(2.0,0.345)(2.0,-0.455)
\psline[linecolor=black, linewidth=0.06, arrowsize=0.05291666666666667cm 2.0,arrowlength=1.4,arrowinset=0.0]{->}(4.4,0.345)(4.4,-0.455)
\psline[linecolor=black, linewidth=0.04](8.8,0.345)(10.0,-0.455)
\psline[linecolor=black, linewidth=0.04](13.6,0.345)(14.0,-0.455)
\psline[linecolor=black, linewidth=0.04](1.6,-1.255)(1.2,-2.055)
\psline[linecolor=black, linewidth=0.06, arrowsize=0.05291666666666667cm 2.0,arrowlength=1.4,arrowinset=0.0]{->}(4.4,-1.255)(4.4,-2.055)
\psline[linecolor=black, linewidth=0.04](6.4,-1.255)(6.0,-2.055)
\psline[linecolor=black, linewidth=0.04](9.6,-1.255)(9.2,-2.055)
\psline[linecolor=black, linewidth=0.06, arrowsize=0.05291666666666667cm 2.0,arrowlength=1.4,arrowinset=0.0]{->}(12.4,-1.255)(12.4,-2.055)
\psline[linecolor=black, linewidth=0.06, arrowsize=0.05291666666666667cm 2.0,arrowlength=1.4,arrowinset=0.0]{->}(14.0,-1.255)(14.0,-2.055)
\psline[linecolor=black, linewidth=0.06, arrowsize=0.05291666666666667cm 2.0,arrowlength=1.4,arrowinset=0.0]{->}(12.4,-2.855)(12.4,-3.655)
\psline[linecolor=black, linewidth=0.06, arrowsize=0.05291666666666667cm 2.0,arrowlength=1.4,arrowinset=0.0]{->}(7.6,-2.855)(7.6,-3.655)
\psline[linecolor=black, linewidth=0.06, arrowsize=0.05291666666666667cm 2.0,arrowlength=1.4,arrowinset=0.0]{->}(4.4,-2.855)(4.4,-3.655)
\psline[linecolor=black, linewidth=0.06, arrowsize=0.05291666666666667cm 2.0,arrowlength=1.4,arrowinset=0.0]{->}(2.8,-2.855)(2.8,-3.655)
\psline[linecolor=black, linewidth=0.06, arrowsize=0.05291666666666667cm 2.0,arrowlength=1.4,arrowinset=0.0]{->}(2.4,-1.255)(2.8,-2.055)
\psline[linecolor=black, linewidth=0.06, arrowsize=0.05291666666666667cm 2.0,arrowlength=1.4,arrowinset=0.0]{->}(2.8,1.945)(2.0,1.145)
\psline[linecolor=black, linewidth=0.06, arrowsize=0.05291666666666667cm 2.0,arrowlength=1.4,arrowinset=0.0]{->}(7.2,-1.255)(7.6,-2.055)
\psline[linecolor=black, linewidth=0.06, arrowsize=0.05291666666666667cm 2.0,arrowlength=1.4,arrowinset=0.0]{->}(10.4,-1.255)(10.8,-2.055)
\psline[linecolor=black, linewidth=0.06, arrowsize=0.05291666666666667cm 2.0,arrowlength=1.4,arrowinset=0.0]{->}(8.0,0.345)(6.8,-0.455)
\psline[linecolor=black, linewidth=0.06, arrowsize=0.05291666666666667cm 2.0,arrowlength=1.4,arrowinset=0.0]{->}(12.8,0.345)(12.4,-0.455)
\psline[linecolor=black, linewidth=0.06, arrowsize=0.05291666666666667cm 2.0,arrowlength=1.4,arrowinset=0.0]{->}(10.4,1.945)(8.4,1.145)
\psline[linecolor=black, linewidth=0.06, arrowsize=0.05291666666666667cm 2.0,arrowlength=1.4,arrowinset=0.0]{->}(7.2,3.545)(10.8,2.745)
\end{pspicture}

}
\end{center}

In this setting finding the \opath\ for the two players is not
an optimization problem as dealt with in combinatorial optimization
but rather the identification
of two sequences of decisions for the two players
fulfilling a certain property in the game tree.
Clearly, such a \opath\ solution can be computed in exponential time
by exploring the full game tree. 

It is the main goal of this paper to
study the complexity status of finding this \opath.
In particular, we want to establish the hardness of computation for general graphs
and identify special graph classes where a \opath\ is either still
hard to determine or can be found in polynomial time without exploring the exponential size game tree.

\medskip
Comparing the outcome of our game, i.e.\ the total cost of \opath,
with the cost of the shortest connection obtained by a cooperative strategy
we can consider the {\em Price of Anarchy} (cf.~\cite{nrt07})
defined by the ratio between these two values.
However, as illustrated by the following example it is easy to see that the Price of Anarchy can become arbitrarily large.  

\begin{example}\label{ex:poa}
In the graph depicted below there are only two feasible solutions with meeting points $v_1$ or $v_2$ and player $A$ can decide between the two.
The cheapest connection would use $v_1$ with total cost $3$ while
the selfish optimal solution of $A$ would use $v_2$ with total cost $1+M$,
i.e.\ the Price of Anarchy is $\frac{M+1}{3}$ which can be become arbitrarily large.
\end{example}

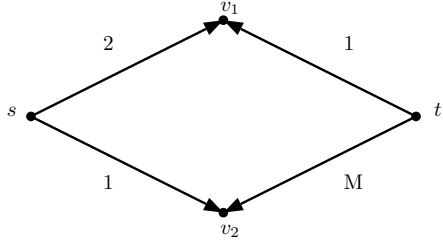
\begin{center}
\psscalebox{0.8 0.8} 
{
\begin{pspicture}(0,-1.9)(7.32,1.9)
\psdots[linecolor=black, dotsize=0.16](3.6,1.7)
\psdots[linecolor=black, dotsize=0.16](3.6,-1.5)
\psdots[linecolor=black, dotsize=0.16](6.8,0.1)
\psdots[linecolor=black, dotsize=0.16](0.4,0.1)
\psline[linecolor=black, linewidth=0.04, arrowsize=0.15cm 2.0,arrowlength=1.4,arrowinset=0.0]{->}(0.4,0.1)(3.6,1.7)
\psline[linecolor=black, linewidth=0.04, arrowsize=0.15cm 2.0,arrowlength=1.4,arrowinset=0.0]{->}(0.4,0.1)(3.6,-1.5)
\psline[linecolor=black, linewidth=0.04, arrowsize=0.15cm 2.0,arrowlength=1.4,arrowinset=0.0]{->}(6.8,0.1)(3.6,1.7)
\psline[linecolor=black, linewidth=0.04, arrowsize=0.15cm 2.0,arrowlength=1.4,arrowinset=0.0]{->}(6.8,0.1)(3.6,-1.5)
\rput[bl](0.0,0.1){$s$}
\rput[bl](7.1,0.1){$t$}
\rput[bl](1.6,1.2){2}
\rput[bl](5.6,1.2){1}
\rput[bl](1.6,-1.1){1}
\rput[bl](5.6,-1.1){M}
\rput[bl](3.55,1.8){$v_1$}
\rput[bl](3.55,-1.9){$v_2$}
\end{pspicture}
}
\end{center}

\medskip
An intuitive approach to the problem may suggest that each agent should
reach the meeting point via a {\em shortest path} from its origin
using a certain number of edges resulting from the iterative selection process. 
However, the following example shows that the restriction
to shortest paths from both origins to every possible meeting point
under any possible number of rounds may miss the optimal strategy.

\begin{example}
Consider the graph represented in the following figure.
All unlabelled edges have cost $0$. 
For sake of clarity we number only some of the vertices.
Edges with two arrows represent two edges in either direction.

$A$ starts the game by moving from $s$ to $v_1$.
If $B$ chooses the edge $(t, v_7)$, 
$A$ can move to $v_2$ and to $v_3$. 
This forces $B$ to move via $v_4$ to the meeting point $v_2$
implying cost $M$ for both $A$ and $B$.
However, $B$ can do better by starting with $(t, v_6)$.
If $A$ still moves to $v_2$, then $B$ enters the detour via $v_8$,
which forbids $A$ to take the path to $v_3$ but forces $A$
to go via $v_4$ and $v_5$ to the meeting point $v_6$
at a cost of $2M+4$, while $B$ has cost $6$.
A better option for $A$ would be the more expensive edge to $v_9$.
Now there is no need for $B$ to use the detour as a ``deterrent''.
Instead, $B$ can go via $v_5$ to the meeting point $v_4$ 
with cost $2M+2$ for $A$ and cost $4$ for $B$
which is the \opath\ of the game.

Note that neither $A$ nor $B$ can use their shortest paths 
with three edges from $s$ (resp.\ $t$) 
via $v_2$ (resp.\ $v_7$) to $v_4$.
\end{example}
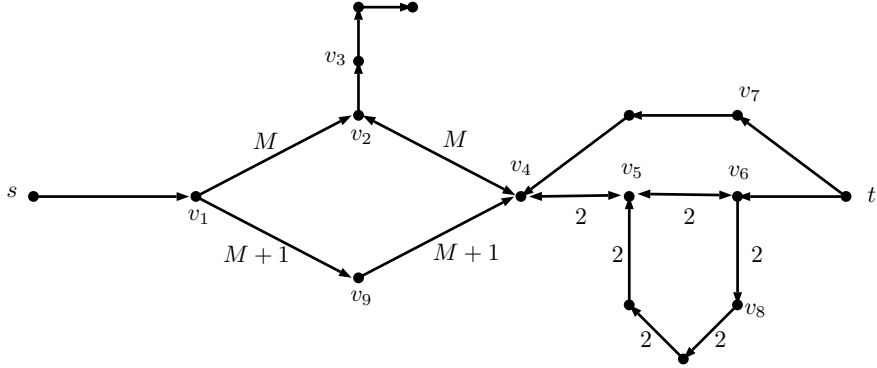
\begin{center}
\psscalebox{0.9 0.9} 
{
\begin{pspicture}(0,-2.6788464)(12.84,2.6788464)
\psdots[linecolor=black, dotsize=0.16](0.4,-0.19999985)
\psdots[linecolor=black, dotsize=0.16](2.8,-0.19999985)
\psdots[linecolor=black, dotsize=0.16](5.2,1.0000001)
\psdots[linecolor=black, dotsize=0.16](5.2,-1.3999999)
\psdots[linecolor=black, dotsize=0.16](7.6,-0.19999985)
\psdots[linecolor=black, dotsize=0.16](9.2,-0.19999985)
\psdots[linecolor=black, dotsize=0.16](10.8,-0.19999985)
\psdots[linecolor=black, dotsize=0.16](12.4,-0.19999985)
\psline[linecolor=black, linewidth=0.04, arrowsize=0.05291666666666667cm 2.0,arrowlength=1.4,arrowinset=0.0]{->}(0.4,-0.19999985)(2.7,-0.19999985)
\psline[linecolor=black, linewidth=0.04, arrowsize=0.05291666666666667cm 2.0,arrowlength=1.4,arrowinset=0.0]{->}(2.8,-0.19999985)(5.1,1.0000001)
\psline[linecolor=black, linewidth=0.04, arrowsize=0.05291666666666667cm 2.0,arrowlength=1.4,arrowinset=0.0]{->}(2.8,-0.19999985)(5.1,-1.3999999)
\psline[linecolor=black, linewidth=0.04, arrowsize=0.05291666666666667cm 2.0,arrowlength=1.4,arrowinset=0.0]{->}(5.2,-1.3999999)(7.5,-0.19999985)
\psdots[linecolor=black, dotsize=0.16](9.2,-1.7999998)
\psdots[linecolor=black, dotsize=0.16](10.8,-1.7999998)
\psline[linecolor=black, linewidth=0.04, arrowsize=0.05291666666666667cm 2.0,arrowlength=1.4,arrowinset=0.0]{<-}(10.8,-0.19999985)(12.4,-0.19999985)
\psline[linecolor=black, linewidth=0.04, arrowsize=0.05291666666666667cm 2.0,arrowlength=1.4,arrowinset=0.0]{<-}(10.8,-1.7999998)(10.8,-0.19999985)
\psdots[linecolor=black, dotsize=0.16](10.0,-2.6)
\psline[linecolor=black, linewidth=0.04, arrowsize=0.05291666666666667cm 2.0,arrowlength=1.4,arrowinset=0.0]{<-}(9.2,-1.7999998)(10.0,-2.6)
\psline[linecolor=black, linewidth=0.04, arrowsize=0.05291666666666667cm 2.0,arrowlength=1.4,arrowinset=0.0]{<-}(10.0,-2.6)(10.8,-1.7999998)
\psdots[linecolor=black, dotsize=0.16](9.2,1.0000001)
\psdots[linecolor=black, dotsize=0.16](10.8,1.0000001)
\psline[linecolor=black, linewidth=0.04, arrowsize=0.05291666666666667cm 2.0,arrowlength=1.4,arrowinset=0.0]{<-}(10.8,1.0000001)(12.4,-0.19999985)
\psline[linecolor=black, linewidth=0.04, arrowsize=0.05291666666666667cm 2.0,arrowlength=1.4,arrowinset=0.0]{<-}(9.2,1.0000001)(10.8,1.0000001)
\psline[linecolor=black, linewidth=0.04, arrowsize=0.05291666666666667cm 2.0,arrowlength=1.4,arrowinset=0.0]{<-}(7.6,-0.19999985)(9.2,1.0000001)
\psline[linecolor=black, linewidth=0.04, arrowsize=0.05291666666666667cm 2.0,arrowlength=1.4,arrowinset=0.0]{<-}(9.2,-0.19999985)(9.2,-1.7999998)
\psdots[linecolor=black, dotsize=0.16](5.2,1.8000002)
\psdots[linecolor=black, dotsize=0.16](5.2,2.6000001)
\psdots[linecolor=black, dotsize=0.16](6.0,2.6000001)
\psline[linecolor=black, linewidth=0.04, arrowsize=0.05291666666666667cm 2.0,arrowlength=1.4,arrowinset=0.0]{<-}(5.2,1.8000002)(5.2,1.0000001)
\psline[linecolor=black, linewidth=0.04, arrowsize=0.05291666666666667cm 2.0,arrowlength=1.4,arrowinset=0.0]{<-}(5.2,2.6000001)(5.2,1.8000002)
\psline[linecolor=black, linewidth=0.04, arrowsize=0.05291666666666667cm 2.0,arrowlength=1.4,arrowinset=0.0]{<-}(6.0,2.6000001)(5.2,2.6000001)
\rput[bl](0.0,-0.19999985){$s$}
\rput[bl](12.72,-0.27999985){$t$}
\rput[bl](2.68,-0.59999985){$v_1$}
\rput[bl](5.08,0.58000016){$v_2$}
\rput[bl](5.08,-1.8){$v_9$}
\rput[bl](4.7,1.7400001){$v_3$}
\rput[bl](7.44,0.120000154){$v_4$}
\rput[bl](9.08,0.080000155){$v_5$}
\rput[bl](10.66,0.06000015){$v_6$}
\rput[bl](10.82,1.2200001){$v_7$}
\rput[bl](10.9,-1.9799999){$v_8$}
\rput[bl](8.4,-0.59999985){$2$}
\rput[bl](10.0,-0.59999985){$2$}
\rput[bl](8.94,-1.1599998){$2$}
\rput[bl](9.36,-2.4199998){$2$}
\rput[bl](10.46,-2.4199998){$2$}
\rput[bl](11.0,-1.1599998){$2$}
\rput[bl](3.64,0.5000002){$M$}
\rput[bl](6.44,0.54000014){$M$}
\rput[bl](3.2,-1.1999998){$M+1$}
\rput[bl](6.3,-1.1599998){$M+1$}
\psline[linecolor=black, linewidth=0.04, arrowsize=0.05291666666666667cm 2.07,arrowlength=1.4,arrowinset=0.0]{<->}(5.26,1.0000001)(7.52,-0.17999984)
\psline[linecolor=black, linewidth=0.04, arrowsize=0.05291666666666667cm 2.07,arrowlength=1.4,arrowinset=0.0]{<->}(7.72,-0.19999985)(9.08,-0.17999984)
\psline[linecolor=black, linewidth=0.04, arrowsize=0.05291666666666667cm 2.07,arrowlength=1.4,arrowinset=0.0]{<->}(9.32,-0.15999985)(10.72,-0.17999984)
\end{pspicture}
}
\end{center}

\bigskip
In Section~\ref{sec:hard} we will show that this failure of the straightforward
idea can not be overcome at all since the underlying decision problem
is \psp\ on general graphs.
Note that by using depth-first search
the problem can be seen to be in $\sf{PSPACE}$ since the height
of the game tree is bounded by $2|E|$. 
In every node currently under consideration we keep a list of decisions 
still remaining to be explored 
and among all previously explored options starting in this node 
we keep the cost of the currently preferred subpath.
By proceeding in a depth-first manner there are at most $2|E|$ vertices on the path from
the root of the game tree to the current vertex for which the information is kept.

\subsection{Related Literature and Our Contribution}
\label{sec:lit}

Recently, \citet{dps15a} (see also \citet{dps14}) considered the closely related 
\spg where two players move together from a joint origin to a joint destination and take turns in selecting the next edge. 
The player choosing an edge has to pay its cost, and each of the players is aiming at 
minimizing its own total cost. 
Several complexity results are given for \spg including
\psp ness for directed bipartite graphs and a linear time algorithm for directed acyclic graphs. 

In contrast, the results of our work show that not only \scg\ is \psp\ in directed bipartite graphs, but also remains strongly $\np$-hard in the case of directed acyclic bipartite graphs, irrespective of whether or not the game is restricted to simple paths, i.e.\ (R3) is imposed (see Sections~\ref{sec:pspace} and \ref{sec:acyclic}). 
Searching for polynomially solvable special cases in Section~\ref{sec:cactus}, 
we move on from trees, where \scg\ is trivial to solve, 
to a slightly more general graph class,
namely cactus graphs, 
i.e.\ graphs where each edge is contained in at most one simple cycle.
On the positive side, we provide a polynomial time algorithm for the case of cactus graphs 
when (R3) is imposed (see Section~\ref{sec:cactussimple}), 
while we can show that \scg\ is already strongly $\np$-hard on cactus graphs
without condition (R3) (see Section~\ref{sec:cactushard}). 
Thus, we can describe a rather precise boundary between polynomially solvable 
and hard graph classes.
Note that the latter result is somewhat surprising 
given the fact that the related \spg is easy to solve on cactus graphs 
(see \citet{dps15a}). 

\medskip
Another related game is \geo\ (see \citet{sch78}), and in particular its variants 
\pageo and \pvgeo (see \citet{frsi93}). \pageo and \pvgeo both are two-player games played on a directed graph (there are no costs involved). 
In these games, starting from its designated homebase vertex, 
each of the two players forms a path by alternately moving along edges 
(in  \pageo only not-yet traversed edges can be taken, 
in \pvgeo only not-yet visited vertices can be moved to).
The first player unable to perform another move is the loser of the game. 
Note that the restrictions defining the two games share the spirit of (R2) and (R3). 
Both of these games are known to be \psp\  for directed bipartite graphs and $\np$-hard for directed acyclic graphs, while they are solvable in polynomial time for trees (see \citet{frsi93}). 
Recall that our problem \scg\ turns out to share this complexity behavior and, in particular, is also  polynomial time solvable for trees (see Proposition~\ref{th:tree}). 
However, not only do we enrich the analysis of \scg\ by considering cactus graphs, but we also  show that imposing or relaxing (R3) yields in fact different computational complexity results on this graph class.

In the classical game \geo, the two players 
move together along unused edges starting from the same designated starting vertex
and take turns in selecting the next edge. 
Again, the first player unable to choose an edge loses the game.
Besides the well-known \psp ness result of \citet{sch78} 
several complexity results have been provided for \geo\ and its variants such as  
\vgeo (each vertex can be visited at most once); see, e.g., \citet{lisi80}, \citet{frsi93}, \citet{fsu93}, and \citet{bo93}. In particular, both \geo\ and \vgeo are known to be solvable in polynomial time when played on cactus graphs (cf.~\citet{bo93}).

Another variant of two players taking turns in their decision on a discrete
optimization problem and each of them optimizing its own objective function
was considered for the Subset Sum problem in \citet{dnp13}.

\section{Hardness Results for \scg}
\label{sec:hard}

Let us first define formally the decision problem associated to \scg.
\begin{quote}
\pathdec

\textbf{Input:} A directed graph $G=(V,E)$ with nonnegative edge lengths $c_{e}$
for $e\in E$, two dedicated vertices $s,t$ in $V$, and nonnegative values $C_{A}, C_{B}$.

\textbf{Question:} Does \opath\ induce a path for agent $A$ 
with total cost $\leq C_{A}$ and a path for $B$ with total cost $\leq C_{B}$? 
\end{quote}

\subsection{\psp ness for directed bipartite graphs}
\label{sec:pspace}

Our proof is based on a reduction from \noun{Quantified $\leq4$-Sat}.
By \noun{Quantified $\leq k$-Sat,} we refer to the variation of \noun{Quantified
$3$-Sat} where each clause consists of at most  $k$ literals (instead
of exactly three). 

\begin {theorem}\noun{Quantified $\leq4$-Sat} is \psp~even when
each clause contains exactly one universal literal.\end{theorem}

\ifdefined\FULL
\textbf{Proof.} It can be concluded from~\cite{floegl} that \noun{Quantified
$\leq3$-Sat }is \psp~when restricted to quantified Boolean formulas
with at most one universal literal per clause. Let $F$ be such a
formula, i.e., 
\[
F=\exists x_{1}\forall x_{2}\ldots\exists x_{n-1}\forall x_{n}:\phi(x_{1},\ldots,x_{n})
\]
with $\phi(x_{1},\ldots,x_{n})=C_{1}\wedge C_{2}\ldots\wedge C_{m}$
such that $C_{j}$ contains at most three literals of which at most
one is universal. \\
Let $\gamma$ be the set of clauses among $C_{1},\ldots,C_{m}$ that
contain existential literals only. If $\gamma=\emptyset$, there is
nothing to show. Otherwise, w.l.o.g.\ let $\gamma=\{C_{1},C_{2},\ldots,C_{g}\}$,
for some $g\in\mathbb{N}$. For each $C_{j}\in\gamma$ introduce a
new (dummy) existential variable $e_{j}$ and a new universal variable
$u_{j}$, and define new clauses $C'_{j,1}:=C_{j}\vee u_{j}$ and
$C'_{j,2}:=C_{j}\vee\bar{u}_{j}$. Let 
\[
F'=\exists x_{1}\forall x_{2}\ldots\exists x_{n-1}\forall x_{n}\exists e_{1}\forall u_{1}\exists e_{2}\forall u_{2}\ldots\exists e_{g}\forall u_{g}:\phi'(x_{1},\ldots,x_{n})
\]
where $\phi'(x_{1},\ldots,x_{n})$ is derived from $\phi(x_{1},\ldots,x_{n})$
by replacing each $C_{j}\in\gamma$ by the clauses $C'_{j,1}$ and
$C'_{j,2}$. \\
Clearly, this transformation is polynomial, and $F$ is true if and
only if $F'$ is true. Note that each clause contains at most $4$
literals exactly one of which is universal.\qed 

\else{}
\textbf{Proof.} By reduction from  \noun{Quantified $\leq3$-Sat } 
restricted to quantified Boolean formulas with at most one universal literal per clause,
which is \psp~from~\cite{floegl}.
A detailed proof is given in~\cite{dpsconn}.\qed
\fi{}

\begin{theorem}\label{theo:pspace} \pathdec\ is \psp~for directed bipartite graphs
even if all costs are bounded by a constant.
\end{theorem}

\ifdefined\FULL

\proof Let $\mathcal{I}$ be an instance of \noun{Quantified $\leq4$-Sat}
with set $X=\{x_{1},\ldots,x_{n}\}$ of variables and set $C=\{C_{1},\ldots,C_{m}\}$
of clauses, with exactly one universal literal per clause. We reduce
$\mathcal{I}$ to an instance $\mathcal{M}$ of \pathdec\ as follows. 

\textbf{\emph{\noun{Part 1: Construction of  Graph $G$.}}} \\
First, we introduce a directed graph $G=(V,E)$, which again will
be $2$-colored in order to verify that $G$ is bipartite. An illustration
of $G$ is given in Fig.~\ref{fig:varIIdirected}, where, for the
sake of readability, the attention is restricted to the clause $C_{1}=(x_{1}\vee\bar{x}_{2}\vee x_{3})$.
To formally create $G$ from $\mathcal{I}$, we introduce 
\begin{itemize}
\item for each $i$, $1\leq i\leq n$, a ``hexagon'', i.e., 

\begin{itemize}
\item green vertices $v_{i,0},v_{i,2},v_{i,4}$ 
\item red vertices $v_{i,1},v_{i,3},v_{i,5}$
\item edges $x_{i}:=(v_{i,0},v_{i,1})$, $(v_{i,1},v_{i,2})$, and $(v_{i,2},v_{i,3})$
\item edges $\bar{x}_{i}:=(v_{i,0},v_{i,5})$, $(v_{i,5},v_{i,4})$ and
$(v_{i,4},v_{i,3})$
\end{itemize}
\item for each $i$, $1\leq i\leq n-1$, an edge $(v_{i,3},v_{i+1,0})$
\item green vertices $y,d_{B}$ and red vertex $d_{A}$
\item edges $(y,d_{A})$ and $(v_{n,3},d_{B})$
\item for each $j$, $1\leq j\leq m$

\begin{itemize}
\item green vertices $C_{j,A},q_{j,B},w_{j}$ and red vertices $C_{j,B},q_{j,A},z_{j}$
\item edges $(d_{A},C_{j,A})$, $(C_{j,A},q_{j,A})$, $(q_{j,A},w_{j})$,
$(w_{j},z_{j})$, $(z_{j},q_{j,B})$
\item edges $(d_{B},C_{j,B})$, $(C_{j,B},q_{j,B})$, $(q_{j,B},z_{j})$,
$(z_{j},w_{j})$
\item edges $(C_{j,A},v_{i,1})$ and $(v_{i,2},q_{j,A})$ if $i$ is odd
and $\bar{x}_{i}\in C_{j}$
\item edges $(C_{j,A},v_{i,5})$ and $(v_{i,4},q_{j,A})$ if $i$ is odd
and $x_{i}\in C_{j}$
\item edge $(q_{j,B},v_{i,5})$ if $i$ is even and $\bar{x}_{i}\in C_{j}$
\item edge $(q_{j,B},v_{i,1})$ if $i$ is even and $x_{i}\in C_{j}$
\end{itemize}
\end{itemize}
The edge costs are given by, $c(w_{j},z_{j})=c(z_{j},w_{j})=1$, $c(z_{j},q_{j,B})=2.1$
for $1\leq j\leq m$. \kommentar{For $1\leq i\leq n$ with $(v_{i,2},q_{j,A})\in E$
resp.\ $(v_{i,4},q_{j,A})\in E$, let $c(v_{i,2},q_{j,A})=c(v_{i,4},q_{j,A})=0.5$,
$1\leq j\leq m$. }Finally, $c(e)=\eps$ with $\strut$$\eps<\tfrac{1}{2n+6}$
for each of the remaining edges $e\in E$. \\
Note that $G$ is bipartite, since each edge links a green vertex
with a red vertex. In instance $\mathcal{M}$ of \scg, player
$A$ starts in vertex $s:=v_{1,0}$ and $B$ starts in vertex $t:=v_{2,0}$.
It is not hard to see that, by construction, in order to establish
a meeting point at least one of the players has to traverse an edge
of cost at $1$. In other words, $c(P)<1$ implies $c(Q)>1$ for $P,Q\in\{A,B\}$,
$P\not=Q$.

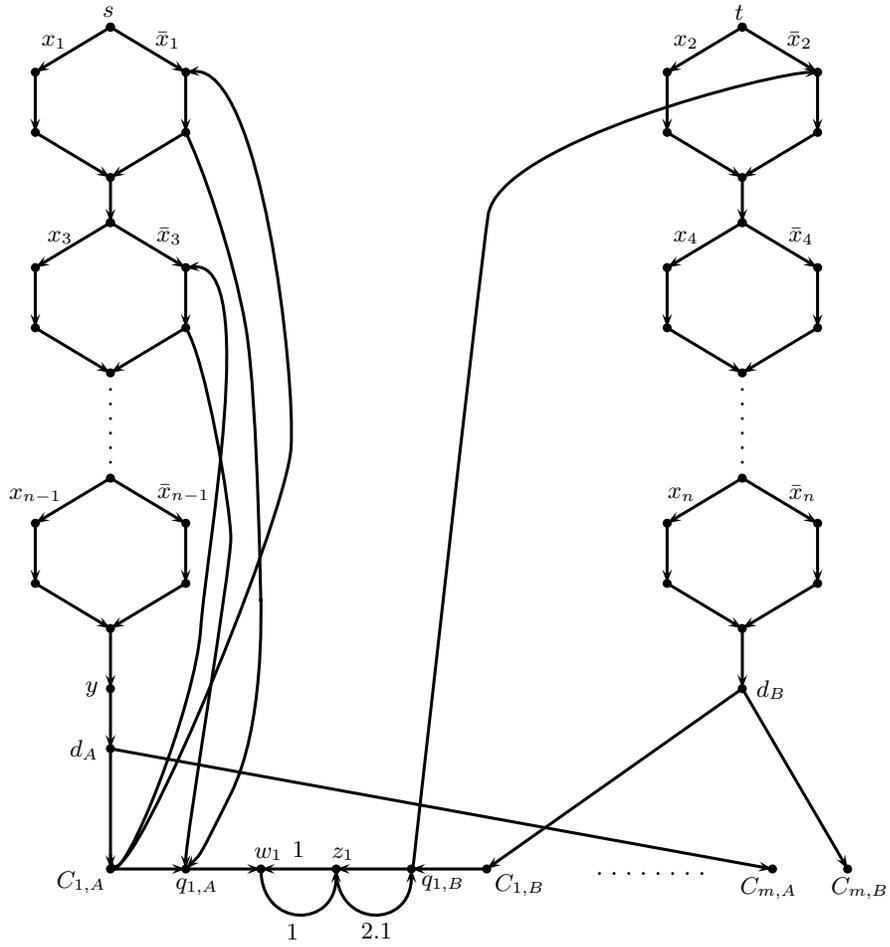
\begin{figure}
\begin{center}
\scalebox{1}
{

\begin{pspicture}(0,-6.309219)(13.28,6.3192186) \psline[linewidth=0.04cm,arrowsize=0.05291667cm 2.0,arrowlength=1.4,arrowinset=0.4]{->}(9.471875,5.320781)(9.471875,4.520781) \psdots[dotsize=0.12](10.471875,5.920781) \psdots[dotsize=0.12](10.471875,5.920781) \psline[linewidth=0.04cm,arrowsize=0.05291667cm 2.0,arrowlength=1.4,arrowinset=0.4]{->}(10.471875,5.920781)(9.471875,5.320781) \psline[linewidth=0.04cm,arrowsize=0.05291667cm 2.0,arrowlength=1.4,arrowinset=0.4]{<-}(10.471875,3.9207811)(11.471875,4.520781) \psline[linewidth=0.04cm,arrowsize=0.05291667cm 2.0,arrowlength=1.4,arrowinset=0.4]{<-}(11.471875,4.520781)(11.471875,5.320781) \psline[linewidth=0.04cm,arrowsize=0.05291667cm 2.0,arrowlength=1.4,arrowinset=0.4]{<-}(11.471875,5.320781)(10.471875,5.920781) \psdots[dotsize=0.12](11.471875,5.320781) \psdots[dotsize=0.12](11.471875,4.520781) \usefont{T1}{ppl}{m}{n} \rput(11.251875,3.1407812){\small $\bar{x}_4$} \psline[linewidth=0.04cm,arrowsize=0.05291667cm 2.0,arrowlength=1.4,arrowinset=0.4]{<-}(10.471875,3.9207811)(9.471875,4.520781) \psdots[dotsize=0.12](9.471875,5.320781) \psdots[dotsize=0.12](9.471875,4.520781) \psdots[dotsize=0.12](10.471875,3.9207811) \psline[linewidth=0.04cm,arrowsize=0.05291667cm 2.0,arrowlength=1.4,arrowinset=0.4]{->}(10.471875,3.9207811)(10.471875,3.3207812) \psline[linewidth=0.04cm,arrowsize=0.05291667cm 2.0,arrowlength=1.4,arrowinset=0.4]{->}(9.471875,2.7207813)(9.471875,1.9207813) \psdots[dotsize=0.12](10.471875,3.3207812) \psdots[dotsize=0.12](10.471875,3.3207812) \psline[linewidth=0.04cm,arrowsize=0.05291667cm 2.0,arrowlength=1.4,arrowinset=0.4]{->}(10.471875,3.3207812)(9.471875,2.7207813) \psline[linewidth=0.04cm,arrowsize=0.05291667cm 2.0,arrowlength=1.4,arrowinset=0.4]{<-}(10.471875,1.3207812)(11.471875,1.9207813) \psline[linewidth=0.04cm,arrowsize=0.05291667cm 2.0,arrowlength=1.4,arrowinset=0.4]{<-}(11.471875,1.9207813)(11.471875,2.7207813) \psline[linewidth=0.04cm,arrowsize=0.05291667cm 2.0,arrowlength=1.4,arrowinset=0.4]{<-}(11.471875,2.7207813)(10.471875,3.3207812) \psdots[dotsize=0.12](11.471875,2.7207813) \psdots[dotsize=0.12](11.471875,1.9207813) \psline[linewidth=0.04cm,arrowsize=0.05291667cm 2.0,arrowlength=1.4,arrowinset=0.4]{<-}(10.471875,1.3207812)(9.471875,1.9207813) \psdots[dotsize=0.12](9.471875,2.7207813) \psdots[dotsize=0.12](9.471875,1.9207813) \psdots[dotsize=0.12](10.471875,1.3207812) \psline[linewidth=0.04cm,linestyle=dotted,dotsep=0.16cm](10.471875,1.1207813)(10.471875,-0.07921875) \psline[linewidth=0.04cm,arrowsize=0.05291667cm 2.0,arrowlength=1.4,arrowinset=0.4]{->}(9.471875,-0.67921877)(9.471875,-1.4792187) \psdots[dotsize=0.12](10.471875,-0.07921875) \psline[linewidth=0.04cm,arrowsize=0.05291667cm 2.0,arrowlength=1.4,arrowinset=0.4]{->}(10.471875,-0.07921875)(9.471875,-0.67921877) \psline[linewidth=0.04cm,arrowsize=0.05291667cm 2.0,arrowlength=1.4,arrowinset=0.4]{<-}(10.471875,-2.0792189)(11.471875,-1.4792187) \psline[linewidth=0.04cm,arrowsize=0.05291667cm 2.0,arrowlength=1.4,arrowinset=0.4]{<-}(11.471875,-1.4792187)(11.471875,-0.67921877) \psline[linewidth=0.04cm,arrowsize=0.05291667cm 2.0,arrowlength=1.4,arrowinset=0.4]{<-}(11.471875,-0.67921877)(10.471875,-0.07921875) \psdots[dotsize=0.12](11.471875,-0.67921877) \psdots[dotsize=0.12](11.471875,-1.4792187) \psline[linewidth=0.04cm,arrowsize=0.05291667cm 2.0,arrowlength=1.4,arrowinset=0.4]{<-}(10.471875,-2.0792189)(9.471875,-1.4792187) \psdots[dotsize=0.12](9.471875,-0.67921877) \psdots[dotsize=0.12](9.471875,-1.4792187) \psdots[dotsize=0.12](10.471875,-2.0792189) \psline[linewidth=0.04cm,arrowsize=0.05291667cm 2.0,arrowlength=1.4,arrowinset=0.4]{->}(10.471875,-2.8792188)(7.071875,-5.2792187) \usefont{T1}{ppl}{m}{n} \rput(11.231875,5.7407813){\small $\bar{x}_2$} \usefont{T1}{ppl}{m}{n} \rput(11.271875,-0.27921876){\small $\bar{x}_{n}$} \usefont{T1}{ppl}{m}{n} \rput(9.661875,-0.29921874){\small ${x}_{n}$} \usefont{T1}{ppl}{m}{n} \rput(9.721875,3.1207812){\small ${x}_{4}$} \usefont{T1}{ppl}{m}{n} \rput(9.721875,5.7207813){\small ${x}_{2}$} \usefont{T1}{ppl}{m}{n} \rput(10.431875,6.1207814){\small $t$} 
\psline[linewidth=0.04cm,arrowsize=0.05291667cm 2.0,arrowlength=1.4,arrowinset=0.4]{->}(1.071875,5.320781)(1.071875,4.520781) \psdots[dotsize=0.12](2.071875,5.920781) \psdots[dotsize=0.12](2.071875,5.920781) \psline[linewidth=0.04cm,arrowsize=0.05291667cm 2.0,arrowlength=1.4,arrowinset=0.4]{->}(2.071875,5.920781)(1.071875,5.320781) \psline[linewidth=0.04cm,arrowsize=0.05291667cm 2.0,arrowlength=1.4,arrowinset=0.4]{<-}(2.071875,3.9207811)(3.071875,4.520781) \psline[linewidth=0.04cm,arrowsize=0.05291667cm 2.0,arrowlength=1.4,arrowinset=0.4]{<-}(3.071875,4.520781)(3.071875,5.320781) \psline[linewidth=0.04cm,arrowsize=0.05291667cm 2.0,arrowlength=1.4,arrowinset=0.4]{<-}(3.071875,5.320781)(2.071875,5.920781) \psdots[dotsize=0.12](3.071875,5.320781) \psdots[dotsize=0.12](3.071875,4.520781) \usefont{T1}{ppl}{m}{n} \rput(2.851875,3.1407812){\small $\bar{x}_3$} \psline[linewidth=0.04cm,arrowsize=0.05291667cm 2.0,arrowlength=1.4,arrowinset=0.4]{<-}(2.071875,3.9207811)(1.071875,4.520781) \psdots[dotsize=0.12](1.071875,5.320781) \psdots[dotsize=0.12](1.071875,4.520781) \psdots[dotsize=0.12](2.071875,3.9207811) \psline[linewidth=0.04cm,arrowsize=0.05291667cm 2.0,arrowlength=1.4,arrowinset=0.4]{->}(2.071875,3.9207811)(2.071875,3.3207812) \psline[linewidth=0.04cm,arrowsize=0.05291667cm 2.0,arrowlength=1.4,arrowinset=0.4]{->}(1.071875,2.7207813)(1.071875,1.9207813) \psdots[dotsize=0.12](2.071875,3.3207812) \psdots[dotsize=0.12](2.071875,3.3207812) \psline[linewidth=0.04cm,arrowsize=0.05291667cm 2.0,arrowlength=1.4,arrowinset=0.4]{->}(2.071875,3.3207812)(1.071875,2.7207813) \psline[linewidth=0.04cm,arrowsize=0.05291667cm 2.0,arrowlength=1.4,arrowinset=0.4]{<-}(2.071875,1.3207812)(3.071875,1.9207813) \psline[linewidth=0.04cm,arrowsize=0.05291667cm 2.0,arrowlength=1.4,arrowinset=0.4]{<-}(3.071875,1.9207813)(3.071875,2.7207813) \psline[linewidth=0.04cm,arrowsize=0.05291667cm 2.0,arrowlength=1.4,arrowinset=0.4]{<-}(3.071875,2.7207813)(2.071875,3.3207812) \psdots[dotsize=0.12](3.071875,2.7207813) \psdots[dotsize=0.12](3.071875,1.9207813) \psline[linewidth=0.04cm,arrowsize=0.05291667cm 2.0,arrowlength=1.4,arrowinset=0.4]{<-}(2.071875,1.3207812)(1.071875,1.9207813) \psdots[dotsize=0.12](1.071875,2.7207813) \psdots[dotsize=0.12](1.071875,1.9207813) \psdots[dotsize=0.12](2.071875,1.3207812) \psline[linewidth=0.04cm,linestyle=dotted,dotsep=0.16cm](2.071875,1.1207813)(2.071875,-0.07921875) \psline[linewidth=0.04cm,arrowsize=0.05291667cm 2.0,arrowlength=1.4,arrowinset=0.4]{->}(1.071875,-0.67921877)(1.071875,-1.4792187) \psdots[dotsize=0.12](2.071875,-0.07921875) \psline[linewidth=0.04cm,arrowsize=0.05291667cm 2.0,arrowlength=1.4,arrowinset=0.4]{->}(2.071875,-0.07921875)(1.071875,-0.67921877) \psline[linewidth=0.04cm,arrowsize=0.05291667cm 2.0,arrowlength=1.4,arrowinset=0.4]{<-}(2.071875,-2.0792189)(3.071875,-1.4792187) \psline[linewidth=0.04cm,arrowsize=0.05291667cm 2.0,arrowlength=1.4,arrowinset=0.4]{<-}(3.071875,-1.4792187)(3.071875,-0.67921877) \psline[linewidth=0.04cm,arrowsize=0.05291667cm 2.0,arrowlength=1.4,arrowinset=0.4]{<-}(3.071875,-0.67921877)(2.071875,-0.07921875) \psdots[dotsize=0.12](3.071875,-0.67921877) \psdots[dotsize=0.12](3.071875,-1.4792187) \psline[linewidth=0.04cm,arrowsize=0.05291667cm 2.0,arrowlength=1.4,arrowinset=0.4]{<-}(2.071875,-2.0792189)(1.071875,-1.4792187) \psdots[dotsize=0.12](1.071875,-0.67921877) \psdots[dotsize=0.12](1.071875,-1.4792187) \psdots[dotsize=0.12](2.071875,-2.0792189) \psline[linewidth=0.04cm,arrowsize=0.05291667cm 2.0,arrowlength=1.4,arrowinset=0.4]{->}(2.071875,-2.0792189)(2.071875,-2.8792188) \usefont{T1}{ppl}{m}{n} \rput(2.831875,5.7407813){\small $\bar{x}_1$} \usefont{T1}{ppl}{m}{n} \rput(3.041875,-0.27921876){\small $\bar{x}_{n-1}$} \usefont{T1}{ppl}{m}{n} \rput(1.681875,-5.459219){\small ${C}_{1,A}$} \usefont{T1}{ppl}{m}{n} \rput(1.401875,3.1207812){\small ${x}_{3}$} \usefont{T1}{ppl}{m}{n} \rput(1.321875,5.7207813){\small ${x}_{1}$} \usefont{T1}{ppl}{m}{n} \rput(2.041875,6.1207814){\small $s$} \psdots[dotsize=0.12](2.071875,-5.2792187) \psdots[dotsize=0.12](6.071875,-5.2792187) \psline[
linewidth=0.04cm,arrowsize=0.05291667cm 2.0,arrowlength=1.4,arrowinset=0.4]{->}(10.471875,-2.8792188)(11.871875,-5.2792187) \psline[linewidth=0.04cm,arrowsize=0.05291667cm 2.0,arrowlength=1.4,arrowinset=0.4]{->}(2.071875,-3.6792188)(2.071875,-5.2792187) \psline[linewidth=0.04cm,arrowsize=0.05291667cm 2.0,arrowlength=1.4,arrowinset=0.4]{->}(2.071875,-3.6792188)(10.871875,-5.2792187) \psline[linewidth=0.04cm,linestyle=dotted,dotsep=0.16cm](8.551875,-5.3392186)(9.951875,-5.3392186) \psdots[dotsize=0.12](10.871875,-5.2792187) \psdots[dotsize=0.12](11.871875,-5.2792187) \usefont{T1}{ppl}{m}{n} \rput(7.511875,-5.499219){\small ${C}_{1,B}$} \usefont{T1}{ppl}{m}{n} \rput(10.821875,-5.539219){\small ${C}_{m,A}$} \usefont{T1}{ppl}{m}{n} \rput(12.031875,-5.539219){\small ${C}_{m,B}$} \psdots[dotsize=0.12](3.071875,-5.2792187) \psdots[dotsize=0.12](4.071875,-5.2792187) \psline[linewidth=0.04cm,arrowsize=0.05291667cm 2.0,arrowlength=1.4,arrowinset=0.4]{->}(2.071875,-5.2792187)(3.071875,-5.2792187) \psline[linewidth=0.04cm,arrowsize=0.05291667cm 2.0,arrowlength=1.4,arrowinset=0.4]{->}(3.071875,-5.2792187)(4.071875,-5.2792187) \psline[linewidth=0.04cm,arrowsize=0.05291667cm 2.0,arrowlength=1.4,arrowinset=0.4]{->}(5.071875,-5.2792187)(4.071875,-5.2792187) \psdots[dotsize=0.12](5.071875,-5.2792187) \psline[linewidth=0.04cm,arrowsize=0.05291667cm 2.0,arrowlength=1.4,arrowinset=0.4]{->}(7.071875,-5.2792187)(6.071875,-5.2792187) \psbezier[linewidth=0.04,arrowsize=0.05291667cm 2.0,arrowlength=1.4,arrowinset=0.4]{->}(4.071875,-5.2792187)(4.071875,-6.079219)(5.071875,-6.079219)(5.071875,-5.2792187) \usefont{T1}{ppl}{m}{n} \rput(4.491875,-6.119219){\small $1$} \usefont{T1}{ppl}{m}{n} \rput(4.571875,-5.019219){\small $1$} \psline[linewidth=0.04cm,arrowsize=0.05291667cm 2.0,arrowlength=1.4,arrowinset=0.4]{->}(6.071875,-5.2792187)(5.071875,-5.2792187) \psdots[dotsize=0.12](7.071875,-5.2792187) \psdots[dotsize=0.12](2.071875,-3.6792188) \psline[linewidth=0.04cm,arrowsize=0.05291667cm 2.0,arrowlength=1.4,arrowinset=0.4]{->}(2.071875,-2.8792188)(2.071875,-3.6792188) \psdots[dotsize=0.12](2.071875,-2.8792188) \psbezier[linewidth=0.04,arrowsize=0.05291667cm 2.0,arrowlength=1.4,arrowinset=0.4]{->}(5.071875,-5.2792187)(5.071875,-6.079219)(6.071875,-6.079219)(6.071875,-5.2792187) \usefont{T1}{ppl}{m}{n} \rput(5.611875,-6.099219){\small $2.1$} \psline[linewidth=0.04cm,arrowsize=0.05291667cm 2.0,arrowlength=1.4,arrowinset=0.4]{->}(10.471875,-2.0792189)(10.471875,-2.8792188) \psdots[dotsize=0.12](10.471875,-2.8792188) \psbezier[linewidth=0.04,arrowsize=0.05291667cm 2.0,arrowlength=1.4,arrowinset=0.4]{->}(2.071875,-5.079219)(2.071875,-6.2792187)(4.405981,-0.6770454)(4.471875,0.32078126)(4.537769,1.3186079)(4.067875,5.4101357)(3.071875,5.320781) \psbezier[linewidth=0.04](3.071875,4.520781)(3.271875,4.1207814)(3.7042,2.9066236)(3.871875,1.9207813)(4.03955,0.9349389)(4.071875,-2.0792189)(4.071875,-1.6792188) \psbezier[linewidth=0.04,arrowsize=0.05291667cm 2.0,arrowlength=1.4,arrowinset=0.4]{->}(4.071875,-1.6792188)(4.071875,-1.8792187)(4.1162376,-3.3833718)(3.671875,-4.2792187)(3.2275121,-5.1750655)(3.271875,-5.079219)(3.071875,-5.2792187) \psbezier[linewidth=0.04,arrowsize=0.05291667cm 2.0,arrowlength=1.4,arrowinset=0.4]{->}(2.071875,-5.2792187)(2.271875,-5.479219)(3.2059813,-3.0770454)(3.271875,-2.0792189)(3.3377688,-1.081392)(4.067875,2.8101356)(3.071875,2.7207813) \psbezier[linewidth=0.04,arrowsize=0.05291667cm 2.0,arrowlength=1.4,arrowinset=0.4]{->}(3.071875,1.9207813)(3.3385417,1.5086601)(3.671875,-0.47921875)(3.671875,-0.87921876)(3.671875,-1.2792188)(3.071875,-4.8792186)(3.071875,-5.2792187) \usefont{T1}{ppl}{m}{n} \rput(1.721875,-3.6992188){\small ${d}_{A}$} \usefont{T1}{ppl}{m}{n} \rput(10.851875,-2.8992188){\small ${d}_{B}$} \usefont{T1}{ppl}{m}{n} \rput(3.221875,-5.519219){\small ${q}_{1,A}$} \usefont{T1}{ppl}{m}{n} \rput(6.491875,-5.479219){\small ${q}_{1,B}$} \psbezier[linewidth=0.04,arrowsize=0.05291667cm 2.0,arrowlength=1.4,arrowinset=0.4]{->}(6.091875,-5.2792187)(6.091875,-5.2792187)(6.9533734,2.4474173)(7.091875,3.4182172)(7.2303762,4.389017)(
10.891875,5.320781)(11.491875,5.320781) \usefont{T1}{ppl}{m}{n} \rput(1.071875,-0.31921875){\small ${x}_{n-1}$} \usefont{T1}{ppl}{m}{n} \rput(1.821875,-2.8792188){\small ${y}$} \usefont{T1}{ppl}{m}{n} \rput(4.171875,-5.079219){\small ${w}_{1}$} \usefont{T1}{ppl}{m}{n} \rput(5.161875,-5.079219){\small ${z}_{1}$} \end{pspicture}  }

\end{center}

\caption{\label{fig:varIIdirected}Directed graph $G$ in instance $\mathcal{M}$}

\end{figure}

\textbf{\emph{\noun{Part 2: The Reduction.}}}\\
 In what follows, we show that $\mathcal{I}$ is a ``yes'' instance
of \noun{Quantified $\leq4$-Sat} if and only if the outcome of instance
$\mathcal{M}$ of \pathdec\ yields $c(A)<1$. 

In instance $\mathcal{M}$, player $A$ starts the game with moving
an edge emanating from $s=v_{1,0}$. I.e., $A$ moves along one of
the edges $x_{1}$ and $\bar{x}_{1}$. In the next step, $B$ needs
to decide between moving along $x_{2}$ and $\bar{x}_{2}$. In the
next step ($A$ is either in $v_{1,1}$ or $v_{1,5}$, depending on
$A$'s choice in the first step), there is only one edge $A$ can
move along; after that move, $A$ is either in $v_{1,2}$ or $v_{1,4}$.
Similarly, $B$ has no choice but to follow the only available edge,
and ends up in either $v_{2,2}$ or $v_{2,4}$. In the following step
$A$ possibly has the choice between moving (i) along an edge to $q_{j,A}$
for some $1\leq j\leq m$ or (ii) along the edge to $v_{1,3}$.\\
(i) Assume that $A$ chooses the former. Clearly, $B$ cannot reach
a meeting point with $A$ in less than four moves, because $B$ would
have to traverse $d_{B}$ in order to do so. Hence, irrespective of
$B$'s decisions in the next three steps, $A$'s decision to move
to $q_{j,A}$ either causes $A$ to move along the edges $(q_{j,A},w_{j})$,
$(w_{j},z_{j}),(z_{j},q_{j,B})$, resulting in $c(A)=2\eps+1+2.1=3.1+2\eps$,
or is even infeasible because a meeting point is no longer possible.
\\
(ii) Assume that $A$ moves to $v_{1,3}$; $B$ necessarily moves
to $v_{2,3}$. In the next steps, first $A$ moves to $v_{3,0}$ and
$B$ to $v_{4,0}$. The above situation repeats: First, $A$ chooses
between $x_{3}$ and $\bar{x}_{3}$, then $B$ chooses between $x_{4}$
and $\bar{x}_{4}$. After each of the players makes another move,
$A$ again needs to decide between moving along an edge to $q_{j,A}$
for some $1\leq j\leq m$ or along the edge to $v_{3,3}$. Again,
the former move is either infeasible, or the game would end up with
$A$ having to carry at least the cost of $c(A)=3.1+6\eps$. This
leads us to the following simple observation. 

\textbf{Observation~1.} If, for some $1\leq j\leq m$, $A$ reaches
$q_{j,A}$ via a feasible move without traversing $y$, then the game
ends with $A$ having to carry a cost of at least $3.1+2\eps$.

We will argue that using the optimal strategy, $A$ will always traverse
$y$. Note that if the meeting point is not one of the vertices $w_{j}$,
then $A$ has to traverse at least one edge of cost at least one.
Together with the above observation, this means that the best possible
outcome for $A$ results from meeting in one of these vertices such
that $A$ uses edges of cost $\eps$ only. It is not hard to verify
that the number of edges $B$ needs to traverse to reach such a vertex
$w_{j}$ is $4\frac{n}{2}+4$. Hence, we get the following observation. 

\textbf{Observation~2.} In instance $\mathcal{M}$, we have $c(A)\geq\eps(2n+4)$.

Assume that $A$ reaches vertex $y$. By construction of $G$, this
means that $A$ has not moved along an edge with endpoint $q_{j,A}$
for some $1\leq j\leq m$ yet. In particular, $A$ has used exactly
$4\frac{n}{2}=2n$ edges. After $2n$ edges, $B$'s position is in
$d_{B}$ in any case. Let $\tau$ be the truth assignment that sets
to true exactly the set of edges (i.e., literals) of $X$ choses so
far. \\
After each player has chosen $2n$ edges, it is $A$'s turn, moving
along the only edge $(y,d_{A})$ emanating from $y$. Now, it is $B$'s
turn to choose the next vertex to be visited, i.e., to pick a vertex
$C_{j,B}$, $1\leq j\leq m$. We will show that $B$ can choose a
vertex $C_{j,B}$ with the result that the game ends with $c(B)<1$
(and, as a consequence, $c(A)>1$) if and only if there is a clause
$C_{j}$ which is not satisfied by the produced truth assignment $\tau$.
Putting things the other way round, we show that $\mathcal{I}$ is
a ``yes''-instance if and only if in instance $\mathcal{M}$ of
\pathdec ~the outcome is $c(A)\leq1$. For the sake of readability,
w.l.o.g.\ we assume that $B$ chooses vertex $C_{1,B}$. By construction
of $G$, this implies that there are exactly three possible meeting
points: $w_{1},z_{1}$ and $q_{1,B}$.\\
\emph{Case~1: $\mathcal{I}$ is a ``yes''-instance.}\noun{ }I.e.,
assume that $C_{1}$ is satisfied by $\tau$. \\
\uline{\noun{Case I:}} $C_{1}$ is satisfied by its unique universal
literal, i.e., $\bar{x}_{2}$. We show that for $A$, moving along
the edge $(d_{A},C_{1,A})$ results in the lowest possible total cost
for $A$. If $A$ does so, in the next step, $B$ necessarily moves
along $(C_{1,B},q_{1,B})$. Next, $A$ moves along $(C_{1,A},q_{1,A})$.
We distinguish the following cases (see also Table~\ref{tab:Case-I}).

\begin{table}
{\small }%
\begin{tabular}{|c|c||ccccc||c|}
\hline 
\multirow{2}{*}{\noun{\small Case}{\small{} Ia}} & {\small $A$} & {\small $C_{1,A}$} &  & {\small \hs$q_{1,A}$} &  & {\small \hs$w_{1}$} & \multirow{2}{*}{{\small infeasible}}\tabularnewline
\cline{2-7} 
 & {\small $B$} &  & {\small \hs$q_{1,B}$} &  & {\small \hs$v_{2,5}$} &  & \tabularnewline
\hline 
\end{tabular}{\small }\\
{\small \par}

{\small }%
\begin{tabular}{|c|c||cccccc||l|}
\hline 
\multirow{2}{*}{\noun{\small Case}{\small{} Ib}} & {\small $A$} & {\small $C_{1,A}$} &  & {\small \hs$q_{1,A}$} &  & {\small \hs$w_{1}$} &  & {\small $c(A)=\eps(2n+4)$}\tabularnewline
\cline{2-9} 
 & {\small $B$} &  & {\small \hs$q_{1,B}$} &  & {\small \hs$z_{1}$} &  & {\small \hs$z_{1}$} & {\small $c(B)=1+\eps(2n+3)$}\tabularnewline
\hline 
\end{tabular}{\small \par}

\caption{\label{tab:Case-I}Moves and resulting outcomes in subcase I of Case~1
in instance $\mathcal{M}$.}

\end{table}

\uline{\noun{Case I}}\uline{a}\uline{\noun{:}} $B$ moves
along $(q_{1,B},v_{2,5})$. Then, $A$ necessarily moves along $(q_{1,A},w_{1})$,
and $B$ is not able to move along an edge, because $B$ has used
the edge $\bar{x}_{2}$ already. Hence, in contradiction with (R1),
a meeting point is no longer possible, because the players need to
move alternately. 

\uline{\noun{Case I}}\uline{b}\uline{\noun{:}} $B$ needs
to move along $(q_{1,B},z_{1})$. $A$ moves along $(q_{1,A},w_{1})$.
The game ends with $B$ moving along $(z_{1},w_{1})$, because moving
along $(z_{1},q_{1,B})$ would be more expensive for $B$. The corresponding
costs are $c(A)=\eps\cdot(4\frac{n}{2}+4)=\eps(2n+4)$ (which is the
lowest possible cost for $A$, see Observation~2) and $c(B)=\eps(4\frac{n}{2}+3)+1=1+\eps(2n+3)$.

To summarize these cases, $B$ moves along $(q_{1,B},z_{1})$, i.e.,
\noun{C}ase Ib applies. Thus, in Case~I, $A$ has a strategy such
that the game ends with $c(A)=\eps(2n+4)<1$.

\uline{\noun{Case II:}} $C_{1}$ is not satisfied by its universal
literal $\bar{x}_{2}$. I.e., $C_{1}$ is satisfied by at least one
of its existential literals $x_{1},x_{3}$. W.l.o.g.\ assume that
$x_{1}$ is set to true. \\
Now, $A$ has to choose between the edges $(d_{A},C_{1,A})$ and $(d_{A},C_{j,A})$
for some $j>1$. 

\begin{table}
{\small }%
\begin{tabular}{|c|c||ccccccccc||l|}
\hline 
\multirow{2}{*}{\noun{\small Case}{\small{} IIa(i)\hspace{0.65cm}}} & {\small $A$} & {\small $C_{1,A}$} &  & {\small \hs$q_{1,A}$} &  & {\small \hs$w_{1}$} &  & {\small \hs$z_{1}$} &  & {\small \hs$q_{1,B}$} & {\small $c(A)=3.1+\eps(2n+4)$}\tabularnewline
\cline{2-12} 
 & {\small $B$} &  & {\small \hs$q_{1,B}$} &  & {\small \hs$v_{2,5}$} &  & {\small \hs$v_{2,4}$} &  & {\small \hs$v_{2,3}$} &  & {\small $c(B)=\eps(2n+5)$}\tabularnewline
\hline 
\end{tabular}{\small }\\
{\small \par}

{\small }%
\begin{tabular}{|c|c||cccccc||l|}
\hline 
\multirow{2}{*}{\noun{\small Case}{\small{} IIa(ii)({*})\hspace{0.15cm}}} & {\small $A$} & {\small $C_{1,A}$} &  & {\small \hs$v_{1,5}$} &  & {\small \hs$v_{1,4}$} &  & \multirow{2}{*}{{\small infeasible}}\tabularnewline
\cline{2-8} 
 & {\small $B$} &  & \hs{\small $q_{1,B}$} &  & \hs{\small $v_{2,5}$} &  & \hs{\small $v_{2,4}$} & \tabularnewline
\hline 
\end{tabular}{\small }\\
{\small \par}

{\small }%
\begin{tabular}{|c|c||ccccccccc||l|}
\hline 
\multirow{2}{*}{\noun{\small Case}{\small{} IIa(ii)({*}{*})}} & {\small $A$} & {\small $C_{1,A}$} &  & {\small \hs$v_{1,5}$} &  & \hs{\small $v_{1,4}$} &  & \hs{\small $q_{1,A}$} &  & \hs{\small $w_{1}$} & {\small $c(A)=\eps(2n+6)$}\tabularnewline
\cline{2-12} 
 & {\small $B$} &  & \hs{\small $q_{1,B}$} &  & \hs{\small $z_{1}$} &  & \hs{\small $w_{1}$} &  & \hs{\small $z_{1}$} &  & {\small $c(B)=2+\eps(2n+3)$}\tabularnewline
\hline 
\end{tabular}{\small \par}

\caption{\label{tab:Case-II}Moves and resulting outcomes in subcase IIa of
Case~1 in instance $\mathcal{M}$. }
\end{table}

\uline{\noun{Case II}}\uline{a}\uline{\noun{:}}\emph{ } $A$
moves along the edge $(d_{A},C_{1,A})$. In the next step, $B$ has
to move along $(C_{1,B},q_{1,B})$. $A$ may decide to (see also Table~\ref{tab:Case-II})

\begin{itemize}\item[(i)]move along $(C_{1,A},q_{1,A})$. Then,
$B$ moves along $(q_{1,B},v_{2,5})$, because otherwise the game
ends with $c(B)>1$. In the next step, $A$ has to move along $(q_{1,A},w_{1})$.
Since $C_{1}$ is not satisfied by $\bar{x}_{2}$, i.e., $\bar{x}_{2}$
has not been used yet, $B$ can move along $(v_{2,5},v_{2,4})$ in
the next step. The following moves must be $A$ moving along $(w_{1},z_{1})$
and $B$ along $(v_{2,4},v_{2,3})$. In the final step, $A$ has to
end the game using the edge $(z_{1},q_{1,B})$, because otherwise
a meeting point is no longer possible. Thus, the game ends with $c(A)>3.1$
and $c(B)=\eps(2n+5)$.

\item[(ii)]move along $(C_{1,A},v_{1,5})$. $B$ has the choice between
moving along ({*}) $(q_{1,B},v_{2,5})$ or ({*}{*}) $(q_{1,B},z_{1})$.

\begin{itemize}\item[(*)]Assume $B$ chooses the former. In the
following steps, $A$ has to move along $(v_{1,5},v_{1,4})$ and $B$
along $(v_{2,5},v_{2,4})$. Now, if $A$ moves along $(v_{1,4},v_{1,3})$,
$B$ has to move along $(v_{2,4},q_{2,3})$, and $A$ cannot perform
another move, i.e., a meeting point is no longer possible. If $A$
chooses $(v_{1,4},q_{1,A})$ instead of $(v_{1,4},v_{1,3})$, again
$B$ has to move along $(v_{2,4},q_{2,3})$, leaving $A$ with the
only possible move along $(q_{1,A},w_{1})$. Now, $B$ is unable to
move, and a meeting point is no longer possible. In either case, we
hence get a contradiction with (R1). 

\item[(**)]Hence, $B$ moves along $(q_{1,B},z_{1})$, followed by
$A$ moving along $(v_{1,5},v_{1,4})$. If $B$ moves along $(z_{1},q_{1,B})$,
the cost of $B$ increases by $c(z_{1},q_{1,B})=2.1$; i.e., in the
end, $B$ will have to carry a cost of $c(B)\geq2.1+\eps(4\frac{n}{2}+2)$.
If $B$ moves along $(z_{1},w_{1})$, $A$ has to move along $(v_{1,4},q_{1,A})$
to ensure a meeting point is still possible. The final two moves necessarily
are that $B$ moves along $(w_{1},z_{1})$, followed by $A$ moving
along $(q_{1,A},w_{1})$. The corresponding costs are $c(A)=\eps(2n+6)$
and $c(B)=\eps(2n+3)+2$.

\end{itemize}

\end{itemize}

As a consequence, in Case~\noun{II}a, player $A$ chooses to move
along $(C_{1,A},v_{1,5})$, which finally leads to an outcome with
$c(A)=\eps(2n+6)$ and $c(B)=\eps(2n+3)+2$. 

\uline{\noun{Case II}}\uline{b}\uline{\noun{:}} $A$ moves
along the edge $(d_{A},C_{j,A})$ for some $j>1$. By construction,
it follows that this admits a meeting point if and only if $C_{j,A}$
is made up of one the existential literals of $C_{1}$ that have already
been selected by $A$, i.e., set to true under $\tau$. Again, $B$
has no choice but to move along $(C_{1,B},q_{1,B})$. Now, it is easy
to see that $A$ moving along $(C_{j,A},q_{j,A})$ would not admit
a meeting point since from $q_{j,A}$ ($j>1$) none of the vertices
$w_{1},z_{1},q_{1,B}$ can be reached. The same is true if $A$ moves
towards the endpoint of an edge, i.e., literal, whose negation is
not contained in $C_{1}$. Thus, $A$ needs to move towards the endpoint
of a literal, whose negation is contained in $C_{1}$. Analogously
to \noun{Case II}a above, it follows that the game ends with $c(A)=\eps(2n+6)$
and $c(B)=\eps(2n+3)+2$ if $A$ moves along $(C_{j,A},v_{1,5})$. 

Summing up Case~I and Case~II it follows that, if $C_{1}$ is satisfied,
then \pathdec ends with $c(A)\leq\max\{\eps(2n+4),\eps(2n+6)\}<1$
and $c(B)\geq\min\{1+\eps(2n+3),2+\eps(2n+3)\}>1$.

\emph{Case~2: $\mathcal{I}$ is a ``no''-instance.}\noun{ }I.e.,
assume that $C_{1}$ is not satisfied by $\tau$. \\
Assume $A$ does not move along the edge $(d_{A},C_{1,A})$. Hence,
$A$ must move along an edge $(d_{A},C_{j,A})$, for some $j>1$.
Recall that $w_{1},z_{1},q_{1,B}$ are the only possible meeting points.
However, from $C_{j,A}$ it is impossible for $A$ to reach one of
these points, because otherwise it must be reached via the endpoint
of the negation of an existential literal that make up clause $C_{1}$,
i.e., $C_{1}$ must be satisfied by $\tau$. 

Thus, $A$ moves along $(d_{A},C_{1,A})$. As it turns out, the same
moves (and hence the same outcome) as in Case~IIa(i) of Case~1 results.
Next, $B$ moves along $(C_{1,B},q_{1,B})$. Since $C_{1}$ is not
satisfied, $A$ is not able to move along $(C_{1,A},v_{1,5})$ or
$(C_{1,A},v_{3,5})$, because otherwise a meeting point is no longer
possible. Thus, $A$ needs to move along $(C_{1,A},q_{1,A})$. Now,
assume $B$ moves along $(q_{1,B},v_{2,5})$, which is possible because
$\bar{x}_{2}$, and hence the edges $(v_{2,5},v_{2,4})$ and $(v_{2,4},v_{2,3})$
have not been used yet. In the next steps, $A$ has to move along
$(q_{1,A},w_{1})$, $B$ along $(v_{2,5},v_{2,4})$, $A$ along $(w_{1},z_{1})$,
and $B$ along $(v_{2,4},v_{2,3})$. Finally, $A$ has no choice but
to move along $(z_{1},q_{1,B})$ since otherwise a meeting point is
no longer possible, because $B$ cannot select any further edge because
$B$ is stuck in vertex $v_{2,3}$. Consequently, the game ends with
$c(A)=1+2.1+\eps(2n+4)=3.1+\eps(2n+4)$ and $c(B)=\eps(2n+5)$.

To summarize, instance $\mathcal{I}$ is a ``yes''-instance of \noun{Quantified
$\leq4$-Sat} if and only if instance $\mathcal{M}$ of \pathdec\ ends
with $c(A)<1$.~\qed

\else{}
\textbf{Proof.} 
The general outline of the proof follows to some extend the construction of the 
\psp ness proof for \geo\ given in \cite{sch78}.
The details of the reduction require careful treatment of auxiliary edges
and subtle case distinctions.
They can be found in an accompanying technical report~\cite{dpsconn}.\qed
\fi{}

Note that the above \psp ness result holds also if \scg\ is restricted to simple paths,
i.e., (R3) is imposed.

\begin{theorem}
 Imposing (R3), \scg\ is \psp\ for directed bipartite graphs.
\end{theorem}

\ifdefined\FULL
\proof The proof works analogously to the proof of Theorem~\ref{theo:pspace}, when the considered instance 
of \scg\ is slightly modified. In particular, we replace instance $\mathcal{M}$ with graph $G$ by instance $\mathcal{M}'$ with graph $G'$, where $G'$ is created from $G$ by,  
for each $1\leq j\leq m$, 
\begin{itemize}
 \item  removing the vertices $w_j , z_j$ and all edges to which $w_j$ or $z_j$ are incident
\item introducing the edges $(q_{j,A}, q_{j,B})$ and $(q_{j,B}, q_{j,A})$ of cost $1$ each.~\qed
\end{itemize}
 
\else{}
\proof The proof follows analogously to the proof of Theorem~\ref{theo:pspace}, when the considered instance of \scg\ is slightly modified. 
Details are given in~\cite{dpsconn}.\qed
\fi{}

\subsection{Directed Acyclic Bipartite Graphs}
\label{sec:acyclic}

In this section, we consider \scg\ in directed acyclic bipartite graphs. 
Note that due to the acyclicity, (R2) (and also (R3)) is obviously satisfied. 
It turns out that \pathdec\ remains intractable when restricted to
directed acyclic bipartite graphs. 
In particular, we provide an $\np$-hardness result by  
giving a reduction from \noun{Vertex Cover}.
\begin{quote}
\noun{Vertex Cover}

\textbf{Input:} Undirected graph $\gvc=(\vvc,\evc)$, integer $k\in\mathbb{N}$.

\textbf{Question:} Is there a vertex cover of size at most $k$, i.e.,
a subset $V' \subseteq \vvc$ with $|V'|\leq k$ such that for each edge
$\{i,j\}\in \evc$ at least one of $i,j$ belongs to $V'$? 
\end{quote}
\textbf{Remark.} Clearly, a vertex cover of size at most $k$ exists
if and only if there is a vertex cover of size exactly $k$.

\begin{theorem}\label{th:np-hardacyclic} \pathdec\ is strongly $\np$-hard
for directed acyclic bipartite graphs.\end{theorem}

\proof For the sake of readability, we provide a detailed proof for the case of 
directed acyclic graphs; finally, however, we will describe how the proof can 
easily be extended to directed acyclic \textit{bipartite} graphs.\newline
Given an instance $\mathcal{V}$ of \noun{Vertex Cover} with
an undirected graph $\gvc=(\vvc,\evc)$ and an integer $k\in\mathbb{N}$, we
construct an instance $\mathcal{W}$ of \pathdec\ with directed graph
$\gscg$ as follows. We assume $k\geq3$. Within this proof, let $n$
denote the number of vertices and $m$ the number of edges of $\gvc$,
i.e., $n:=|\vvc|$ and $m:=|\evc|$. W.l.o.g.\ let $\vvc=\{1,\ldots,n\}$
and $\evc=\{e_{1},\ldots,e_{m}\}$. W.l.o.g., we assume that $\{1,2\}\in \evc$.

\textbf{\emph{\noun{Part 1: Construction of Graph $\gscg$.}}} \\
We derive a directed graph $\gscg=(\vscg,\escg)$ as follows (see also
Fig.~\ref{fig:varIIacyclicNP}):
\begin{itemize}
\item we introduce vertices $s,t,d,d',f,g$ and $u_{1},\ldots,u_{k+1}$
\item edges $(d,d')$, $(g,d)$, $(g,f)$, $(t,u_{1})$, $(u_{k+1},d)$,
and $(u_{h},u_{h+1})$ for $1\leq h\leq k$
\item for each vertex $i\in \vvc$ introduce

\begin{itemize}
\item a vertex $s_{i}$ and the edges $(s,s_{i})$ and $(s_{i},g)$
\item edge $(s_{i},s_{j})$ for all $j>i$
\end{itemize}
\item for each edge $e=\{i,j\}\in \evc$ with $i<j$ we introduce 

\begin{itemize}
\item vertices $t_{i,j}$, $b_{i,j}$, $\ell_{i,j}$, $h_{i,j}$, $v_{i,j}$,
$w_{i,j}$, $y_{i,j}$, $z_{i,j}$
\item vertices $s_{i,j}$, $\alpha_{i,j}$, $\beta_{i,j}$, $\gamma_{i,j}$,
$\eta_{i,j}$, $\delta_{i,j}$, $\kappa_{i,j}$, $\lambda_{i,j}$,
$\mu_{i,j}$
\item edges $(u_{k+1},t_{i,j})$, $(t_{i,j},b_{i,j})$, $(b_{i,j},\ell_{i,j})$,
$(b_{i,j},h_{i,j})$
\item edges $(\ell_{i,j},v_{i,j})$, $(\ell_{i,j},w_{i,j})$, $(h_{i,j},y_{i,j})$,
$(h_{i,j},z_{i,j})$
\item edges $(v_{i,j},s_{i})$, $(y_{i,j},s_{j})$, $(f,s_{i,j})$
\item edges $(s_{i,j},\alpha_{i,j})$, $(\alpha_{i,j},\beta_{i,j})$, $(\beta_{i,j},\gamma_{i,j})$,
$(\gamma_{i,j},\eta_{i,j})$
\item edges $(s_{i,j},\delta_{i,j})$, $(\delta_{i,j},\kappa_{i,j})$, $(\kappa_{i,j},\lambda_{i,j})$,
$(\lambda_{i,j},\mu_{i,j})$
\item edges $(\alpha_{i,j},z_{i,j})$, $(\delta_{i,j},w_{i,j})$
\end{itemize}
\end{itemize}
\begin{figure}
\begin{center}
\scalebox{0.8}
{

\begin{pspicture}(0,-7.959219)(17.86,7.979219) \psdots[dotsize=0.12](8.971875,7.7607813) \psdots[dotsize=0.12](8.971875,7.1607814) \psdots[dotsize=0.12](8.971875,6.5607815) \psdots[dotsize=0.12](8.971875,4.960781) \psline[linewidth=0.04cm,linestyle=dotted,dotsep=0.16cm](8.971875,6.360781)(8.971875,5.1607814) \usefont{T1}{ppl}{m}{n} \rput(9.191875,7.7807813){\small $t$} \usefont{T1}{ppl}{m}{n} \rput(9.291875,6.540781){\small $u_2$} \usefont{T1}{ppl}{m}{n} \rput(9.301875,4.960781){\small $u_{k}$} \psline[linewidth=0.04cm,arrowsize=0.05291667cm 2.0,arrowlength=1.4,arrowinset=0.4]{->}(8.971875,7.7607813)(8.971875,7.1607814) \psline[linewidth=0.04cm,arrowsize=0.05291667cm 2.0,arrowlength=1.4,arrowinset=0.4]{->}(8.971875,7.1607814)(8.971875,6.5607815) \usefont{T1}{ppl}{m}{n} \rput(9.451875,4.500781){\small $u_{k+1}$} \psdots[dotsize=0.12](8.971875,4.360781) \usefont{T1}{ppl}{m}{n} \rput(9.251875,7.1407814){\small $u_1$} \psline[linewidth=0.04cm,arrowsize=0.05291667cm 2.0,arrowlength=1.4,arrowinset=0.4]{->}(8.971875,4.960781)(8.971875,4.360781) \psdots[dotsize=0.12](2.971875,2.3607812) \psdots[dotsize=0.12](4.971875,2.3607812) \psdots[dotsize=0.12](6.971875,2.3607812) \psdots[dotsize=0.12](8.971875,2.3607812) \psdots[dotsize=0.12](12.971875,2.3607812) \psdots[dotsize=0.12](17.171875,2.3607812) \usefont{T1}{ppl}{m}{n} \rput(17.421875,2.3807812){\small $d$} \usefont{T1}{ppl}{m}{n} \rput(2.611875,2.3607812){\small $t_{1,2}$} \usefont{T1}{ppl}{m}{n} \rput(4.591875,2.3807812){\small $t_{1,3}$} \usefont{T1}{ppl}{m}{n} \rput(6.611875,2.3607812){\small $t_{1,4}$} \usefont{T1}{ppl}{m}{n} \rput(8.611875,2.3607812){\small $t_{2,4}$} \usefont{T1}{ppl}{m}{n} \rput(13.421875,2.3407812){\small $t_{e_m}$} \psline[linewidth=0.04cm,arrowsize=0.05291667cm 2.0,arrowlength=1.4,arrowinset=0.4]{->}(8.971875,4.360781)(2.971875,2.3607812) \psline[linewidth=0.04cm,arrowsize=0.05291667cm 2.0,arrowlength=1.4,arrowinset=0.4]{->}(8.971875,4.360781)(4.971875,2.3607812) \psline[linewidth=0.04cm,arrowsize=0.05291667cm 2.0,arrowlength=1.4,arrowinset=0.4]{->}(8.971875,4.360781)(6.971875,2.3607812) \psline[linewidth=0.04cm,arrowsize=0.05291667cm 2.0,arrowlength=1.4,arrowinset=0.4]{->}(8.971875,4.360781)(8.971875,2.3607812) \psline[linewidth=0.04cm,arrowsize=0.05291667cm 2.0,arrowlength=1.4,arrowinset=0.4]{->}(8.971875,4.360781)(12.971875,2.3607812) \psline[linewidth=0.04cm,arrowsize=0.05291667cm 2.0,arrowlength=1.4,arrowinset=0.4]{->}(8.971875,4.360781)(17.171875,2.3607812) \psdots[dotsize=0.12](1.971875,-4.639219) \psdots[dotsize=0.12](3.971875,-4.639219) \psline[linewidth=0.04cm,arrowsize=0.05291667cm 2.0,arrowlength=1.4,arrowinset=0.4]{->}(1.971875,-4.639219)(3.971875,-4.639219) \psline[linewidth=0.04cm,arrowsize=0.05291667cm 2.0,arrowlength=1.4,arrowinset=0.4]{->}(3.971875,-4.639219)(5.471875,-4.639219) \psdots[dotsize=0.12](6.971875,-4.639219) \psbezier[linewidth=0.04,arrowsize=0.05291667cm 2.0,arrowlength=1.4,arrowinset=0.4]{->}(3.971875,-4.639219)(3.971875,-5.3592186)(6.311875,-5.2992187)(6.971875,-4.639219) \psline[linewidth=0.04cm,arrowsize=0.05291667cm 2.0,arrowlength=1.4,arrowinset=0.4]{->}(5.471875,-4.639219)(6.971875,-4.639219) \psline[linewidth=0.04cm,arrowsize=0.05291667cm 2.0,arrowlength=1.4,arrowinset=0.4]{->}(13.971875,-4.639219)(15.871875,-4.639219) \psline[linewidth=0.04cm,linestyle=dotted,dotsep=0.16cm](7.871875,-4.639219)(13.171875,-4.639219) \psbezier[linewidth=0.04,arrowsize=0.05291667cm 2.0,arrowlength=1.4,arrowinset=0.4]{->}(6.971875,-4.639219)(7.471875,-5.3392186)(12.971875,-5.039219)(13.971875,-4.639219) \psbezier[linewidth=0.04,arrowsize=0.05291667cm 2.0,arrowlength=1.4,arrowinset=0.4]{->}(6.971875,-4.639219)(7.721875,-6.039219)(14.171875,-5.8392186)(15.971875,-4.639219) \psbezier[linewidth=0.04,arrowsize=0.05291667cm 2.0,arrowlength=1.4,arrowinset=0.4]{->}(5.471875,-4.639219)(5.471875,-5.639219)(13.371875,-5.4392185)(13.971875,-4.639219) \psbezier[linewidth=0.04,arrowsize=0.05291667cm 2.0,arrowlength=1.4,arrowinset=0.4]{->}(5.471875,-4.639219)(5.471875,-6.3392186)(14.911269,-6.2392187)(15.971875,-4.639219) \psdots[dotsize=0.12](5.471875,-4.639219) \psbezier[linewidth=0.04,arrowsize=0.05291667cm 2.0,arrowlength=1.4,arrowinset=0.4]{->}(3.971875,-4.639219)(3.971875,-6.639219)(13.571875,-6.539219)(13.971875,-4.639219) \psbezier[linewidth=0.04,arrowsize=0.05291667cm 2.0,arrowlength=1.4,arrowinset=0.4]{->}(3.971875,-4.639219)(3.971875,-7.3392186)(15.971875,-7.4392185)(15.971875,-4.7392187) \usefont{T1}{ppl}{m}{n} \rput(1.601875,-4.619219){\small $s$} \usefont{T1}{ppl}{m}{n} \rput(3.921875,-4.3392186){\small $s_1$} \usefont{T1}{ppl}{m}{n} \rput(5.361875,-4.3592186){\small $s_2$} \usefont{T1}{ppl}{m}{n} \rput(6.881875,-4.3592186){\small $s_3$} \usefont{T1}{ppl}{m}{n} \rput(16.241875,-4.6992188){\small $g$} \usefont{T1}{ppl}{m}{n} \rput(13.881875,-4.3992186){\small $s_n$} \psdots[dotsize=0.12](15.971875,-2.6592188) \psline[linewidth=0.04cm,arrowsize=0.05291667cm 2.0,arrowlength=1.4,arrowinset=0.4]{->}(15.971875,-4.659219)(15.971875,-2.6592188) \usefont{T1}{ppl}{m}{n} \rput(17.351875,-1.5592188){\small $3$} \usefont{T1}{ppl}{m}{n} \rput(15.681875,-2.6392188){\small $f$} \psdots[dotsize=0.12](7.971875,-0.63921875) \psdots[dotsize=0.12](11.971875,-0.63921875) \psdots[dotsize=0.12](15.671875,-0.6592187) \usefont{T1}{ppl}{m}{n} \rput(7.661875,-0.51921874){\small $s_{1,2}$} \usefont{T1}{ppl}{m}{n} \rput(10.921875,-0.87921876){\small $s_{1,3}$} \usefont{T1}{ppl}{m}{n} \rput(11.921875,-0.87921876){\small $s_{1,4}$} \usefont{T1}{ppl}{m}{n} \rput(12.901875,-0.89921874){\small $s_{2,4}$} \psline[linewidth=0.04cm,linestyle=dotted,dotsep=0.16cm](9.371875,2.3607812)(12.271875,2.3607812) \psline[linewidth=0.04cm,linestyle=dotted,dotsep=0.16cm](13.371875,-0.63921875)(15.371875,-0.63921875) \usefont{T1}{ppl}{m}{n} \rput(15.731875,-0.8592188){\small $s_{e_m}$} \psbezier[linewidth=0.04,arrowsize=0.05291667cm 2.0,arrowlength=1.4,arrowinset=0.4]{->}(15.971875,-2.6592188)(16.03345,-2.3592188)(16.071875,0.14078125)(15.838542,0.24078125)(15.605208,0.34078124)(15.671875,-0.05921875)(15.671875,-0.6592187) \psbezier[linewidth=0.04,arrowsize=0.05291667cm 2.0,arrowlength=1.4,arrowinset=0.4]{->}(15.971875,-2.7408514)(16.171875,-1.9592187)(16.271875,-0.8592188)(16.071875,0.04078125)(15.871875,0.94078124)(14.971875,0.36078125)(14.471875,0.16078125)(13.971875,-0.03921875)(13.671875,-0.23921876)(12.971875,-0.63921875) \psbezier[linewidth=0.04,arrowsize=0.05291667cm 2.0,arrowlength=1.4,arrowinset=0.4]{->}(15.971875,-2.6592188)(16.571875,-1.7592187)(16.371876,-0.25921875)(16.171875,0.54078126)(15.971875,1.3407812)(13.971875,0.26078126)(13.471875,0.06078125)(12.971875,-0.13921875)(12.571875,-0.33921874)(11.971875,-0.63921875) \psbezier[linewidth=0.04,arrowsize=0.05291667cm 2.0,arrowlength=1.4,arrowinset=0.4]{->}(15.971875,-2.6592188)(16.571875,-1.780149)(16.671875,0.36868823)(16.071875,0.95473474)(15.471875,1.5407813)(13.171875,0.46403706)(12.671875,0.26078126)(12.171875,0.057525437)(11.771875,-0.23921876)(10.971875,-0.63921875) \psbezier[linewidth=0.04,arrowsize=0.05291667cm 2.0,arrowlength=1.4,arrowinset=0.4]{->}(15.971875,-2.6392188)(16.608503,-1.6793108)(16.971874,0.57384825)(16.365133,1.2173147)(15.758391,1.8607812)(10.971875,0.7607812)(10.371875,0.56078124)(9.771875,0.36078125)(8.477493,-0.3314404)(7.971875,-0.6453773) \psdots[dotsize=0.12](6.971875,-1.6392188) \psdots[dotsize=0.12](6.971875,-2.1392188) \psdots[dotsize=0.12](8.971875,-1.6392188) \psdots[dotsize=0.12](8.971875,-2.1392188) \psdots[dotsize=0.12](2.971875,0.36078125) \psdots[dotsize=0.12](3.971875,-0.63921875) \psdots[dotsize=0.12](1.971875,-0.63921875) \psdots[dotsize=0.2](15.971875,-4.659219) \psdots[dotsize=0.12](13.971875,-4.639219) \psline[linewidth=0.04cm,arrowsize=0.05291667cm 2.0,arrowlength=1.4,arrowinset=0.4]{->}(2.971875,2.3607812)(2.971875,0.36078125) \psline[linewidth=0.04cm,arrowsize=0.05291667cm 2.0,arrowlength=1.4,arrowinset=0.4]{->}(2.971875,0.36078125)(1.971875,-0.63921875) \psline[linewidth=0.04cm,arrowsize=0.05291667cm 2.0,arrowlength=1.4,arrowinset=0.4]{->}(2.971875,0.36078125)(3.971875,-0.63921875) \psline[linewidth=0.04cm,arrowsize=0.05291667cm 2.0,arrowlength=1.4,arrowinset=0.4]{->}(7.971875,-0.63921875)(6.971875,-1.6392188) \psline[linewidth=0.04cm,arrowsize=0.05291667cm 2.0,arrowlength=1.4,arrowinset=0.4]{->}(7.971875,-0.63921875)(8.971875,-1.6392188) \psdots[dotsize=0.12](7.971875,-3.5392187) \psdots[dotsize=0.12](8.971875,-2.6392188) \psdots[dotsize=0.12](6.971875,-2.6392188) \psline[linewidth=0.04cm,arrowsize=0.05291667cm 2.0,arrowlength=1.4,arrowinset=0.4]{->}(6.971875,-1.6392188)(6.971875,-2.1392188) \psline[linewidth=0.04cm,arrowsize=0.05291667cm 2.0,arrowlength=1.4,arrowinset=0.4]{->}(6.971875,-2.1392188)(6.971875,-2.6392188) \psline[linewidth=0.04cm,arrowsize=0.05291667cm 2.0,arrowlength=1.4,arrowinset=0.4]{->}(8.971875,-1.6392188)(8.971875,-2.1392188) \psline[linewidth=0.04cm,arrowsize=0.05291667cm 2.0,arrowlength=1.4,arrowinset=0.4]{->}(8.971875,-2.1392188)(8.971875,-2.6392188) \psdots[dotsize=0.12](2.471875,-2.6392188) \psline[linewidth=0.04cm,arrowsize=0.05291667cm 2.0,arrowlength=1.4,arrowinset=0.4]{->}(1.971875,-0.63921875)(2.471875,-2.6392188) \psdots[dotsize=0.12](1.271875,-2.6392188) \psline[linewidth=0.04cm,arrowsize=0.05291667cm 2.0,arrowlength=1.4,arrowinset=0.4]{->}(1.971875,-0.63921875)(1.271875,-2.6392188) \usefont{T1}{ppl}{m}{n} \rput(4.131875,-1.7392187){\small $1$} \psline[linewidth=0.04cm,arrowsize=0.05291667cm 2.0,arrowlength=1.4,arrowinset=0.4]{->}(1.271875,-2.6392188)(3.971875,-4.639219) \usefont{T1}{ppl}{m}{n} \rput(7.671875,-2.5992188){\small $2$} \usefont{T1}{ppl}{m}{n} \rput(3.151875,-1.5392188){\small $2$} \psbezier[linewidth=0.04,arrowsize=0.05291667cm 2.0,arrowlength=1.4,arrowinset=0.4]{->}(8.971875,-1.6193919)(8.060186,-1.5175894)(5.171875,-1.0392188)(4.048758,-1.1103797)(2.9256413,-1.1815405)(3.7440827,-2.0066016)(2.471875,-2.6392188) \usefont{T1}{ppl}{m}{n} \rput(1.351875,-1.6392188){\small $1$} \usefont{T1}{ppl}{m}{n} \rput(2.541875,0.40078124){\small $b_{1,2}$} \usefont{T1}{ppl}{m}{n} \rput(1.541875,-0.63921875){\small $\ell_{1,2}$} \usefont{T1}{ppl}{m}{n} \rput(4.351875,-0.5392187){\small $h_{1,2}$} \usefont{T1}{ppl}{m}{n} \rput(0.841875,-2.6592188){\small $v_{1,2}$} \usefont{T1}{ppl}{m}{n} \rput(2.731875,-2.8392189){\small $w_{1,2}$} \usefont{T1}{ppl}{m}{n} \rput(4.401875,-2.6392188){\small $y_{1,2}$} \usefont{T1}{ppl}{m}{n} \rput(8.321875,-3.5792189){\small $z_{1,2}$} \usefont{T1}{ppl}{m}{n} \rput(6.541875,-1.6592188){\small ${\alpha}_{1,2}$} \usefont{T1}{ppl}{m}{n} \rput(6.491875,-2.0992188){\small ${\beta}_{1,2}$} \usefont{T1}{ppl}{m}{n} \rput(6.521875,-2.6392188){\small ${\gamma}_{1,2}$} \usefont{T1}{ppl}{m}{n} \rput(9.361875,-1.6192187){\small ${\delta}_{1,2}$} \usefont{T1}{ppl}{m}{n} \rput(9.381875,-2.0992188){\small ${\kappa}_{1,2}$} \usefont{T1}{ppl}{m}{n} \rput(9.361875,-2.6392188){\small ${\lambda}_{1,2}$} \psbezier[linewidth=0.04,arrowsize=0.05291667cm 2.0,arrowlength=1.4,arrowinset=0.4]{->}(15.971875,-4.639219)(16.671875,-4.2592187)(16.971874,-3.4392188)(17.071875,-2.6392188)(17.171875,-1.8392187)(17.171875,0.86078125)(17.171875,2.3607812) \psbezier[linewidth=0.04,arrowsize=0.05291667cm 2.0,arrowlength=1.4,arrowinset=0.4]{->}(3.971875,-0.63921875)(5.471875,-0.94993305)(5.1200542,-3.5106473)(6.071875,-3.524933)(7.023696,-3.5392187)(6.871875,-3.5392187)(7.971875,-3.5392187) \psdots[dotsize=0.12](3.971875,-2.6392188) \psline[linewidth=0.04cm,arrowsize=0.05291667cm 2.0,arrowlength=1.4,arrowinset=0.4]{->}(3.971875,-0.63921875)(3.971875,-2.6392188) \psbezier[linewidth=0.04,arrowsize=0.05291667cm 2.0,arrowlength=1.4,arrowinset=0.4]{->}(3.971875,-2.6392188)(3.971875,-3.4392188)(4.771875,-4.3392186)(5.471875,-4.639219) \psdots[dotsize=0.12](10.971875,-0.63921875) \psdots[dotsize=0.12](12.971875,-0.63921875) \psbezier[linewidth=0.04,arrowsize=0.05291667cm 2.0,arrowlength=1.4,arrowinset=0.4]{->}(1.971875,-4.639219)(2.371875,-5.166048)(4.751875,-5.3592186)(5.471875,-4.639219) \psbezier[linewidth=0.04,arrowsize=0.05291667cm 2.0,arrowlength=1.4,arrowinset=0.4]{->}(1.971875,-4.639219)(2.671875,-5.8392186)(7.171875,-6.139219)(6.971875,-4.639219) \psbezier[linewidth=0.04,arrowsize=0.05291667cm 2.0,arrowlength=1.4,arrowinset=0.4]{->}(1.971875,-4.639219)(1.971875,-7.9392185)(16.071875,-7.639219)(13.971875,-4.639219) \psline[linewidth=0.04cm,arrowsize=0.05291667cm 2.0,arrowlength=1.4,arrowinset=0.4]{->}(6.971875,-1.6392188)(7.971875,-3.5392187) \psdots[dotsize=0.12](17.171875,3.3607812) \psline[linewidth=0.04cm,arrowsize=0.05291667cm 2.0,arrowlength=1.4,arrowinset=0.4]{->}(17.171875,2.3607812)(17.171875,3.3607812) \usefont{T1}{ppl}{m}{n} \rput(17.451876,3.3607812){\small $d'$} \psline[linewidth=0.04cm,arrowsize=0.05291667cm 2.0,arrowlength=1.4,arrowinset=0.4]{->}(6.971875,-2.6392188)(6.971875,-3.1392188) \psdots[dotsize=0.12](6.971875,-3.1392188) \usefont{T1}{ppl}{m}{n} \rput(7.151875,-2.8992188){\small $3$} \usefont{T1}{ppl}{m}{n} \rput(6.551875,-3.1792188){\small ${\eta}_{1,2}$} \psline[linewidth=0.04cm,arrowsize=0.05291667cm 2.0,arrowlength=1.4,arrowinset=0.4]{->}(8.971875,-2.6392188)(8.971875,-3.1392188) \psdots[dotsize=0.12](8.971875,-3.1392188) \usefont{T1}{ppl}{m}{n} \rput(9.381875,-3.1392188){\small ${\mu}_{1,2}$} \usefont{T1}{ppl}{m}{n} \rput(8.791875,-2.8992188){\small $3$} \end{pspicture}  } 
\end{center}

\caption{\label{fig:varIIacyclicNP}Directed graph $\gscg=(\vscg,\escg)$ in instance
$\mathcal{W}$}

\end{figure}
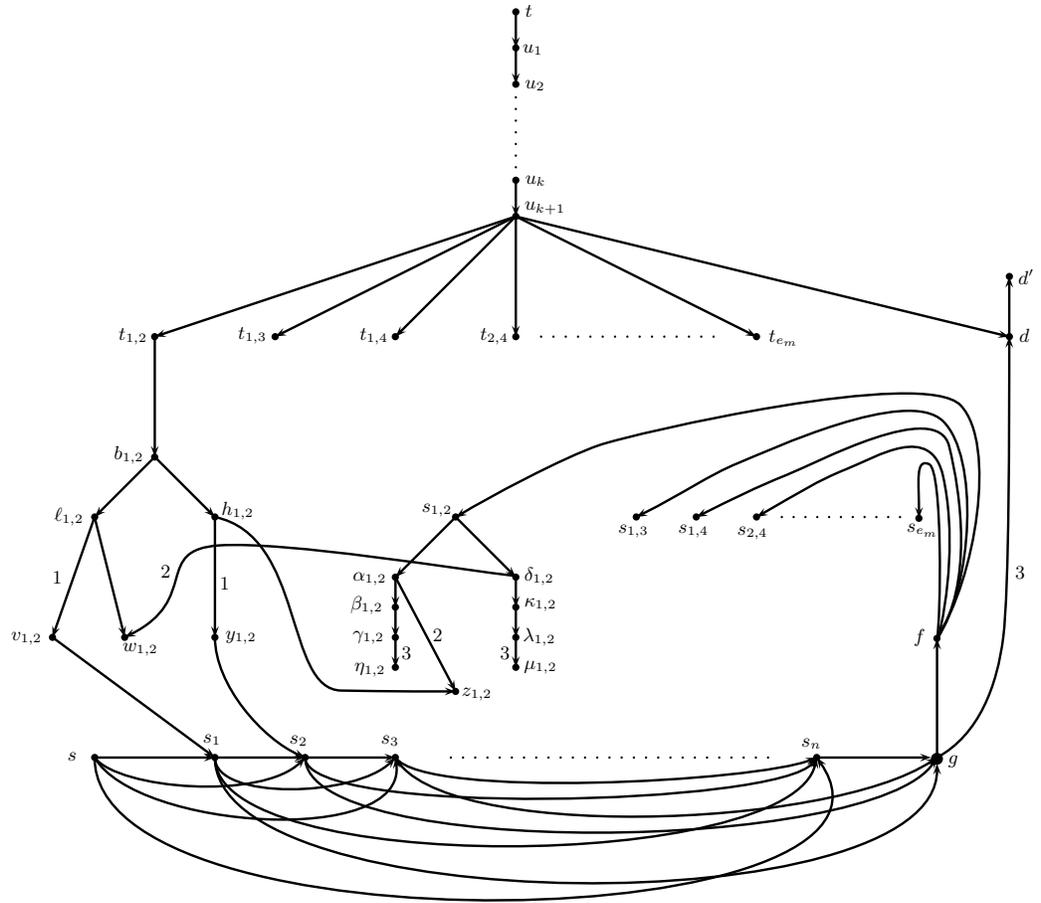

For the sake of readability, Fig.~\ref{fig:varIIacyclicNP} displays
the detailed construction of the graph for the edge $e=\{1,2\}\in \evc$.
Note that $\gscg$ is an acyclic directed graph. \\
The costs of the edges in graph $\gscg=(\vscg,\escg)$ are defined as
follows. Let $c(g,d)=3$. For each $e=\{i,j\}\in \evc$ let $c(\ell_{i,j},v_{i,j})=c(h_{i,j},y_{i,j})=1$,
$c(\delta_{i,j},w_{i,j})=c(\alpha_{i,j},z_{i,j})=2$, and $c(\gamma_{i,j},\eta_{i,j})=c(\lambda_{i,j},\mu_{i,j})=2$.
For the remaining edges $e'\in \escg$ let $c(e')=\eps$ with $\eps=\frac{1}{|\escg|}$.
Note that in any instance of \noun{Vertex Cover} obviously $k\leq n$
must hold. Hence, the size of $\gscg$ is polynomial in the size of $\gvc$,
since 
\[
|\vscg|=6+(k+1)+n+17m\,\in\mathcal{O}(n+m)
\]
 and 
\[
|\escg|=5+k+2n+\frac{(n-1)n}{2}+21m\,\in\mathcal{O}(n^{2}+m)
\]
In instance $\mathcal{W}$ of \scg, we are given the directed
graph $\gscg$ and the two players $A,B$. $A$ starts in vertex $s$,
while $B$ starts in vertex $t$. $A$ is the first player to move
along an edge.

\textbf{Part 2: The Reduction.} We show that the following claim holds:\\
Instance $\mathcal{W}$ of \pathdec\ ends with $c(A)<1$ if and only
if in instance $\mathcal{V}$ of \noun{Vertex Cover,} graph $\gvc=(\vvc,\evc)$
admits a vertex cover of size at most $k$.

By construction of $\gscg$, the first $(k+1)$ edges player $B$ selects
obviously form the path $(t,u_{1},u_{2},\ldots,u_{k+1})$. 
\begin{quote}
\textbf{Observation~1.} If the head of the $(k+1)$-th edge selected
by player $A$ is different from $g$, then the game ends with $c(A)>3$.

\textbf{Proof of Observation~1.} \\
\uline{\noun{Case }}\uline{a}\uline{\noun{:}} Assume that
$A$ has already visited $g$ before selecting the $(k+1)$-th edge.
I.e., after selecting the $(k+1)$-the edge, player $A$ is either
in one of the vertices $d,d'$ or in a vertex of the set $E'\cup\{f\}$,
where $E':=\{s_{i,j},\alpha_{i,j},\beta_{i,j},\gamma_{i,j},\delta_{i,j},\kappa_{i,j},\lambda_{i,j},z_{i,j},w_{i,j}|\{i,j\}\in \evc\}$.\\
If $A$ is in $d$ or $d'$, $A$ must have moved along the edge $(g,d)$
of cost $3$ because this is the only way $A$ can reach the vertices
$d,d'$. Thus, observation~1 is obviously true.\\
Assume that $A$ is in $E'\cup\{f\}$. 
\begin{itemize}
\item \vspace{-0.1cm}First, suppose $A$ is in one of the vertices in $E'$.
We will show that this in fact is not feasible, because otherwise
a meeting point is not possible. W.l.o.g., assume $A$ to be in one
of the vertices of $E'$ with index $\{i,j\}=\{1,2\}$. Then,
it is not hard to verify that the possible meeting points are either
(i) $w_{1,2}$ or $z_{1,2}$, or (ii) the vertices $s,s_{1},\ldots,s_{n},g$.
\\
To reach one of the vertices $w_{1,2}$ or $z_{1,2}$ (if reachable
for $A$), $A$ has at most $2$ more edges to select before $A$
gets stuck in one of these vertices; however, $B$, who is currently
in $u_{k+1}$, still needs to move along at least four edges to reach
one of $w_{1,2}$, $z_{1,2}$. Hence, meeting in one of these vertices
is impossible, because the players need to take turns at selecting
edges. For the vertices listed in (ii), $B$ needs to move along at
least $5$ edges to reach one of them; however, $A$ has at most four
edges to select before $A$ is stuck in a vertex. Hence, meeting in
one of these vertices is impossible as well. 
\item Second, suppose that $A$ is in $\{f\}$, i.e., has chosen edge $(g,f)$
as her $(k+1)$-th edge. Recall that $B$ is in $u_{k+1}$ after selecting
her $(k+1)$-th edge. Now, $A$ has to move along an edge $(f,s_{i,j})$.
W.l.o.g., assume that $A$ moves along $(f,s_{1,2})$. Clearly, the
best strategy for $B$ is to connect with $A$ as soon as possible,
only moving along edges of cost $\eps$. Moving to $d$ is infeasible,
since $A$ cannot reach $d$ anymore. Hence, $B$ moves along $(u_{k+1},t_{1,2})$.
Now $A$ has to choose between moving to $\alpha_{1,2}$ and $\delta_{1,2}$.
Assume that $A$ moves to $\alpha_{1,2}$ (the case $A$ moves to
$\delta_{1,2}$ follows in analogous manner). Next, $B$ necessarily
selects edge $(t_{1,2},b_{1,2})$. 

\begin{itemize}
\item If $A$ moves along the edge $(\alpha_{1,2},z_{1,2})$, $A$ is stuck
in $z_{1,2}$. However, $B$ cannot end the game with the selection
of just one additional edge, hence a meeting point is not possible. 
\item Hence, $A$ needs to move along the edge $(\alpha_{1,2},\beta_{1,2})$
instead. After $B$'s next move (either to $h_{1,2}$ or $\ell_{1,2}$),
$A$ necessarily moves along $(\beta_{1,2},\gamma_{1,2})$. Clearly,
irrespective of whether $B$ is in $h_{1,2}$ or $\ell_{1,2}$, the
game does not end after $B$'s next move. Hence, $A$ needs to select
an additional edge, which necessarily is the edge $(\gamma_{1,2},\eta_{1,2})$
of cost $3$. 
\end{itemize}
\end{itemize}
\uline{\noun{Case }}\uline{b}\uline{\noun{:}} Assume that
$A$ has not visited $g$ yet. Since $g$ is not the head of the $(k+1)$-th
edge selected by $A$, by construction of $\gscg$ that means that after
both players have selected $(k+1)$ edges, $A$ is in some vertex
$s_{i}$, $i\in \vvc$, while $B$ is located in $u_{k+1}$. \\
Irrespective of $A$'s next move (after which $A$ is in some $s_{i}$,
$i\in \vvc$ or in $g$), it is $B$'s best choice to move to $d$: As
a consequence, $A$ is forced to move along $(g,d)$ of cost $3$
either immediately or after each of the players have moved along another
edge of cost $\eps$.\kommentar{If $A$ moves to $g$ now, then $B$'s
best choice is to move to $d$, forcing $A$ to end the game by moving
along $(g,d)$ of cost $3$, because otherwise a meeting point is
no longer possible. \\
If $A$ does not move to $g$, i.e., moves along an edge $(s_{i},s_{j})$,
$1\leq i<j\leq n$, still it is $B$'s best choice to move to $d$:
it is easy to verify that otherwise $B$ necessarily has to move along
at least three edges of cost $\eps$, resulting in an outcome $c(B)\geq((k+1)+3)\eps$.
Hence, $B$ moves along $(u_{k+1},d)$. This forces $A$ to move along
$(s_{j},g)$, because otherwise a meeting point would not be possible
any more, since the players need to alternately select edges. Next,
$B$ moves along $(d,d')$ and $A$ needs to move along the costly
edge $(g,d)$. The respective costs are $c(A)>3$ and $c(B)=((k+1)+2)\eps$.}

Summing up, we have shown that instance $\mathcal{W}$ of \pathdec\ ends
with $c(A)>3$ if $g$ is not the head of the $(k+1)$-th edge selected
by player $A$.\hfill$\diamondsuit$
\end{quote}
\emph{Assume that the $(k+1)$-th edge selected by $A$ ends in $g$,
i.e., after both players have selected $(k+1)$ edges, $A$ is in
$g$ and $B$ in $u_{k+1}$.} Let $V':=\{i\in \vvc|A\mbox{ has visited }s_{i}\}$.
Note that $|V'|=k$.

\uline{\noun{Case I:}} $A$ moves along $(g,d$). Clearly, then
$B$ moves along $(u_{k+1},d)$; the game hence ends with $c(A)>3$
and $c(B)=(k+2)\eps$.

\uline{\noun{Case II:}} $A$ moves along $(g,f$). Now, it is $B$'s
turn. In what follows, we argue that $c(B)>1$ -- or, equivalently,
$c(A)<1$ -- if and only if $V'$ is a vertex cover of graph $\gvc$
in instance $\mathcal{V}$. In other words, we show that $c(B)<1$
if and only if $B$ is able to move along an edge $(u_{k+1},t_{i,j})$
such that $\{i,j\}\cap V'=\emptyset$. We illustrate this by assuming
that, w.l.o.g., $B$ moves along $(u_{k+1},t_{1,2})$. 

\uline{\noun{Case II}}\uline{a}\uline{\noun{:}} $\{1,2\}\cap V'=\emptyset$.

\begin{table}
{\small }%
\begin{tabular}{|c|c||cccccccccccc||l|}
\hline 
\multirow{2}{*}{\noun{\small Case}{\small{} IIa(1)(i)({*})\hspace{0.15cm}}} & {\small $A$} & {\small $f$} &  & {\small \hs$s_{1,2}$} &  & {\small \hs$\alpha_{1,2}$} &  & {\small \hs$\beta_{1,2}$} &  & {\small \hs$\gamma_{1,2}$} &  & {\small \hs$\eta_{1,2}$} &  & {\small $c(A)=3+\eps(k+6)$}\tabularnewline
\cline{2-15} 
 & {\small $B$} &  & {\small \hs$t_{1,2}$} &  & {\small \hs$b_{1,2}$} &  & {\small \hs$h_{1,2}$} &  & {\small \hs$y_{1,2}$} &  & {\small \hs$s_{2}$} &  & {\small \hs$g$} & {\small $c(B)=1+\eps(k+6)$}\tabularnewline
\hline 
\end{tabular}{\small }\\
{\small }%
\begin{tabular}{|c|c||cccccccc||l|}
\hline 
\multirow{2}{*}{\noun{\small Case}{\small{} IIa(1)(i)({*}{*})}} & {\small $A$} & {\small $f$} &  & {\small \hs$s_{1,2}$} &  & {\small \hs$\alpha_{1,2}$} &  & {\small \hs$z_{1,2}$} &  & {\small $c(A)=2+\eps(k+4)$}\tabularnewline
\cline{2-11} 
 & {\small $B$} &  & {\small \hs$t_{1,2}$} &  & {\small \hs$b_{1,2}$} &  & {\small \hs$h_{1,2}$} &  & {\small \hs$z_{1,2}$} & {\small $c(B)=\eps(k+5)$}\tabularnewline
\hline 
\end{tabular}{\small \par}

\caption{\label{tab:dir-acyclic-Case-IIa}Moves and outcomes in case IIa(1)
in instance $\mathcal{W}$. }
\end{table}

\begin{itemize}\item[(1)]$A$ moves along $(f,s_{1,2})$. $B$'s
next move is along $(t_{1,2},b_{1,2})$. Now, $A$ needs to decide
between moving to $\alpha_{1,2}$ or $\delta_{1,2}$. Assume $A$
chooses the former. We show that for $B$, the optimal behaviour is
to move to $h_{1,2}$, leading to an outcome%
\footnote{It is not hard to see that by analogous arguments we can conclude
the game also ends with $c(A)=2+\eps((k+1)+3)$ and $c(B)=((k+1)+4)\eps$
in the case that $A$ moves to $\delta_{1,2}$ instead of $\alpha_{1,2}$. %
} of $c(A)=2+\eps((k+1)+3)$ and $c(B)=((k+1)+4)\eps$ (see also Table~\ref{tab:dir-acyclic-Case-IIa}).

\begin{itemize}\item[(i)]Assume $B$ moves to $h_{1,2}$. 

\begin{itemize}\item[(*)]If $A$ moves along $(\alpha_{1,2},\beta_{1,2})$,
$B$ must not move along $(h_{1,2},z_{1,2})$, since otherwise a meeting
point is no longer possible. Hence, $B$ has to move along $(h_{1,2},y_{1,2})$,
imposing additional costs of $1$ on player $B$. Next, $A$ moves
along the only available edge $(\beta_{1,2},\gamma_{1,2})$. Now,
$B$ has no other choice but to move along $(y_{1,2},s_{2})$. Recall
that, by assumption, $s_{2}$ has not been visited yet. I.e., the
game is not over yet. Hence, $A$ needs to move along one more edge,
i.e., the edge $(\gamma_{1,2},\eta_{1,2})$ of cost $3$. In the next
move, $B$ ends the game by moving, e.g., along edge $(s_{2},g)$.
Summing up, $A$'s decision to move along $(\alpha_{1,2},\beta_{1,2})$
results in an outcome of $c(A)>3$. 

\item[(**)]If $A$ moves along $(\alpha_{1,2},z_{1,2})$ of cost
$2$ instead, in the next step the game ends with $B$'s move along
$(h_{1,2},z_{1,2})$. The respective costs are $c(A)=2+\eps((k+1)+3)$
and $c(B)=((k+1)+4)\eps$. 

\end{itemize}\emph{Hence, $A$ will move along $(\alpha_{1,2},z_{1,2})$
instead of $(\alpha_{1,2},\beta_{1,2})$ after $B$'s move to $h_{1,2}$.
The game ends with }$c(A)=2+\eps(k+4)$ and $c(B)=\eps(k+5)$.

\item[(ii)]Assume $B$ moves to $\ell_{1,2}$. Irrespective of $A$'s
next move, $B$ must not move along $(\ell_{1,2},w_{1,2})$, because
there is no edge emanating from $w_{1,2}$ and $A$ is unable to reach
$w_{1,2}$ from $\alpha_{1,2}$. Hence, $B$ would have to move along
$(\ell_{1,2},v_{1,2})$ of cost $1$, which exceeds the total cost
$B$ would have to carry if $B$ moved along $h_{1,2}$ instead of
$\ell_{1,2}$. \end{itemize}

\item[(2)]$A$ moves along $(f,s_{i,j})$ for some $\{i,j\}\in \evc\setminus\{1,2\}$.
$B$ moves along $(t_{1,2},b_{1,2})$. Clearly, $A$ and $B$ cannot
meet in one of the vertices $w_{i,j},z_{i,j}$ by construction of
$\gscg$. Hence, $B$ must traverse one of the vertices $s_{1},s_{2}$
in order to enable a meeting point. Now, $A$ has exactly four more
edges to select before $A$ gets stuck in a vertex (i.e, in $\eta_{i,j}$
or $\mu_{i,j}$), while $B$ needs to select exactly three more edges
to reach one of the vertices $s_{1},s_{2}$. However, since none of
the vertices have been visited yet, that means that $B$ would have
to select a further edge in order to end the game. Hence, $A$ needs
to move along either $(\gamma_{i,j},\eta_{i,j})$ or $(\lambda_{i,j},\mu_{i,j})$,
i.e., along an edge of cost $3$. \end{itemize}

Consequently, $A$'s optimal strategy is to move along $(f,s_{1,2})$.
\\
\emph{Summing up, in Case IIa the game ends with $c(A)=2+\eps(k+4)$
and $c(B)=\eps(k+5)$.}\\

\uline{\noun{Case II}}\uline{b:} $\{1,2\}\cap V'\not=\emptyset$.
Assume that $1\in V'$ (the case $2\in V'$ follows by analogous arguments). 

\begin{itemize}\item[(1)]$A$ moves along $(f,s_{1,2})$. $B$ moves
along $(t_{1,2},b_{1,2})$. Assume that $A$ moves along $(s_{1,2},\delta_{1,2})$
(the case that $A$ moves along $(s_{1,2},\alpha_{1,2})$ follows
analogously).

\begin{table}
{\small }%
\begin{tabular}{|c|c||cccccccccc|}
\hline 
\multirow{2}{*}{\noun{\small Case}{\small{} IIb(1)(i)\hspace{0.15cm}}} & {\small $A$} & {\small $f$} &  & {\small \hs$s_{1,2}$} &  & {\small \hs$\delta_{1,2}$} &  & {\small \hs$\kappa_{1,2}$} &  & {\small \hs$\lambda_{1,2}$} & \tabularnewline
\cline{2-12} 
 & {\small $B$} &  & {\small \hs$t_{1,2}$} &  & {\small \hs$b_{1,2}$} &  & {\small \hs$h_{1,2}$} &  & {\small \hs$y_{1,2}$} &  & {\small \hs$s_{2}$}\tabularnewline
\hline 
\end{tabular}{\small \par}

\caption{\label{tab:dir-acyclic-Case-IIb}Situation in Case~IIb(1)(i) in instance
$\mathcal{W}$. }
\end{table}

\begin{itemize}\item[(i)]$B$ moves along $(b_{1,2},h_{1,2})$. Now,
if $A$ moves along $(\delta_{1,2},w_{1,2})$, a meeting point is
no longer possible. Hence, $A$ needs to move along $(\delta_{1,2},\kappa_{1,2})$.
Clearly, $B$ has to move along $(h_{1,2},y_{1,2})$ with cost $1$,
otherwise a meeting point is not possible. Now, $A$ needs to move
along $(\kappa_{1,2},\lambda_{1,2})$, followed by $B$ moving along
the edge $(y_{1,2},s_{2})$. Table~\ref{tab:dir-acyclic-Case-IIb}
summarizes the moves performed.

\begin{itemize}\item[(*)]If $s_{2}$ has been visited already by
$A$, i.e., if $2\in V'$, then the game ends, the total costs being
$c(A)=\eps((k+1)+5)$ and $c(B)=1+\eps(k+5)$. 

\item[(**)]Otherwise, $A$ needs to move along one more edge, i.e.,
the edge $(\lambda_{1,2},\mu_{1,2})$ of cost $3$, before the game
ends with $B$'s last move, e.g., along the edge $(s_{2},g)$. In
this case the respective costs are $c(A)=3+\eps((k+1)+5)$ and $c(B)=1+\eps(k+6)$.
\end{itemize}

\item[(ii)]$B$ moves along $(b_{1,2},\ell_{1,2})$. Now, $A$ has
the choice between $(\delta_{1,2},w_{1,2})$ and $(\delta_{1,2},\kappa_{1,2})$
(see also Table~\ref{tab:dir-acyclic-Case-IIb-ii}).

\begin{table}
{\small }%
\begin{tabular}{|c|c||cccccccc||l|}
\hline 
\multirow{2}{*}{\noun{\small Case}{\small{} IIb(1)(ii)({*})\hspace{0.15cm}}} & {\small $A$} & {\small $f$} &  & {\small \hs$s_{1,2}$} &  & {\small \hs$\delta_{1,2}$} &  & {\small \hs$w_{1,2}$} &  & $c(A)=2+\eps(k+4)$\tabularnewline
\cline{2-11} 
 & {\small $B$} &  & {\small \hs$t_{1,2}$} &  & {\small \hs$b_{1,2}$} &  & {\small \hs$\ell_{1,2}$} &  & {\small \hs$w_{1,2}$} & {\small $c(B)=\eps(k+5)$}\tabularnewline
\hline 
\end{tabular}{\small }\\
{\small }%
\begin{tabular}{|c|c||cccccccccc||l|}
\hline 
\multirow{2}{*}{\noun{\small Case}{\small{} IIb(1)(ii)({*}{*})}} & {\small $A$} & {\small $f$} &  & {\small \hs$s_{1,2}$} &  & {\small \hs$\delta_{1,2}$} &  & {\small \hs$\kappa_{1,2}$} &  & \hs$\lambda_{1,2}$ &  & $c(A)=\eps(k+6)$\tabularnewline
\cline{2-13} 
 & {\small $B$} &  & {\small \hs$t_{1,2}$} &  & {\small \hs$b_{1,2}$} &  & {\small \hs$\ell_{1,2}$} &  & {\small \hs$v_{1,2}$} &  & \hs$s_{1}$ & {\small $c(B)=1+\eps(k+5)$}\tabularnewline
\hline 
\end{tabular}{\small \par}

\caption{\label{tab:dir-acyclic-Case-IIb-ii}Moves and outcomes in Case~IIb(1)(ii)
in instance $\mathcal{W}$. }
\end{table}

\begin{itemize}\item[(*)]If $A$ chooses $(\delta_{1,2},w_{1,2})$
with cost $2$, then $B$ obviously moves along $(\ell_{1,2},w_{1,2})$
and the game ends. In this case, the costs of $A$ exceed $2$, i.e.,
$c(A)>2$ holds.

\item[(**)]If $A$ chooses the cheaper edge $(\delta_{1,2},\kappa_{1,2})$
with cost $\eps$, then $B$ needs to move along $(\ell_{1,2},v_{1,2})$
of cost $1$ to enable a meeting point. In the next steps, $A$ moves
along $(\kappa_{1,2},\lambda_{1,2})$, followed by $B$ moving along
$(v_{1,2},s_{1})$. This last move ends the game, because by assumption
$s_{1}$ has been visited by $A$, i.e., $1\in V'$ holds. The respective
costs are $c(A)=((k+1)+5)\eps$ and $c(B)=1+\eps(k+5)$. \end{itemize}\end{itemize}

Hence, if $B$ moves along $(b_{1,2},\ell_{1,2})$, then $A$ moves
along $(\delta_{1,2},\kappa_{1,2})$, which finally leads to an outcome
of $c(A)=((k+1)+5)\eps$ and $c(B)=1+\eps(k+5)$. The same outcome
is achieved if $B$ moves along $(b_{1,2},h_{1,2})$ instead and $2\in V'$
holds; if the latter is not the case, $B$ does not move along $(b_{1,2},h_{1,2})$
because this would lead to a worse outcome for $B$. \\
\emph{To sum up, if $A$ moves along $(f,s_{1,2})$, then the game
ends with $c(A)=((k+1)+5)\eps$ and $c(B)=1+\eps(k+5)$.}

\item[(2)]$A$ moves along $(f,s_{i,j})$ for some $\{i,j\}\in \evc\setminus\{1,2\}$.
After $B$'s move along $(t_{1,2},b_{1,2})$, assume that $A$ moves
along $(s_{i,j},\delta_{i,j})$ (the case $A$ moves along $(s_{i,j},\alpha_{i,j})$
follows by analogous arguments). As above, $B$ has the choice between
$(b_{1,2},h_{1,2})$ and $(b_{1,2},\ell_{1,2})$. 

\begin{itemize}\item[(i)]$B$ moves along $(b_{1,2},h_{1,2})$. With
exactly the same argumentation as above, it is not hard to see that
the game ends with (i) $c(A)=((k+1)+5)\eps$ and $c(B)=1+\eps(k+5)$
if $2\in V'$ or (ii) $c(A)=3+\eps((k+1)+5)$ and $c(B)=1+\eps(k+6)$
if $2\notin V'$.

\item[(ii)]$B$ moves along $(b_{1,2},\ell_{1,2})$. Now, $A$ has
no choice but to move along $(\delta_{i,j},\kappa_{i,j})$, since
moving to $w_{i,j}$ makes $A$ getting stuck in $w_{i,j}$, and $B$
is unable to reach a meeting point by moving along a single edge emanating
from $\ell_{1,2}$. Exactly as above, it follows that the game ends
with $c(A)=((k+1)+5)\eps$ and $c(B)=1+\eps(k+5)$.\end{itemize}

\emph{Hence, analogously to above it follows that the game ends with
$c(A)=((k+1)+5)\eps$ and $c(B)=c(B)=1+\eps(k+5)$ if $A$ moves along
$(f,s_{i,j})$ for some $\{i,j\}\in \evc\setminus\{1,2\}$. }\end{itemize}

\emph{As a consequence, Case IIb ends with $c(A)=((k+1)+5)\eps$ and
$c(B)=1+\eps(k+5)$.}\\

\emph{\uline{\noun{Conclusion:}}}\uline{ }If (i) the head of
the $(k+1)$-th edge selected by $A$ is not $g$, or (ii) $A$ moves
along $(g,d)$, then the game ends with $c(A)>3$ (Observation~1
resp.\ Case~I). On the other hand, in Case II $A$ is guaranteed
an outcome of less than $3$. \\
As a consequence, in our considered instance $\mathcal{W}$ player
$A$ is in vertex $g$ after selecting her $(k+1)$-th edge and moves
along $(g,f)$ in the next step. Thus, Case IIa ($V'$ is not a vertex
cover of size $k$ in the graph $\gvc$) or Case IIb ($V'$ is a vertex
cover of size $k$ in the graph $\gvc$) applies. In particular, it follows
that instance $\mathcal{W}$ of \pathdec\ ends with $c(A)=((k+1)+5)\eps<1$
(or equivalently, $c(B)>1$) if and only if there is a vertex cover
of size $k$ in instance $\mathcal{V}$ of \textsc{Vertex Cover}.
Therewith, \scg\ is $\np$-hard for directed acyclic graphs. \newline
Finally, we extend this $\np$-hardness result to directed acyclic graphs that are bipartite. 
In order to do so, we modify instance $\mathcal{W}$ with graph $\gscg=(\vscg,\escg)$ into 
instance $\mathcal{W}'$ with graph $G'$ by ``splitting'' each edge of $\escg$; 
i.e., $G'$ is created from $\gscg$ by, for each $e=(u,v)\in \escg$ of cost $c$, introducing
vertex $m_e$  and replacing $e$ with the two edges $(u,m_e)$ and $(m_e ,v)$ of cost $\frac{c}{2}$ each.
It is easy to see that the resulting graph is bipartite. 
It is also not hard to verify that the nature of the game is preserved and arguing analogously to above yields the desired result.~\qed

\section{\scg\ on Cactus Graphs}
\label{sec:cactus}

After the negative complexity results of Section~\ref{sec:hard} we try to identify
polynomially solvable special cases and give some indications on the boundary
between $\np$-hard and polynomially solvable cases.
As an easy first step we state a trivial positive result for trees.
Clearly (R3) is satisfied in this case.
 
\begin{prop}\label{th:tree}
If $G$ is a tree, then \opath\ of \scg\ can be determined in $O(n)$ time.
\end{prop}

As a graph class which is only slightly more complex than trees, 
many authors consider cactus graphs for optimization problems.
A {\em cactus graph} 
is a graph where each edge is contained in at most one simple cycle.
Equivalently, any two simple cycles have at most one vertex in common.
This means that one could contract each cycle into a vertex in a unique way
and obtain a tree.
Note that cactus graphs are a subclass of series-parallel graphs and
thus have treewidth at most $2$.

Rather surprisingly, we will show in the section below that \scg\
is already $\np$-hard on cactus graphs without restriction (R3).
This is remarkable since it contrasts related games such as 
\geo\ or \spg (cf.~\citet{bo93} and \citet{dps15a}). 
On the other hand, imposing (R3) makes the problem polynomially solvable on cacti.

\subsection{Hardness  for \scg\ on Cactus Graphs}
\label{sec:cactushard}

To show the $\np$-hardness of \scg\ on cactus graphs
in its general form without condition (R3)
we use a reduction from the well known strongly $\np$-complete problem \textsc{$3$-Partition}:
\begin{quote}
\noun{$3$-Partition}

\textbf{Input:} A multiset $C=\{\tilde{c}_{1},\tilde{c}_{2},...,\tilde{c}_{3n-1},\tilde{c}_{3n}\}$ of $3n$ integers
such that $\sum_{i=1}^{3n} = n \tilde{K}$.

\textbf{Question:} Is there a partition of $C$ into $n$ triples such that the elements of each triple sum up to exactly $\tilde{K}$. 
\end{quote}

\begin{theorem}\label{th:hardcactus}
\scg\ is strongly \np -hard for directed cactus graphs.
\end{theorem}

\proof
Given a \noun{$3$-Partition} instance $I$, 
we consider an equivalent \noun{$3$-Partition} instance $J$ 
which is defined as follows. Let $M  = 2 n \tilde{K}$, $c_i = M + 2 \tilde{c}_i$ and $K = 3 M + 2 \tilde{K}$. 
Then $J$ asks for a partition of the $c_i$ into triples such that each triple 
sums to exactly $K$. Clearly $J$ is polynomially bounded in the size of $I$. 

Given $J$, we will construct an instance of \scg\ for players $A$ and $B$ 
by means of a cactus graph $G_{3P}$ 
as illustrated by Figure \ref{fig:cac}.
All edge costs are $0$ with the exception of two edges emanating in $a_4$ resp.\ $b_4$
with cost $1$ resp.\ $2$ (marked in red).
The black integers describe the number of edges of the corresponding cycles. 
Player $A$ and $B$ start at vertices $a_1$ and $b_1$ respectively; $A$ is first to choose an edge.  

The high-level idea of the construction can be described as follows:
Whenever $J$ is a "yes" instance, players $A$ and $B$ 
reach a \opath\ with cost $(0,1)$, whereas in a "no" instance 
\opath\ has cost $(1,0)$. 
In case of a "yes" instance player $B$ can dictate a path starting at $b_1$ going to $b_2$, $b_3$ and $b_4$. At $b_4$ player $B$ moves 
through the cycles of length $K$ and returns back to $b_4$ while avoiding the cycle 
of length $3M-2$ with cost $2$. 
From there the path continues via $b_5$, $b_3$ and $b_2$ to $t$, the unique meeting point.
Note that at $b_5$ some of the red edges might be used. 
In case of a "no" instance $A$ can force $B$ to use the cycle of length $3M-2$ with cost $2$
at $b_4$ whenever $B$ starts by moving to $b_2$ (and onwards to $b_4$).
Therefore, $B$ chooses the direct path of length $3$ to $t$ with cost $1$.

\medskip
Now we proceed with the details of the construction.

\textbf{Let $J$ be a "yes" instance of} \noun{$3$-Partition}.
We will show that \opath\ of the corresponding \scg\ yields costs $(0,1)$, i.e., $0$ cost to $A$ a cost of $1$ to $B$. 
$A$ starts in $a_1$ and $B$ in $b_1$. $A$ has no choice but to move to $a_2$. 
Now, going to $t$ directly by the unique path consisting of three edges only has cost $1$ for $B$ and cost $0$ for $A$; in the following, we will argue that indeed this is the path choosen by $B$ in \opath.

Assume that $B$'s first move is to go to vertex $b_2$. 
At this point of the game, player $A$ has no choice but to move towards $a_3$ since otherwise a feasible solution is not possible because $B$ would get stuck. 
Analogously, player $B$ now cannot move towards $t$ from $b_2$ (via the path of length four), since this does not allow for a feasible solution. 
Hence $B$ has to go towards $b_3$. 
When the players are in $a_3$ and $b_3$ respectively, 
$A$ cannot go towards $t$ by going ``back'' towards $a_2$ since again $B$ would get stuck.
Hence, $A$ moves towards $a_4$. It is not difficult to see that this means that $B$ has to move towards $b_4$ 
because otherwise a feasible outcome is not possible anymore.
As a consequence, $A$ and $B$ end up in  $a_4$ and $b_4$ respectively, with $B$ having to choose the next edge.

Going directly to $t$ (via $b_5$ and $b_3$) is not feasible for $B$ since this requires $B$ to traverse at least $8$ edges; $A$ however needs only $6$ moves directly or at least $6 + \frac{3M}{4}$ moves via the shortest cycle at $a_4$.
Also, going to $t$ via $b_5$ and some of the red cycles at $b_5$ does not allow for a feasible solution because in that way,  $B$ can traverse at most  $8 + \frac{M}{2} + 5 \max_{i}\{\tilde{c}_i\}$ edges.

Now, $B$ tries to avoid the cycle of length $3M-2$ since it contains the most expensive edge in the graph of cost $2$. 
So $B$ can choose either a cycle of length $K$ or the cycle of length $\frac{3M}{4}$ 
(and optionally the second one of the same length). 

In the following we say that $B$ and $A$ are {\em synchronized} whenever $B$ is the one to choose the next move starting from $b_4$ and $A$ at the same time is located at $a_4$. 
Assume that $B$ chooses a cycle of length $K$: 
since $J$ is a "yes" instance $A$ can choose three cycles corresponding to three of the $c_i$ summing up to exactly $K$, so that $B$ and $A$ stay synchronized. 
When $B$ chooses to traverse a cycle of length $\frac{3M}{4}$ (or both of them), $A$ can choose to traverse a cycle of the same length (or both of them). 
By the same argument as above as long as $B$ and $A$ are synchronized $B$ cannot go to $t$. 
Hence, no matter which cycle $B$ chooses $A$ can choose cycles so that they stay synchronized until all cycles of length $K$ are consumed by $B$ and $A$ has only three of the cycles corresponding to three of the $c_i$ left. 
So $B$ has to take the cycle of length $3M-2$ which results in a suboptimal cost of $2$ for $B$ (see remarks at the beginning of this proof). 
Note that this move is feasible since $A$ can use the remaining cycles corresponding to the $c_i$. 
By the length of these cycles, $A$ is still at a vertex of its third chosen cycle when $B$ reaches $b_4$. 
Due to the choice of the $c_i$, $A$ needs an even number of edges until $a_4$ is reached. 
By construction, now $B$ can go to $t$ by using as many 
red edges at $b_5$ as needed in order to guarantee a feasible path for $A$ 
(in particular, note that the number of red edges used by $B$ is even too). 
Realizing this outcome with cost $2$ resulting from the move to $b_2$,
$B$ will instead choose the direct path from $b_1$ to $t$ of length three with cost $1$, 
and $A$ can choose the path via $a_2$ with zero cost. 

\medskip
\textbf{Let $J$ be a "no" instance of} \noun{$3$-Partition}. 
We will argue that in this case $B$ can always choose a path with cost $0$. 
Now $B$ avoids the direct path from $b_1$ to $t$ which would have cost $1$. 
By the same arguments as in the "yes" case $B$ goes to $b_4$ and $A$ to $a_4$ and they are synchronized. 
$B$ can choose cycles of length $K$ and $A$ is not able to stay synchronized with $B$ by only choosing cycles corresponding to $c_i$ until all cycles of length $K$ are consumed by $B$. 
This follows by the fact that the corresponding \noun{$3$-Partition} instance is a "no" instance. 
In the following we will analyze all possible moves by $A$ which lead to a loss of synchronization. 
Assume that at stage $r_0$ (i.e., after player $B$ has traversed exactly $r_0$ edges and $A$ has traversed $r_0+1$ edges, i.e., it is player $B$'s turn to choose an edge)  the players are synchronized.
Assume $B$ chooses a cycle of length $K$ and $A$ chooses cycles such that at stage $r_1 = r_0 + 2 K$ they are not synchronized any more. Note that  $B$ is the one to choose the next edge at $r_1$. 
At stage $r_1$, $B$ has at least one cycle of length $K$ and 
$A$ at least six cycles corresponding to six $c_i$ left to choose because $J$ is a "no" instance. 

\smallskip

\textit{Case 1:} $A$ answers with three cycles corresponding to three $c_i$ such that at $r_1$ they are not synchronized any more. 
Recall that $K= 3M + 2 \tilde{K}$ and $c_i = M + 2\tilde{c}_i$, hence there are two possibilities for being not synchronized at $r_1$.

\textit{Case 1.1:} $A$ did not yet finish its third cycle corresponding to some $c_i$ at $r_1$.\\
In this case, $A$ has still less than 
$$ 3  (M + 2\max_{i}\{\tilde{c}_i\}) - (3M + 2 \tilde{K}) \leq 4 \max_{i}\{\tilde{c}_i\} $$
edges to go in order to reach $a_4$ for the next time. 
Here,  $ 3  (M + 2\max_{i}\{\tilde{c}_i\})$ is the total sum of edges  that $A$ needs for three large cycles corresponding to $c_i$;   $3M + 2 \tilde{K} = K$ is the number of edges that $B$ needs in order to reach $b_4$ at $r_1$. 
Note that $\tilde{K} > \tilde{c}_i$ for all $i$.  
But now $B$ can feasibly go to $t$ by using as many of the red cycles at $b_5$ as necessary in order to feasibly force $A$ to $t$ via the edge of cost $1$ at $a_4$ (see Figure~\ref{fig:1.1}). The game thus ends with a total cost of $1$ for $A$ and $0$ for $B$.

\begin{figure}
\begin{center}

\scalebox{0.8} 
{
\begin{pspicture}(0,-2.2192187)(14.782,2.2192187)
\definecolor{color1294}{rgb}{1.0,0.0,0.023529411764705882}
\definecolor{color1296}{rgb}{0.0,0.050980392156862744,1.0}
\definecolor{color1299}{rgb}{0.00392156862745098,0.00392156862745098,0.00392156862745098}
\psline[linewidth=0.024cm,tbarsize=0.07055555cm 5.0]{|*-|*}(0.77,0.30578125)(7.77,0.30578125)
\usefont{T1}{ptm}{m}{n}
\rput(0.7,0.0){$r_0$}
\usefont{T1}{ptm}{m}{n}
\rput(7.8,0.0){$r_1$}
\usefont{T1}{ptm}{m}{n}
\rput(0.34453124,1.2157812){$A:$}
\usefont{T1}{ptm}{m}{n}
\rput(0.34453124,-0.38421875){$B:$}
\psline[linewidth=0.024cm,tbarsize=0.07055555cm 5.0]{|*-|*}(7.77,0.30578125)(14.77,0.30578125)
\psline[linewidth=0.024cm,linecolor=color1294,tbarsize=0.07055555cm 5.0]{|*-|*}(0.77,-0.49421874)(7.77,-0.49421874)
\usefont{T1}{ptm}{m}{n}
\rput(4.3945312,-0.9842188){$K = 3M +2 \tilde{K}$}
\psline[linewidth=0.024cm,linecolor=color1296,tbarsize=0.07055555cm 5.0]{|*-|*}(0.77,1.1057812)(3.17,1.1057812)
\psline[linewidth=0.024cm,linecolor=color1296,tbarsize=0.07055555cm 5.0]{|*-|*}(3.17,1.1057812)(5.77,1.1057812)
\psline[linewidth=0.024cm,linecolor=color1296,tbarsize=0.07055555cm 5.0]{|*-|*}(5.77,1.1057812)(8.37,1.1057812)
\psline[linewidth=0.024cm,linecolor=color1299,linestyle=dashed,dash=0.16cm 0.16cm,tbarsize=0.07055555cm 5.0]{|-|}(7.77,1.5057813)(8.37,1.5057813)
\usefont{T1}{ptm}{m}{n}
\rput(8.884531,2.0157812){$ \leq 6 \max_{i}\{\tilde{c}_i\} $}
\psline[linewidth=0.024cm,linecolor=color1294,tbarsize=0.07055555cm 5.0]{|*-|*}(7.77,-0.49421874)(7.97,-0.49421874)
\psline[linewidth=0.024cm,linecolor=color1294,tbarsize=0.07055555cm 5.0]{|*-|*}(7.97,-0.49421874)(8.57,-0.49421874)
\psline[linewidth=0.024cm,linecolor=color1294,tbarsize=0.07055555cm 5.0]{|*-|*}(8.77,-0.49421874)(8.57,-0.49421874)
\psline[linewidth=0.024cm,linecolor=color1296,tbarsize=0.07055555cm 5.0]{|*-|*}(8.37,1.1057812)(8.77,1.1057812)
\usefont{T1}{ptm}{m}{n}
\rput(8.464531,0.81578124){$a_4$}
\usefont{T1}{ptm}{m}{n}
\rput(9.034532,1.0157813){$t$}
\usefont{T1}{ptm}{m}{n}
\rput(9.034532,-0.58421874){$t$}
\usefont{T1}{ptm}{m}{n}
\rput(8.074532,-0.7842187){$b_5$}
\psline[linewidth=0.024cm,linecolor=color1299,linestyle=dotted,dotsep=0.16cm](8.35,-0.6742188)(9.35,-1.6742188)
\usefont{T1}{ptm}{m}{n}
\rput(10.0295315,-1.9842187){extra cycles at $b_5$}
\end{pspicture} 
}
\caption{Case 1.1} \label{fig:1.1}
\end{center}
\end{figure}

\medskip

\textit{Case 1.2:} $A$ has finished its third cycle (and reached $a_4$ again) before stage $r_1$ is reached.
Following $B$'s next move after $A$ has reached $a_4$, $A$ has the following three options:

\textit{Case 1.2.1:} $A$ chooses to traverse both cycles of length $\frac{3M}{4}$. 
But now $B$ can simply choose a cycle of length $K$ and go to $t$ via enough red cycles at $b_5$. 
This is feasible since after the cycles of length $\frac{3M}{4}$ $A$ can still choose two cycles corresponding to two $c_i$ and then has a ``delay'' 
(with respect to the moment when $B$ reaches $b_4$ after the cycle of length $K$) which lies in the following interval:
\begin{equation} \label{eq:ul}
\left[\frac{M}{2} - 4 \tilde{K}, \frac{M}{2} +  4 \max_{i}\{\tilde{c}_i\}\right] 
\end{equation}
The lower bound comes from the following estimate. 
At $r_0$ $A$ and $B$ are synchronized, hence at $r_1$ $A$ has used  
at most
 $$(3M + 2 \tilde{K}) - 3(M + 2\min_{i}\{\tilde{c}_i\}) \leq  2 \tilde{K}$$
edges of its first cycle of length $\frac{3M}{4}$. 
Hence, at the moment when $B$ reaches $b_4$ after $r_1$ for the next time, $A$ has at least
$$
\left(  \frac{3M}{2} + 2(M + 2 \min_{i}\{\tilde{c}_i\}) -  2 \tilde{K} \right)   - \left( 3M + 2 \tilde{K} \right)\geq \frac{M}{2} - 4 \tilde{K} 
$$
edges to traverse on the cycle corresponding to the second $c_i$ (cf.~Figure~\ref{fig:1.2}). 

Hence, the game ends with a total cost of $1$ for $A$ and $0$ for $B$.

\begin{figure}
\begin{center}
\scalebox{0.7} 
{
\begin{pspicture}(0,-2.2075)(19.9,2.1795)
\definecolor{color184}{rgb}{1.0,0.0,0.023529411764705882}
\definecolor{color193}{rgb}{0.0,0.050980392156862744,1.0}
\definecolor{color194}{rgb}{0.00392156862745098,0.00392156862745098,0.00392156862745098}
\psline[linewidth=0.024cm,tbarsize=0.07055555cm 5.0]{|*-|*}(0.81,0.7675)(8.81,0.7675)
\usefont{T1}{ppl}{m}{n}
\rput(0.8	,0.4775){$r_0$}
\usefont{T1}{ppl}{m}{n}
\rput(8.8,0.4775){$r_1$}
\usefont{T1}{ppl}{m}{n}
\rput(0.38453126,1.5975){$A:$}
\usefont{T1}{ppl}{m}{n}
\rput(0.34453124,0.0175){$B:$}
\psline[linewidth=0.024cm,linecolor=color184,tbarsize=0.07055555cm 5.0]{|*-|*}(0.81,-0.0325)(8.81,-0.0325)
\usefont{T1}{ppl}{m}{n}
\rput(4.434531,-0.3225){$K = 3M +2 \tilde{K}$}
\usefont{T1}{ppl}{m}{n}
\rput(8.524531,1.3375){$a_4$}
\usefont{T1}{ppl}{m}{n}
\rput(9.574532,1.8975){$\frac{3M}{4}$}
\usefont{T1}{ppl}{m}{n}
\rput(10.874531,1.8975){$\frac{3M}{4}$}
\psline[linewidth=0.024cm,linecolor=color193,tbarsize=0.07055555cm 5.0]{|*-|*}(0.81,1.5675)(3.41,1.5675)
\psline[linewidth=0.024cm,linecolor=color194,tbarsize=0.07055555cm 5.0]{|*-|*}(8.61,1.5675)(10.21,1.5675)
\psline[linewidth=0.024cm,linecolor=color184,tbarsize=0.07055555cm 5.0]{|*-|*}(8.81,-0.0325)(16.81,-0.0325)
\psline[linewidth=0.024cm,tbarsize=0.07055555cm 5.0]{|*-|*}(8.81,0.7675)(16.81,0.7675)
\psline[linewidth=0.024cm,linecolor=color193,tbarsize=0.07055555cm 5.0]{|*-|*}(3.41,1.5675)(6.01,1.5675)
\psline[linewidth=0.024cm,linecolor=color193,tbarsize=0.07055555cm 5.0]{|*-|*}(6.01,1.5675)(8.61,1.5675)
\psline[linewidth=0.024cm,linecolor=color194,tbarsize=0.07055555cm 5.0]{|*-|*}(10.21,1.5675)(11.81,1.5675)
\psline[linewidth=0.024cm,linecolor=color193,tbarsize=0.07055555cm 5.0]{|*-|*}(11.81,1.5675)(14.41,1.5675)
\psline[linewidth=0.024cm,linecolor=color193,tbarsize=0.07055555cm 5.0]{|*-|*}(14.41,1.5675)(17.01,1.5675)
\usefont{T1}{ppl}{m}{n}
\rput(11.084531,-0.2625){$K $}
\psline[linewidth=0.024cm,linecolor=color194,linestyle=dotted,dotsep=0.16cm](8.81,2.1675)(8.81,-0.6325)
\psline[linewidth=0.024cm,linecolor=color194,linestyle=dotted,dotsep=0.16cm](16.81,2.1675)(16.81,-0.6325)
\psline[linewidth=0.024cm,linecolor=color184,tbarsize=0.07055555cm 5.0]{|*-|*}(16.81,-0.0325)(17.81,-0.0325)
\psline[linewidth=0.024cm,linecolor=color193,tbarsize=0.07055555cm 5.0]{|*-|*}(17.01,1.5675)(17.81,1.5675)
\usefont{T1}{ppl}{m}{n}
\rput(18.03453,-0.0425){$t$}
\usefont{T1}{ppl}{m}{n}
\rput(18.054531,1.5575){$t$}
\psline[linewidth=0.024cm,linestyle=dotted,dotsep=0.16cm](17.23,-0.2525)(17.59,-1.6325)
\usefont{T1}{ppl}{m}{n}
\rput(18.018593,-1.9825){possible extra cycles}
\usefont{T1}{ppl}{m}{n}
\rput(16.514532,-0.2625){$b_4$}
\end{pspicture} 
}

\caption{Case 1.2.1} \label{fig:1.2}

\end{center}
\end{figure}
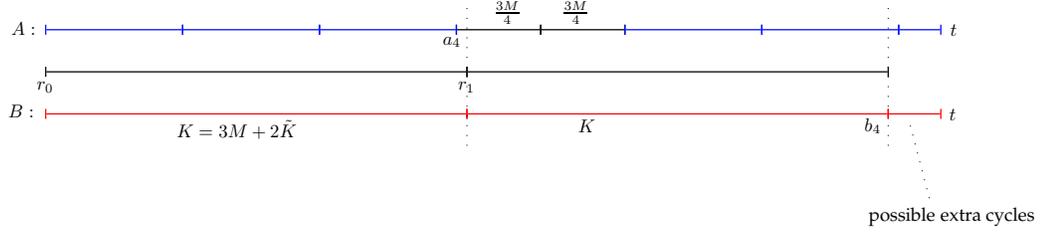

The upper bound of (\ref{eq:ul}) can be derived by the following argument. 
At $r_0$ $A$ and $B$ are synchronized, hence at $r_1$ $A$ has used at least two edges of its first cycle of length $\frac{3M}{4}$. 
Therefore, when $B$ reaches $b_4$ for the first time after stage $r_1$, $A$ has at most
$$
\left(  \frac{3M}{2} + 2(M + 2 \max_{i} \{\tilde{c}_i\}) -  2 \right)  - \left( 3M + 2 \tilde{K} \right)
\leq \frac{M}{2} + 4 \max_{i}\{\tilde{c}_i\}
$$
edges left to traverse on the cycle corresponding to the last $c_i$.

\textit{Case 1.2.2:} $A$ chooses one cycle of length $\frac{3M}{4}$. 
$B$ now has to choose the cycle of length $K$ again ($B$ cannot foresee whether $A$ chooses only one of the cycles of length $\frac{3M}{4}$). 
Now $A$ chooses $3$ cycles corresponding to three $c_i$. 
When $B$ finishes the cycle of length $K$, the number of edges which $A$ needs for finishing the last cycle lies in the interval: 
$$ \left[ \frac{3M}{4} - 4 \tilde{K}, \frac{3M}{4} + 6 \max_{i}\{\tilde{c}_i\} \right]$$
If $A$ has more than $\frac{3M}{4}$ moves left, $B$ chooses one cycle of length $\frac{3M}{4}$ and goes to $t$ via red cycles at $b_5$. 
Otherwise, $B$ chooses  both cycles of length $\frac{3M}{4}$ forcing $A$ to choose another cycle corresponding to a $c_i$ 
(note that at least three of them are still available). 
Finally, the number of edges which $A$ needs for finishing this cycle lies in the interval:  
$$\left[ \frac{M}{4} - 4 \tilde{K}, \frac{M}{4} + 2 \max_{i}\{\tilde{c}_i\} \right]$$
Now again $B$ can reach $t$ via red edges at $b_5$; the final outcome hence yields a cost of $1$ for $A$ and $0$ for $B$.

\textit{Case 1.2.3:} $A$ chooses a cycle corresponding to a $c_i$. 
But now $B$ can choose the cycle of length $\frac{3M}{4}$. 
When $B$ reaches $b_4$ after traversing this cycle, $A$ hast at most  $\frac{M}{4} + 2 \max_{i}\{\tilde{c}_i\}$ edges left on its cycle.
Hence $B$ can again proceed to $t$ via red edges at $b_5$, resulting in a total cost of $1$ for $A$ and $0$ for $B$.

\textit{Case 2:} 
$A$ answers with three cycles containing at least one cycle of length $\frac{3M}{4}$
such that at $r_1$ they are not synchronized any more. 
Recall that at $r_0$ both players were synchronized and $B$ chose a cycle of length $K$.
Now there are five possibilities for $A$ in Case 2:
(1) $A$ chooses one or both cycles of length $\frac{3M}{4}$ followed by cycles corresponding to $c_i$. 
(2) $A$ answers with one cycle corresponding to a $c_i$, one cycle of length $\frac{3M}{4}$ followed by two cycles corresponding to two $c_i$. 
(3) $A$ answers with one cycle corresponding to a $c_i$, both cycles of length $\frac{3M}{4}$ followed by one cycle corresponding to $c_i$. 
(4) $A$ answers with two cycles corresponding to two $c_i$, one cycle of length $\frac{3M}{4}$ followed by one cycle corresponding to $c_i$. 
(5) $A$ answers with two cycles corresponding to two $c_i$ and two cycles of length $\frac{3M}{4}$. 

In any of these cases arguments similar to the arguments of \textit{Case 1.2} show that $B$ can feasibly reach $t$ without using the cycle of length $3M - 2$ with cost $2$. Hence, in this case for \opath\ we have a cost of $1$ for A and $0$ for $B$.~\qed

\begin{sidewaysfigure}
\begin{center}
\scalebox{1.2} 
{
\begin{pspicture}(0,-4.718111)(17.574814,4.716)
\definecolor{color491}{rgb}{0.00392156862745098,0.00392156862745098,0.00392156862745098}
\psarc[linewidth=0.024,linecolor=red,arrowsize=0.05291667cm 2.0,arrowlength=1.4,arrowinset=0.4]{<-}(12.222813,-1.656){0.3}{328.90167}{159.44395}
\psdots[dotsize=0.08,linecolor=red](12.462812,-1.816)
\psline[linewidth=0.024cm,linecolor=red,arrowsize=0.05291667cm 2.0,arrowlength=1.4,arrowinset=0.4]{->}(12.448702,-1.824)(11.948702,-1.544)
\psarc[linewidth=0.024,linecolor=red,arrowsize=0.05291667cm 2.0,arrowlength=1.4,arrowinset=0.4]{<-}(12.762813,-1.956){0.3}{328.90167}{159.44395}
\psline[linewidth=0.024cm,linecolor=red,arrowsize=0.05291667cm 2.0,arrowlength=1.4,arrowinset=0.4]{->}(12.988702,-2.124)(12.488702,-1.844)
\psline[linewidth=0.024cm,linecolor=red,arrowsize=0.05291667cm 2.0,arrowlength=1.4,arrowinset=0.4]{->}(2.2628126,-2.576)(2.2628126,-3.376)
\psline[linewidth=0.024cm,linecolor=red,arrowsize=0.05291667cm 2.0,arrowlength=1.4,arrowinset=0.4]{->}(2.8628125,-2.976)(2.2628126,-2.576)
\psdots[dotsize=0.08,linecolor=red](2.2628126,-3.376)
\psdots[dotsize=0.08,linecolor=red](2.6628125,-3.976)
\psdots[dotsize=0.08,linecolor=red](3.6628125,-4.176)
\psdots[dotsize=0.08,linecolor=red](3.4628124,-3.576)
\psdots[dotsize=0.08,linecolor=red](2.8628125,-2.976)
\psline[linewidth=0.024cm,linecolor=red,arrowsize=0.05291667cm 2.0,arrowlength=1.4,arrowinset=0.4]{->}(2.6628125,-3.976)(3.6628125,-4.176)
\psline[linewidth=0.024cm,linecolor=red,arrowsize=0.05291667cm 2.0,arrowlength=1.4,arrowinset=0.4]{->}(3.6628125,-4.176)(3.4628124,-3.576)
\psline[linewidth=0.024cm,linecolor=red,linestyle=dashed,dash=0.16cm 0.16cm](2.2628126,-3.376)(2.6628125,-3.976)
\psline[linewidth=0.024cm,linecolor=red,linestyle=dashed,dash=0.16cm 0.16cm](3.4628124,-3.576)(2.8628125,-2.976)
\usefont{T1}{ppl}{m}{n}
\rput(2.744531,-3.546){$\frac{3M}{4}$}
\psline[linewidth=0.024cm,linecolor=red,arrowsize=0.05291667cm 2.0,arrowlength=1.4,arrowinset=0.4]{->}(12.462812,-0.176)(12.462812,-0.976)
\psline[linewidth=0.024cm,linecolor=red,arrowsize=0.05291667cm 2.0,arrowlength=1.4,arrowinset=0.4]{->}(12.862813,-0.776)(12.462812,-0.176)
\psdots[dotsize=0.08,linecolor=red](12.062813,0.624)
\psline[linewidth=0.024cm,linecolor=red,arrowsize=0.05291667cm 2.0,arrowlength=1.4,arrowinset=0.4]{->}(12.462812,-0.176)(12.062813,0.624)
\psline[linewidth=0.024cm,linecolor=red,arrowsize=0.05291667cm 2.0,arrowlength=1.4,arrowinset=0.4]{->}(11.462812,0.224)(12.462812,-0.176)
\psdots[dotsize=0.08](7.8628125,-2.576)
\psdots[dotsize=0.08](9.062813,-2.576)
\psdots[dotsize=0.08](10.262813,-2.576)
\psdots[dotsize=0.08](11.462812,-2.576)
\psdots[dotsize=0.08](7.6628127,-1.576)
\psdots[dotsize=0.08](7.4628124,-0.576)
\psdots[dotsize=0.08](7.2628126,0.424)
\psdots[dotsize=0.08](8.862813,-1.776)
\psdots[dotsize=0.08](8.462812,-0.976)
\psdots[dotsize=0.08](7.8628125,-0.176)
\psline[linewidth=0.024cm,arrowsize=0.05291667cm 2.0,arrowlength=1.4,arrowinset=0.4]{->}(7.8628125,-2.576)(9.062813,-2.576)
\psline[linewidth=0.024cm,arrowsize=0.05291667cm 2.0,arrowlength=1.4,arrowinset=0.4]{->}(9.062813,-2.576)(10.262813,-2.576)
\psline[linewidth=0.024cm,arrowsize=0.05291667cm 2.0,arrowlength=1.4,arrowinset=0.4]{->}(10.262813,-2.576)(11.462812,-2.576)
\psline[linewidth=0.024cm,arrowsize=0.05291667cm 2.0,arrowlength=1.4,arrowinset=0.4]{->}(7.8628125,-2.576)(7.6628127,-1.576)
\psline[linewidth=0.024cm,arrowsize=0.05291667cm 2.0,arrowlength=1.4,arrowinset=0.4]{->}(7.6628127,-1.576)(7.4628124,-0.576)
\psline[linewidth=0.024cm,arrowsize=0.05291667cm 2.0,arrowlength=1.4,arrowinset=0.4]{->}(7.4628124,-0.576)(7.2628126,0.424)
\psline[linewidth=0.024cm,arrowsize=0.05291667cm 2.0,arrowlength=1.4,arrowinset=0.4]{->}(9.062813,-2.576)(8.862813,-1.776)
\psline[linewidth=0.024cm,arrowsize=0.05291667cm 2.0,arrowlength=1.4,arrowinset=0.4]{->}(8.862813,-1.776)(8.462812,-0.976)
\psline[linewidth=0.024cm,arrowsize=0.05291667cm 2.0,arrowlength=1.4,arrowinset=0.4]{->}(8.462812,-0.976)(7.8628125,-0.176)
\psline[linewidth=0.024cm,arrowsize=0.05291667cm 2.0,arrowlength=1.4,arrowinset=0.4]{->}(7.8628125,-0.176)(7.2628126,0.424)
\psdots[dotsize=0.08](11.462812,-1.776)
\psdots[dotsize=0.08](12.462812,-0.176)
\psdots[dotsize=0.08](11.662812,-0.976)
\psline[linewidth=0.024cm,arrowsize=0.05291667cm 2.0,arrowlength=1.4,arrowinset=0.4]{->}(11.462812,-2.576)(11.462812,-1.776)
\psline[linewidth=0.024cm,arrowsize=0.05291667cm 2.0,arrowlength=1.4,arrowinset=0.4]{->}(11.462812,-1.776)(11.662812,-0.976)
\psline[linewidth=0.024cm,arrowsize=0.05291667cm 2.0,arrowlength=1.4,arrowinset=0.4]{->}(11.662812,-0.976)(12.462812,-0.176)
\psline[linewidth=0.024cm,arrowsize=0.05291667cm 2.0,arrowlength=1.4,arrowinset=0.4]{->}(12.482813,-0.156)(11.928701,-1.524)
\psarc[linewidth=0.024,arrowsize=0.05291667cm 2.0,arrowlength=1.4,arrowinset=0.4]{<-}(10.265758,-2.2389445){1.2429445}{195.42216}{-14.534455}
\psdots[dotsize=0.08](12.862813,0.424)
\psdots[dotsize=0.08](13.462812,1.424)
\psdots[dotsize=0.08](14.062813,3.624)
\psdots[dotsize=0.08](13.662812,4.024)
\psdots[dotsize=0.08](13.062813,3.824)
\psdots[dotsize=0.08](12.462812,1.424)
\psdots[dotsize=0.08](12.462812,0.624)
\psline[linewidth=0.024cm,arrowsize=0.05291667cm 2.0,arrowlength=1.4,arrowinset=0.4]{->}(12.462812,-0.176)(12.862813,0.424)
\psline[linewidth=0.024cm,arrowsize=0.05291667cm 2.0,arrowlength=1.4,arrowinset=0.4]{->}(12.862813,0.424)(13.462812,1.424)
\psline[linewidth=0.024cm,arrowsize=0.05291667cm 2.0,arrowlength=1.4,arrowinset=0.4]{->}(14.062813,3.624)(13.662812,4.024)
\psline[linewidth=0.024cm,arrowsize=0.05291667cm 2.0,arrowlength=1.4,arrowinset=0.4]{->}(13.662812,4.024)(13.062813,3.824)
\psline[linewidth=0.024cm,arrowsize=0.05291667cm 2.0,arrowlength=1.4,arrowinset=0.4]{->}(12.462812,1.424)(12.462812,0.624)
\psline[linewidth=0.024cm,arrowsize=0.05291667cm 2.0,arrowlength=1.4,arrowinset=0.4]{->}(12.462812,0.624)(12.462812,-0.176)
\usefont{T1}{ppl}{m}{n}
\rput(13.284533,3.514){$K$}
\pscustom[linewidth=0.024,linestyle=dashed,dash=0.16cm 0.16cm]
{
\newpath
\moveto(12.462815,-0.17600012)
\lineto(12.962815,-0.076000124)
\curveto(13.212815,-0.026000122)(13.612814,0.12399988)(13.762813,0.22399987)
\curveto(13.912811,0.32399988)(14.112811,0.6239999)(14.162812,0.8239999)
\curveto(14.212813,1.0239999)(14.312814,1.4239999)(14.362812,1.6239998)
\curveto(14.412811,1.8239999)(14.462812,2.274)(14.462815,2.524)
\curveto(14.462817,2.7740002)(14.462817,3.1740003)(14.462815,3.3240001)
\curveto(14.462812,3.474)(14.562811,3.774)(14.662812,3.924)
\curveto(14.762814,4.074)(15.012814,4.124)(15.162812,4.024)
\curveto(15.312811,3.924)(15.462812,3.5740001)(15.462815,3.3240001)
\curveto(15.462817,3.0740001)(15.412815,2.5740001)(15.362812,2.324)
\curveto(15.31281,2.074)(15.2128105,1.6239998)(15.162812,1.4239999)
\curveto(15.112815,1.2239999)(14.962815,0.8239999)(14.862812,0.6239999)
\curveto(14.76281,0.42399988)(14.51281,0.12399988)(14.362812,0.02399988)
\curveto(14.212815,-0.076000124)(13.862815,-0.22600012)(13.662812,-0.2760001)
\curveto(13.4628105,-0.32600012)(13.06281,-0.32600012)(12.862812,-0.2760001)
}
\pscustom[linewidth=0.024,linestyle=dashed,dash=0.16cm 0.16cm]
{
\newpath
\moveto(12.462815,-0.17600012)
\lineto(12.862812,-0.2760001)
\curveto(13.062811,-0.32600012)(13.462811,-0.42600012)(13.662812,-0.47600013)
\curveto(13.862814,-0.52600014)(14.312814,-0.5760001)(14.562813,-0.5760001)
\curveto(14.812811,-0.5760001)(15.212811,-0.37600014)(15.362812,-0.17600012)
\curveto(15.512814,0.02399988)(15.762814,0.37399986)(15.862812,0.52399987)
\curveto(15.962811,0.6739999)(16.112812,1.0239999)(16.162813,1.2239999)
\curveto(16.212814,1.4239999)(16.262814,1.8739998)(16.262812,2.1239998)
\curveto(16.262812,2.3739998)(16.362812,2.824)(16.462812,3.024)
\curveto(16.562813,3.224)(16.762814,3.5740001)(16.862812,3.724)
\curveto(16.96281,3.8739998)(17.21281,3.8739998)(17.362812,3.724)
\curveto(17.512814,3.5740001)(17.562813,3.1739998)(17.462812,2.9239998)
\curveto(17.362812,2.6739995)(17.16281,2.2739997)(17.06281,2.1239998)
\curveto(16.962812,1.9739999)(16.812813,1.6239998)(16.762812,1.4239999)
\curveto(16.712812,1.2239999)(16.612812,0.77399987)(16.56281,0.52399987)
\curveto(16.51281,0.27399987)(16.362812,-0.12600012)(16.262812,-0.2760001)
\curveto(16.162813,-0.42600012)(15.862814,-0.6260001)(15.662812,-0.6760001)
\curveto(15.462811,-0.72600013)(15.012812,-0.8259995)(14.762813,-0.8759995)
\curveto(14.512814,-0.9259995)(14.012814,-0.9259995)(13.762813,-0.8759995)
\curveto(13.512812,-0.8259995)(13.112812,-0.6260007)(12.962815,-0.47600073)
\curveto(12.812818,-0.32600072)(12.562818,-0.17600073)(12.462815,-0.17600012)
}
\psline[linewidth=0.024cm,linestyle=dotted,dotsep=0.16cm](15.482813,2.764)(16.282812,2.764)
\usefont{T1}{ppl}{m}{n}
\rput(14.844531,3.514){$K$}
\usefont{T1}{ppl}{m}{n}
\rput(16.94453,3.334){$K$}
\usefont{T1}{ppl}{m}{n}
\rput(15.082812,4.474){$(n-1)$ times}
\psdots[dotsize=0.08,linecolor=red](11.662812,1.424)
\psdots[dotsize=0.08,linecolor=red](11.062813,2.024)
\psdots[dotsize=0.08,linecolor=red](10.062813,2.424)
\psdots[dotsize=0.08,linecolor=red](9.862813,1.624)
\psdots[dotsize=0.08,linecolor=red](10.262813,0.824)
\psdots[dotsize=0.08,linecolor=red](11.462812,0.224)
\psline[linewidth=0.024cm,linecolor=red,arrowsize=0.05291667cm 2.0,arrowlength=1.4,arrowinset=0.4]{->}(12.062813,0.624)(11.662812,1.424)
\psline[linewidth=0.024cm,linecolor=red,arrowsize=0.05291667cm 2.0,arrowlength=1.4,arrowinset=0.4]{->}(11.062813,2.024)(10.062813,2.424)
\psline[linewidth=0.024cm,linecolor=red,arrowsize=0.05291667cm 2.0,arrowlength=1.4,arrowinset=0.4]{->}(10.062813,2.424)(9.862813,1.624)
\psline[linewidth=0.024cm,linecolor=red,arrowsize=0.05291667cm 2.0,arrowlength=1.4,arrowinset=0.4]{->}(9.862813,1.624)(10.262813,0.824)
\psline[linewidth=0.024cm,linecolor=red,linestyle=dashed,dash=0.16cm 0.16cm](11.662812,1.424)(11.062813,2.024)
\psline[linewidth=0.024cm,linecolor=red,linestyle=dashed,dash=0.16cm 0.16cm](10.262813,0.824)(11.462812,0.224)
\usefont{T1}{ppl}{m}{n}
\rput(10.474531,1.654){$3M-2$}
\usefont{T1}{ppl}{m}{n}
\rput(12.034532,0.274){\color{red}$2$}
\psdots[dotsize=0.08,linecolor=red](12.462812,-0.976)
\psdots[dotsize=0.08,linecolor=red](12.862813,-1.576)
\psdots[dotsize=0.08,linecolor=red](13.862813,-1.776)
\psdots[dotsize=0.08,linecolor=red](13.662812,-1.176)
\psdots[dotsize=0.08,linecolor=red](12.862813,-0.776)
\psline[linewidth=0.024cm,linecolor=red,arrowsize=0.05291667cm 2.0,arrowlength=1.4,arrowinset=0.4]{->}(12.862813,-1.576)(13.862813,-1.776)
\psline[linewidth=0.024cm,linecolor=red,arrowsize=0.05291667cm 2.0,arrowlength=1.4,arrowinset=0.4]{->}(13.862813,-1.776)(13.662812,-1.176)
\psline[linewidth=0.024cm,linecolor=red,linestyle=dashed,dash=0.16cm 0.16cm](12.462812,-0.976)(12.862813,-1.576)
\psline[linewidth=0.024cm,linecolor=red,linestyle=dashed,dash=0.16cm 0.16cm](13.662812,-1.176)(12.862813,-0.776)
\usefont{T1}{ppl}{m}{n}
\rput(13.004531,-1.306){$\frac{3M}{4}$}
\psdots[dotsize=0.08,linecolor=color491](12.862813,3.024)
\psdots[dotsize=0.08,linecolor=color491](13.862813,2.824)
\psline[linewidth=0.024cm,linecolor=color491,arrowsize=0.05291667cm 2.0,arrowlength=1.4,arrowinset=0.4]{->}(13.862813,2.824)(14.062813,3.624)
\psline[linewidth=0.024cm,linecolor=color491,arrowsize=0.05291667cm 2.0,arrowlength=1.4,arrowinset=0.4]{->}(13.062813,3.824)(12.862813,3.024)
\psline[linewidth=0.024cm,linecolor=color491,linestyle=dashed,dash=0.16cm 0.16cm](12.862813,3.024)(12.462812,1.424)
\psline[linewidth=0.024cm,linecolor=color491,linestyle=dashed,dash=0.16cm 0.16cm](13.862813,2.824)(13.462812,1.424)
\psline[linewidth=0.024cm,linecolor=red,linestyle=dashed,dash=0.16cm 0.16cm](13.542812,-2.416)(15.008701,-3.164)
\usefont{T1}{ppl}{m}{n}
\rput(12.834533,-3.006){$\frac{M}{4} + \frac{5}{2} \max_i \{\tilde{c}_i\}$ }
\usefont{T1}{ppl}{m}{n}
\rput(7.776875,-2.841){\large $B$ }
\psdots[dotsize=0.08,linecolor=color491](5.8628125,-2.576)
\psdots[dotsize=0.08,linecolor=color491](5.2628126,-2.576)
\psdots[dotsize=0.08,linecolor=color491](4.6628127,-2.576)
\psdots[dotsize=0.08,linecolor=color491](4.0628123,-2.576)
\psdots[dotsize=0.08,linecolor=color491](3.5628126,-2.576)
\psdots[dotsize=0.08,linecolor=color491](2.6428125,-2.576)
\psdots[dotsize=0.08,linecolor=color491](2.2628126,-2.576)
\psdots[dotsize=0.08,linecolor=color491](5.8628125,-0.976)
\psline[linewidth=0.024cm,linecolor=color491,arrowsize=0.05291667cm 2.0,arrowlength=1.4,arrowinset=0.4]{->}(5.8628125,-2.576)(5.2628126,-2.576)
\psline[linewidth=0.024cm,linecolor=color491,arrowsize=0.05291667cm 2.0,arrowlength=1.4,arrowinset=0.4]{->}(5.2628126,-2.576)(4.6628127,-2.576)
\psline[linewidth=0.024cm,linecolor=color491,arrowsize=0.05291667cm 2.0,arrowlength=1.4,arrowinset=0.4]{->}(4.6628127,-2.576)(4.0628123,-2.576)
\psline[linewidth=0.024cm,linecolor=color491,arrowsize=0.05291667cm 2.0,arrowlength=1.4,arrowinset=0.4]{->}(5.2628126,-2.576)(5.8628125,-0.976)
\psline[linewidth=0.024cm,linecolor=color491,arrowsize=0.05291667cm 2.0,arrowlength=1.4,arrowinset=0.4]{->}(5.8628125,-0.976)(7.2628126,0.424)
\psarc[linewidth=0.024,linecolor=color491,arrowsize=0.05291667cm 2.0,arrowlength=1.4,arrowinset=0.4]{->}(4.6628127,-2.596){0.6}{-180.0}{0.0}
\psarc[linewidth=0.024,linecolor=color491,arrowsize=0.05291667cm 2.0,arrowlength=1.4,arrowinset=0.4]{<-}(3.1628125,-2.876){0.96}{20.854458}{159.14554}
\psdots[dotsize=0.08,linecolor=color491](1.8628125,-2.176)
\psdots[dotsize=0.08,linecolor=color491](1.4628125,-1.776)
\psdots[dotsize=0.08,linecolor=color491](0.66281253,-1.176)
\psdots[dotsize=0.08,linecolor=color491](0.06281254,-0.976)
\psdots[dotsize=0.08,linecolor=color491](0.26281255,-1.576)
\psdots[dotsize=0.08,linecolor=color491](1.0628126,-2.176)
\psdots[dotsize=0.08,linecolor=color491](1.4628125,-2.376)
\psline[linewidth=0.024cm,linecolor=color491,arrowsize=0.05291667cm 2.0,arrowlength=1.4,arrowinset=0.4]{->}(2.2628126,-2.576)(1.8628125,-2.176)
\psline[linewidth=0.024cm,linecolor=color491,arrowsize=0.05291667cm 2.0,arrowlength=1.4,arrowinset=0.4]{->}(1.8628125,-2.176)(1.4628125,-1.776)
\psline[linewidth=0.024cm,linecolor=color491,arrowsize=0.05291667cm 2.0,arrowlength=1.4,arrowinset=0.4]{->}(0.66281253,-1.176)(0.06281254,-0.976)
\psline[linewidth=0.024cm,linecolor=color491,arrowsize=0.05291667cm 2.0,arrowlength=1.4,arrowinset=0.4]{->}(0.06281254,-0.976)(0.26281255,-1.576)
\psline[linewidth=0.024cm,linecolor=color491,arrowsize=0.05291667cm 2.0,arrowlength=1.4,arrowinset=0.4]{->}(1.0628126,-2.176)(1.4628125,-2.376)
\psline[linewidth=0.024cm,linecolor=color491,arrowsize=0.05291667cm 2.0,arrowlength=1.4,arrowinset=0.4]{->}(1.4628125,-2.376)(2.2628126,-2.576)
\psline[linewidth=0.024cm,linecolor=color491,linestyle=dashed,dash=0.16cm 0.16cm](0.66281253,-1.176)(1.4628125,-1.776)
\psline[linewidth=0.024cm,linecolor=color491,linestyle=dashed,dash=0.16cm 0.16cm](0.26281255,-1.576)(1.0628126,-2.176)
\usefont{T1}{ppl}{m}{n}
\rput(0.4945313,-1.426){$c_i$}
\pscustom[linewidth=0.024,linecolor=color491,linestyle=dashed,dash=0.16cm 0.16cm]
{
\newpath
\moveto(2.2628126,-2.5760002)
\lineto(2.3628125,-2.1759994)
\curveto(2.4128125,-1.9759995)(2.5128126,-1.5759995)(2.5628126,-1.3759996)
\curveto(2.6128125,-1.1760001)(2.6628125,-0.77600014)(2.6628125,-0.5759995)
\curveto(2.6628125,-0.3759989)(2.5128126,-0.07599951)(2.3628125,0.02399988)
\curveto(2.2128124,0.12399896)(2.0128126,0.023998963)(1.9628125,-0.17600012)
\curveto(1.9128125,-0.3759995)(1.8628125,-0.8259995)(1.8628125,-1.0760007)
\curveto(1.8628125,-1.326002)(1.9128125,-1.7760019)(1.9628125,-1.9760001)
\curveto(2.0128126,-2.1759982)(2.1628125,-2.4759989)(2.2628126,-2.5760007)
}
\usefont{T1}{ppl}{m}{n}
\rput(2.134531,-0.466){$c_1$}
\usefont{T1}{ppl}{m}{n}
\rput(0.78453124,-3.126){$c_{3n}$}
\usefont{T1}{ppl}{m}{n}
\rput(2.8345313,-1.766){\color{red}$1$}
\pscustom[linewidth=0.024,linecolor=color491,linestyle=dashed,dash=0.16cm 0.16cm]
{
\newpath
\moveto(2.2628126,-2.5760002)
\lineto(1.8628125,-2.7093332)
\curveto(1.6628125,-2.7759995)(1.2128125,-2.842667)(0.9628125,-2.842667)
\curveto(0.7128125,-2.842667)(0.4128125,-2.9760008)(0.36281243,-3.1093333)
\curveto(0.31281236,-3.2426658)(0.5128124,-3.3759995)(0.7628125,-3.3759995)
\curveto(1.0128126,-3.3759995)(1.4628127,-3.3093333)(1.6628125,-3.2426658)
\curveto(1.8628124,-3.1759982)(2.1628125,-2.8426657)(2.2628126,-2.5760002)
}
\psline[linewidth=0.024cm,linecolor=color491,linestyle=dotted,dotsep=0.16cm](0.9228125,-1.216)(1.7628125,-0.616)
\psline[linewidth=0.024cm,linecolor=color491,linestyle=dotted,dotsep=0.16cm](0.6028125,-1.896)(0.70281255,-2.796)
\psline[linewidth=0.024cm,linecolor=color491,linestyle=dotted,dotsep=0.16cm](16.282812,4.704)(16.202812,4.584)
\usefont{T1}{ppl}{m}{n}
\rput(7.6345315,-2.066){\color{red}$1$}
\usefont{T1}{ppl}{m}{n}
\rput(5.896875,-2.881){\large $A$ }
\usefont{T1}{ptm}{m}{n}
\rput(8.014532,-2.386){$b_1$}
\usefont{T1}{ptm}{m}{n}
\rput(9.214532,-2.366){$b_2$}
\usefont{T1}{ptm}{m}{n}
\rput(11.1145315,-2.286){$b_3$}
\usefont{T1}{ptm}{m}{n}
\rput(11.87453,-0.226){$b_4$}
\usefont{T1}{ptm}{m}{n}
\rput(7.0745306,0.674){$t$}
\usefont{T1}{ptm}{m}{n}
\rput(5.5045314,-2.426){$a_2$}
\usefont{T1}{ptm}{m}{n}
\rput(4.2045317,-2.426){$a_3$}
\usefont{T1}{ptm}{m}{n}
\rput(2.564531,-2.426){$a_4$}
\psline[linewidth=0.024cm,arrowsize=0.05291667cm 2.0,arrowlength=1.4,arrowinset=0.4]{->}(3.1828125,-1.916)(3.2487016,-1.926)
\psline[linewidth=0.024cm,arrowsize=0.05291667cm 2.0,arrowlength=1.4,arrowinset=0.4]{->}(4.6487017,-3.196)(4.6887016,-3.204)
\psdots[dotsize=0.08](11.922812,-1.536)
\psline[linewidth=0.024cm,arrowsize=0.05291667cm 2.0,arrowlength=1.4,arrowinset=0.4]{->}(11.928701,-1.524)(11.468701,-2.564)
\psline[linewidth=0.024cm,arrowsize=0.05291667cm 2.0,arrowlength=1.4,arrowinset=0.4]{->}(10.248701,-3.484)(10.208701,-3.484)
\psdots[dotsize=0.08,linecolor=red](13.002812,-2.116)
\psarc[linewidth=0.024,linecolor=red,arrowsize=0.05291667cm 2.0,arrowlength=1.4,arrowinset=0.4]{<-}(13.302813,-2.236){0.3}{328.90167}{159.44395}
\psdots[dotsize=0.08,linecolor=red](13.542812,-2.396)
\psline[linewidth=0.024cm,linecolor=red,arrowsize=0.05291667cm 2.0,arrowlength=1.4,arrowinset=0.4]{->}(13.528702,-2.404)(13.028702,-2.124)
\psarc[linewidth=0.024,linecolor=red,arrowsize=0.05291667cm 2.0,arrowlength=1.4,arrowinset=0.4]{<-}(15.302813,-3.296){0.3}{328.90167}{159.44395}
\psdots[dotsize=0.08,linecolor=red](15.542812,-3.456)
\psdots[dotsize=0.08,linecolor=red](15.022813,-3.176)
\psline[linewidth=0.024cm,linecolor=red,arrowsize=0.05291667cm 2.0,arrowlength=1.4,arrowinset=0.4]{->}(15.528702,-3.464)(15.028702,-3.184)
\usefont{T1}{ptm}{m}{n}
\rput(11.73453,-1.386){$b_5$}
\psline[linewidth=0.024cm,linecolor=red,arrowsize=0.05291667cm 2.0,arrowlength=1.4,arrowinset=0.4]{->}(13.856405,-1.764)(14.5764065,-2.064)
\psdots[dotsize=0.08,linecolor=red](14.5764065,-2.044)
\psline[linewidth=0.024cm,linecolor=red,linestyle=dashed,dash=0.16cm 0.16cm](14.5764065,-2.044)(15.356405,-2.244)
\psdots[dotsize=0.08,linecolor=red](15.336407,-2.244)
\psdots[dotsize=0.08,linecolor=red](15.976406,-1.764)
\psdots[dotsize=0.08,linecolor=red](15.396405,-1.144)
\psdots[dotsize=0.08,linecolor=red](14.556407,-1.044)
\psline[linewidth=0.024cm,linecolor=red,linestyle=dashed,dash=0.16cm 0.16cm](15.396405,-1.144)(14.556407,-1.044)
\psline[linewidth=0.024cm,linecolor=red,arrowsize=0.05291667cm 2.0,arrowlength=1.4,arrowinset=0.4]{->}(15.336407,-2.244)(15.996407,-1.784)
\psline[linewidth=0.024cm,linecolor=red,arrowsize=0.05291667cm 2.0,arrowlength=1.4,arrowinset=0.4]{->}(15.956407,-1.764)(15.396405,-1.144)
\psline[linewidth=0.024cm,linecolor=red,arrowsize=0.05291667cm 2.0,arrowlength=1.4,arrowinset=0.4]{->}(14.556407,-1.044)(13.876407,-1.764)
\usefont{T1}{ppl}{m}{n}
\rput(14.904531,-1.606){$\frac{3M}{4}$}
\psline[linewidth=0.024cm,linecolor=red,arrowsize=0.05291667cm 2.0,arrowlength=1.4,arrowinset=0.4]{->}(3.6764054,-4.184)(4.3964067,-4.484)
\psdots[dotsize=0.08,linecolor=red](4.3964067,-4.464)
\psline[linewidth=0.024cm,linecolor=red,linestyle=dashed,dash=0.16cm 0.16cm](4.3964067,-4.464)(5.1764054,-4.664)
\psdots[dotsize=0.08,linecolor=red](5.1564064,-4.664)
\psdots[dotsize=0.08,linecolor=red](5.7964067,-4.184)
\psdots[dotsize=0.08,linecolor=red](5.2164054,-3.564)
\psdots[dotsize=0.08,linecolor=red](4.3764067,-3.464)
\psline[linewidth=0.024cm,linecolor=red,linestyle=dashed,dash=0.16cm 0.16cm](5.2164054,-3.564)(4.3764067,-3.464)
\psline[linewidth=0.024cm,linecolor=red,arrowsize=0.05291667cm 2.0,arrowlength=1.4,arrowinset=0.4]{->}(5.1564064,-4.664)(5.8164067,-4.204)
\psline[linewidth=0.024cm,linecolor=red,arrowsize=0.05291667cm 2.0,arrowlength=1.4,arrowinset=0.4]{->}(5.776407,-4.184)(5.2164054,-3.564)
\psline[linewidth=0.024cm,linecolor=red,arrowsize=0.05291667cm 2.0,arrowlength=1.4,arrowinset=0.4]{->}(4.3764067,-3.464)(3.6964066,-4.184)
\usefont{T1}{ppl}{m}{n}
\rput(4.7245297,-4.026){$\frac{3M}{4}$}
\usefont{T1}{ptm}{m}{n}
\rput(5.9845314,-2.446){$a_1$}
\psdots[dotsize=0.08,linecolor=color491](3.1228125,-2.576)
\psline[linewidth=0.024cm,linecolor=color491,arrowsize=0.05291667cm 2.0,arrowlength=1.4,arrowinset=0.4]{->}(4.08,-2.584)(3.56,-2.584)
\psline[linewidth=0.024cm,linecolor=color491,arrowsize=0.05291667cm 2.0,arrowlength=1.4,arrowinset=0.4]{->}(3.56,-2.584)(3.12,-2.584)
\psline[linewidth=0.024cm,linecolor=color491,arrowsize=0.05291667cm 2.0,arrowlength=1.4,arrowinset=0.4]{->}(3.12,-2.584)(2.64,-2.584)
\psline[linewidth=0.024cm,linecolor=color491,arrowsize=0.05291667cm 2.0,arrowlength=1.4,arrowinset=0.4]{->}(2.64,-2.584)(2.26,-2.584)
\end{pspicture} 
}

\label{fig:cac}
\caption{The Cactus Graph $G_{3P}$ }
\end{center}
\end{sidewaysfigure}
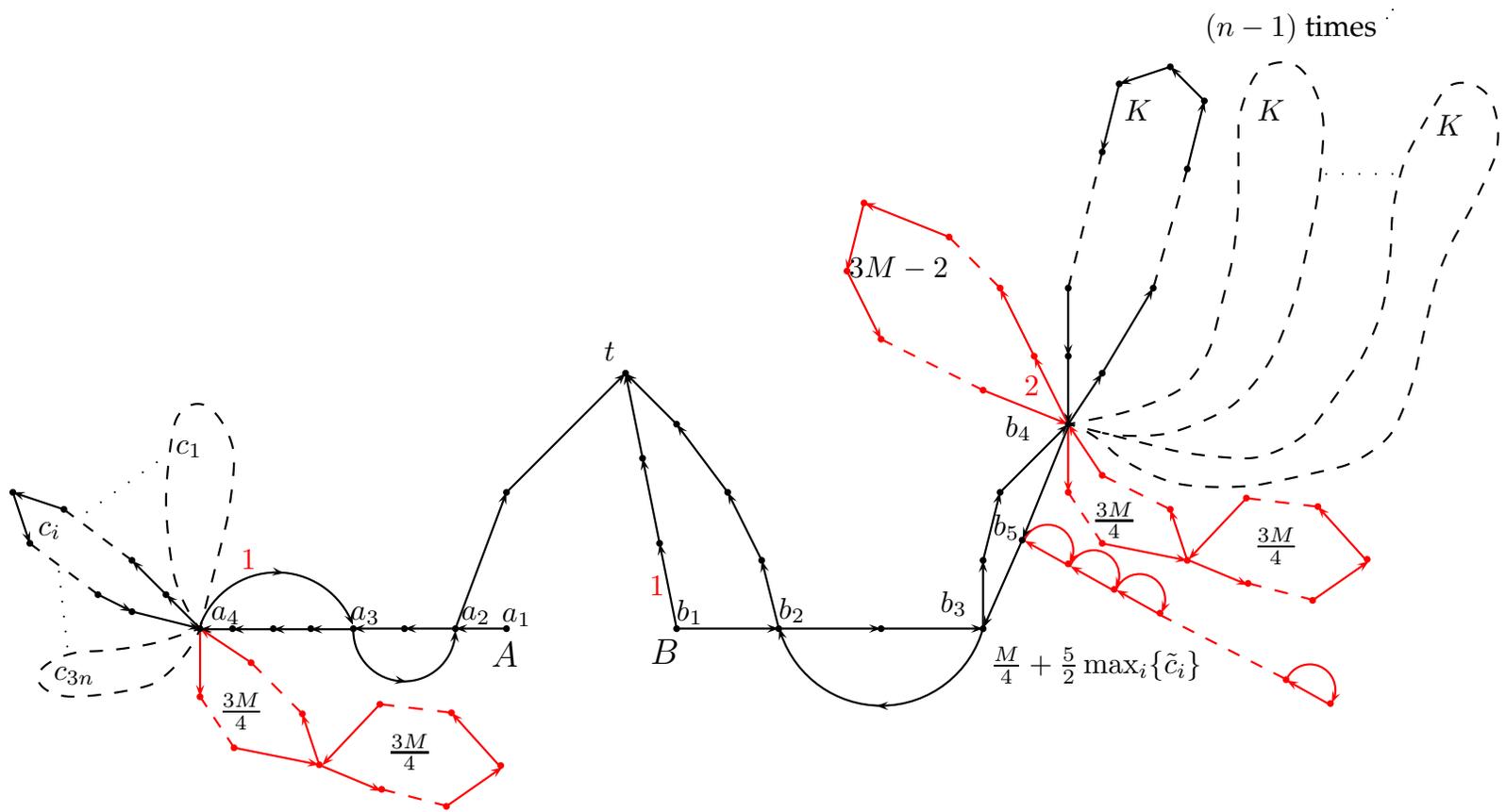

\subsection{A Polynomial Algorithm for \scg\ restricted to Simple Paths on Cactus Graphs}
\label{sec:cactussimple}

It turns out that in contrast to the negative result of Theorem~\ref{th:hardcactus},
imposing (R3) and thus eliminating cycles
allows a polynomial time algorithm for \scg\ on a cactus graph.
It is based on a dynamic programming approach which mimics the backward induction 
process.
However, because of the strong structural property of a cactus graph,
we can reduce the feasible domain for one player, say $A$,
to an acyclic directed graph, which allows a linear time backtracking process.
For each vertex, which $A$ has currently reached,
we consider all (that is $n$) possible vertices as current positions of $B$
and determine recursively the ``best'' (in the \opath\ sense) paths for both players
to a potential meeting point.
The cactus property allows to reduce the options for these meeting points.

To construct the acyclic digraph for $A$ we first introduce 
an auxiliary edge set $E' \subseteq E$.
Initially, $E'$ contains all edges of $E$ that might be used by player $A$.
More formally, $(u,v)\in E$ is contained in $E'$ if there exists a directed path 
from $s$ to $u$.
Next, consider all directed cycles $C$ in $E'$.
By definition of a cactus, there can be only one path from $s$ entering $C$ at some vertex $v$:
Otherwise, a second path entering $C$ at $v'$ would imply that either the edges between $v$ and $v'$
or between $v'$ and $v$ are contained in two cycles.
Clearly, $A$ cannot use the ``last'' edge of cycle $C$, i.e.\ the edge $(u,v)$ entering $v$,
since $A$ already must have passed through $v$ to reach $C$.
Therefore, we eliminate $(u,v)$ from $E'$ which makes $E'$ acyclic and connected.
$E'$ will remain constant throughout the execution of the algorithm.

For every vertex $a \in V$ define $M(a) \subseteq V$ as the set of vertices which player $A$
{\em must} visit in order to reach $a$ from $s$, i.e.\ $M(a)$ consists of those vertices
that are contained in every directed path from $s$ to $a$.
$M(a)$ also contains $a$.

\begin{prop}\label{th:meeting}
If in \opath\ player $A$ moves from $s$ until $a\in V$, then the meeting point $m$ is contained in $M(a)$
w.l.o.g.
\end{prop}

\begin{proof}
Consider a possible meeting point $m \not\in M(a)$.
This means that there exists a cycle $C$ which $A$ has to traverse on the way from $s$ to $a$ and
whose edges are oriented in such a way that $A$ has two possible ways to pass through $C$,
one of them containing $m$.
There can exist at most one such cycle:
Otherwise, assume that there is a second cycle $C'$ with a potential meeting point $m'$.
Now, there is a path from $t$ to $m$ and from $t$ to $m'$.
Since there is also path between $m$ and $m'$,
one of the two edges in $C$ emanating from $m$ must be contained in two cycles.

Since only one such configuration can exist, we can exchange the roles of $A$ and $B$ 
(and add a dummy edge $(t',t)$ to preserve the role of the first mover).~\qed
\end{proof}

\def\costa{\mbox{Cost}_A}
\def\costb{\mbox{Cost}_B}
\def\costx{\mbox{Cost}}
\def\succ{\mbox{succ}}
\def\besta{\tilde{a}}

In the following we describe a dynamic programming scheme to determine \opath.
\begin{quote}
Assume that $A$ is currently in $a \in V$ and $B$ in $b \in V$ with $A$ making the next move.
We introduce arrays $\costa(a,b)$ resp.\ $\costb(a,b)$ for $A$ resp.\ $B$
containing the costs of 
the optimal subgame perfect equilibrium paths from the current vertices to the end of the game.
\end{quote}
If a configuration $(a,b)$ is found out to be infeasible 
(e.g.\ $a$ is a leaf of $E'$ and $b\neq a$),
then we set $\costa(a,b)=\costb(a,b)=\infty$.
Note that this does not concern the aspect whether $(a,b)$ can be reached from the
starting configuration $(s,t)$, but only whether an endpoint of the game 
could be reached from $(a,b)$.
The cost of \opath\ will finally be reported in $\costa(s,t)$ and $\costb(s,t)$.

The recursive computations of these arrays will be performed in a bottom up way
guided by an auxiliary edge set $F\subseteq E$ initialized by $F=E'$.
Recall that $F$ is acyclic.
Throughout the computation $F$ contains all edges leading to vertices
which were not yet considered as positions of $A$.
In the main computation we determine for each vertex $a \in V$
with outdegree $0$ w.r.t.\ $F$
(i.e.\ for vertices whose successors were all dealt with)
the entries $\costx(a,b)$ for all $b\in V$.
Note that each entry of $\costx(a, \cdot)$ is computed only once and never updated.
After completion of this task all incoming edges $(u,a)$ are eliminated from $F$
possibly generating new vertices with outdegree $0$ w.r.t.\ $F$.
This process is continued until $F=\emptyset$ and $\costx(s, b)$ is determined
for all $b$.

Processing a vertex $a$ with no successors w.r.t.\ $F$ works as follows:
For all $b \in M(a)$ player $B$ has reached a meeting point.
Hence, $\costa(a,b)=\costb(a,b)=0$. 
(This includes the case $a=b$.)
For all $b\ \not\in M(a)$, we consider all possible decisions for $A$ in $a$
and all possible reactions of $B$ in $b$.
Therefore, let 
$\succ(a)=\{v \in V \mid (a,v) \in E'\}$ and
$\succ(b)=\{v \in V \mid (b,v) \in E\}$.
If $\succ(a)=\emptyset$ then player $A$ ``got stuck'' and no feasible solution
can be reached. 
Thus, $\costa(a,b)=\costb(a,b)=\infty$. 
Otherwise, if $\succ(b)=\emptyset$ then player $B$ has no feasible move left.
If player $A$ can still reach $b$ by traversing one edge, a meeting point is reached.
Therefore, if $(a,b) \in E'$ then $\costa(a,b)=c(a,b)$, $\costb(a,b)=0$.
Otherwise, if $(a,b) \not\in E'$ then $\costa(a,b)=\costb(a,b)=\infty$.
There remains the general case where both 
$\succ(a)$ and $\succ(b)$ are non-empty.

For all possible decisions taken by $A$ in $a$, i.e.\ for all $a' \in \succ(a)$,
we determine the best reaction by $B$.
If $a'=b$ then no move is necessary for $B$ and we set $b(a')=b$. 
Otherwise, there is
\begin{equation}\label{eq:bestb}
b(a') = \arg\min_{b'\in \succ(b)} \{c(b,b') + \costb(a',b')\}.
\end{equation}

This implies the cost of $A$ resulting from choosing $a'$.
It remains to select the best decision $\besta$ for $A$:
\begin{equation}\label{eq:besta}
\besta = \arg\min_{a'\in \succ(a)} \{c(a,a') + \costa(a', b(a'))\}
\end{equation}
Finally, we set 
\begin{eqnarray*}
\costa(a,b) &=& c(a,\besta) + \costa(\besta, b(\besta))\\
\costb(a,b) &=& c(b,b(\besta)) + \costb(\besta, b(\besta)).
\end{eqnarray*}

\begin{theorem}\label{th:time}
Imposing (R3),  \opath\ of \scg\ on cactus graphs can be computed in $O(n^2)$ time.
\end{theorem}

\begin{proof}
Throughout the execution of the algorithm, the total number of vertices $a'$ considered  in (\ref{eq:besta}) 
as a successor of some $a$ is bounded by $|E'|$,
since each edge is used only once (although a vertex may well be considered multiple times).
For each candidate $a'$ we have to consider all possible combinations
of some vertex with a successor $b'$. 
The number of such pairs $(b, b')$ is trivially bounded by $|E|$.
Thus, the running time is bounded by $O(|E|^2)$.
The statement follows since it is well-known that the number of edges of a cactus graph is $O(n)$.
\qed\end{proof}

\section{Conclusion}

In this work, we have shown that \scg\ is computationally difficult not only 
on general graphs, which follows from our \psp ness-proof for bipartite graphs,
but still $\np$-hard even for more restricted graph classes such as directed acyclic graphs or cactus graphs. 
Since \scg\ is trivial on trees, this hardness result for cactus graphs establishes rather clearly the boundary between tractable and hard cases.

On the other hand, imposing the restriction that each vertex can be visited 
only once by each player,
i.e.\ a restriction to simple paths,
makes \scg\ polynomially solvable on cactus graphs.  
Under this restriction, we believe that it is a challenging problem to determine the computational complexity status of \scg\ and possibly find polynomial time algorithms for classical generalizations of cactus graphs such as, for instance, outerplanar graphs or series parallel graphs. 
At least for a dynamic programming based approach it seems to be necessary to have
some knowledge about possible meeting points on the way to a certain configuration,
as given by Proposition~\ref{th:meeting} for cactus graphs.
Unfortunately, even for outerplanar graphs no direct generalization of this statement 
could be found.
Thus, we believe that \scg\ might be computationally hard even for minor generalizations
of cactus graphs.

\bibliographystyle{plainnat}
\bibliography{path-game}

\end{document}